%% file: main.tex
\documentclass[english, 11pt]{article}
\usepackage[T1]{fontenc}
\usepackage{booktabs} 
\usepackage[utf8x]{inputenc}
\usepackage{geometry}
\usepackage{color,authblk}
\usepackage{babel}
\usepackage{float}
\usepackage{mathtools}
\usepackage{bm}
\usepackage{amsmath}
\usepackage{amsthm}
\usepackage{graphicx}
\usepackage[authoryear,round]{natbib}
\usepackage[unicode=true,
 bookmarks=false,
 breaklinks=false,pdfborder={0 0 1},backref=false,colorlinks=true]
 {hyperref}
\usepackage{natbib}
\usepackage{algorithm}
\usepackage{algpseudocode}
\makeatletter
\newtheorem{prop}{Proposition}
\theoremstyle{definition}
\newtheorem{assumption}{Assumption}
\providecommand{\tabularnewline}{\\}
\floatstyle{ruled}
\newfloat{algorithm}{tbp}{loa}
\providecommand{\algorithmname}{Algorithm}
\floatname{algorithm}{\protect\algorithmname}

\theoremstyle{plain}
\newtheorem{thm}{\protect\theoremname}
\theoremstyle{definition}
\newtheorem{lemma}{Lemma}
\newtheorem{defn}[thm]{\protect\definitionname}
\theoremstyle{remark}
\newtheorem{rem}[thm]{\protect\remarkname}
\theoremstyle{plain}
\newtheorem{cor}[thm]{\protect\corollaryname}

\usepackage{babel}
\usepackage{fancybox}
\usepackage{bm}
\usepackage{accents}

\theoremstyle{plain}
\usepackage{soul}
\usepackage{tikz}\usepackage{amsfonts}
\usepackage{graphicx,booktabs,multicol,longtable}
\usetikzlibrary{positioning}
\usepackage{algorithm}
\usepackage{algpseudocode}

\providecommand{\definitionname}{Definition}
\providecommand{\theoremname}{Theorem}

\providecommand{\remarkname}{Remark}

\providecommand{\corollaryname}{Corollary}

\makeatother

\providecommand{\corollaryname}{Corollary}
\providecommand{\definitionname}{Definition}
\providecommand{\remarkname}{Remark}
\providecommand{\theoremname}{Theorem}

\begin{document}
\title{Kernel-based causal estimators for functional causal effects} 
\author[1]{Yordan P. Raykov}
\author[2]{Hengrui Luo}
\author[3]{Justin D. Strait}
\author[4]{Wasiur R. KhudaBukhsh}
\affil[1]{School of Mathematical Sciences, Horizon Digital Economy Institute, University of Nottingham, Nottingham, UK} 
\affil[2]{Department of Statistics, Rice University, USA; Lawrence Berkeley National Laboratory, USA}
\affil[3]{Statistical Sciences Group, Los Alamos National Laboratory, USA}
\affil[4]{School of Mathematical Sciences, University of Nottingham, Nottingham, UK}
\date{}    
\maketitle
\begin{abstract}
We propose causal effect estimators based on empirical Fr\'{e}chet means and operator-valued kernels, tailored to functional data spaces. These methods address the challenges of high-dimensionality, sequential ordering, and model complexity while preserving robustness to treatment misspecification. Using structural assumptions, we obtain compact representations of potential outcomes, enabling scalable estimation of causal effects over time and across covariates. We provide both theoretical regarding the consistency of functional causal effects, as well as empirical comparison of a range of proposed causal effect estimators\footnote{All experiments can be reproduced via the provided code in https://github.com/JordanRaykov/Kernel-based-estimators-for-Functional-Causal-Effects}.

Applications to binary treatment settings with functional outcomes illustrate the framework's utility in biomedical monitoring, where outcomes exhibit complex temporal dynamics. Our estimators accommodate scenarios with registered covariates and outcomes, aligning them to the Fr\'{e}chet means, as well as cases requiring higher-order representations to capture intricate covariate-outcome interactions. These advancements extend causal inference to dynamic and non-linear domains, offering new tools for understanding complex treatment effects in functional data settings.
\end{abstract}

\section{Introduction }

\subsection{Background}

Causal inference frameworks are often motivated by the desire to estimate causal effects from associational terms computed from observed data. Much of the conventional literature on causal analysis is based on longitudinal observational studies, such as clinical trials and prospective studies, where different configurations of exposure, covariates, and outcome variables are observed at single \citep{rosenbaum1983central} or multiple time points \citep{robins2000marginal}. Experimenters make choices about the exposure mapping or intervention in their causal model, adopting different adjustment procedures depending on whether treatment regimes are statically assigned (i.e.,  at baseline) \citep{robins2000marginal}, dynamically optimized \citep{murphy2003optimal}, time-dependent, or constant over time \citep{pearl1995probabilistic}. Modern applications often involve measurements at various time scales, providing opportunities to study causal effects between variables that follow different temporal dynamics, which are essentially functional data.

Specifically, we study the problem of \emph{estimating the causal effect when covariates and/or outcomes are functional data}. We demonstrate that causal models, equipped with suitable distance metrics for functional data, facilitate more efficient estimators with well-motivated properties, capable of capturing causal effects over different time scales. The primary focus of this work is distinct from designing causal models for time-dependent exposure and treatment strategies detailed in \citet{robins2000marginal}. Instead, our contribution addresses estimating causal effects on complex object spaces \citep{lin2023causal,kennedy2023semiparametric,testa2025doubly,kurisu2024geodesic}, which becomes a prerequisite for the larger challenge of studying time-changing exposures \citep{robins2000marginal,murphy2003optimal}. Our immediate goal is to estimate both scalar and path-valued effects of an exposure on a path-valued functional outcome, potentially extending to multiple time scales \citep{daniel2013methods}.

Causal inference methods for more complex outcomes have gained traction only recently. For example, \citet{lin2023causal} considers distributional differences in generalized Wasserstein space for scalar potential outcomes, while \citet{kurisu2024geodesic} generalizes treatment effect estimation to outcomes in geodesic metric spaces. In our work, the outcomes themselves are more complex, and we introduce a joint aligning-kernelization procedure to estimate causal effects over the mean potential outcomes under both binary and continuous treatments. Previous investigations of binary treatment effects for path-valued outcomes include \citet{belloni2017program}, who develops approximate local treatment effect estimators for non-Euclidean outcomes (forgoing finite-sample inference), and \citet{ecker2024causal}, who uses function-on-scalar regression under a linear parametric assumption. More recently, \citet{testa2025doubly} and \citet{kurisu2024geodesic} propose doubly robust frameworks for complex outcomes. \citet{testa2025doubly} derives point-wise doubly robust estimators for path-valued outcomes in $L_2$, and \citet{kurisu2024geodesic} leverages known geodesic metric spaces to estimate scalar-valued treatment effects. Our kernel-based approach differs in that it does not center on doubly robust properties. Instead, for outcomes in known geodesic spaces, our scalar $\varphi^{dATE}$ estimator can be seen as a special case of \citet{kurisu2024geodesic}'s \textit{geodesic average treatment effect}, and for outcomes in $L_2$, our point-wise effect $\Delta(t)$ aligns with \citet{testa2025doubly}'s estimator. Moreover, we consider practical settings where outcomes may not be truly infinite-dimensional or where the associated geodesic is unknown, in which case we employ nonparametric kernel ridge regression and \emph{square root slope transformations}.

By adopting a nonparametric paradigm for causal effect estimation, our framework covers a more general setting than \citet{belloni2017program} and \citet{ecker2024causal}, and it explicitly addresses a broader range of causal effects motivated by functional outcomes and/or covariates. In contrast to \citet{testa2025doubly}, which focuses on doubly robust estimators, we emphasize kernel-based causal estimation, thus providing an alternative route for practitioners wishing to avoid strong structural assumptions. Our estimators can also be viewed as special cases of conditional average treatment effects, where conditioning occurs on discrete points along the path grid. Crucially, we argue that leveraging a richer representation of outcome structure can improve precision and mitigate potential collider stratification bias \citep{daniel2013methods}.

The proposed path-valued functional potential outcomes remain sensitive to the exposure probability specification (i.e.,  the propensity score) and expected outcome models \citep{savje2024causal}. Building on the kernel-based estimators from \citet{singh2020kernel}, we develop a nonparametric framework that implicitly incorporates the propensity score, eliminating the need for its direct inversion. We extend these estimators to multivariate and infinite-dimensional outcomes by defining purpose-designed operator-valued kernels \citep{kadri2016operator}, and we further demonstrate how the \emph{Fisher--Rao} metric can embed outcome and covariate registration directly into functional causal inference. This preserves both theoretical guarantees of consistency for the expected potential outcome estimators and practical advantages, as illustrated empirically.

Finally, we showcase the utility of our framework for estimating causal effects from digital outcomes in Parkinson’s Disease (PD). Using data from the Parkinson@Home validation study---a passive monitoring study involving 50 participants\footnote{The 2-week follow-up passive monitoring data from the Parkinson@Home validation study is also made available with this study.}---we estimate the effects of levodopa therapy and disease status (PD vs. age-matched non-PD controls). Our higher-dimensional estimators uncover significant disease-related impacts on recorded digital outcomes that often remain undetectable using standard aggregated measures.

\subsection{Illustration}

To illustrate the concepts, consider an example in digital monitoring of Parkinson’s disease (PD) \citep{evers2020real,bloem2019personalized}, where a treatment \(X\) captures daily levodopa-equivalent dosage decisions. Typically, these decisions are set at the start of the day (e.g., total dosage or allocation strategy), rather than intervening anew at every intake event. Meanwhile, the outcome \(Y(t)\) measures PD symptoms (e.g., tremor or bradykinesia) at different times \(t\) throughout the day.

We explore two approaches to modeling \(X\):

\begin{enumerate}
    \item \textbf{Marginal Exposure Models.} Here, \(X \in \{0,1\}\) simply indicates whether any medication is used during the day (or not). This collapses potentially complex dosing schedules into a single binary indicator, yielding a uniform “exposure vs. non-exposure” comparison across the entire day.
    
    \item \textbf{Continuous Treatment Variable.} In this approach, \(X \in \mathbb{R}^+\) is a single, scalar quantity (e.g., total daily dosage) that may also encode timing or frequency of intake (through weights or other aggregation). Although assigned once at baseline, this continuous measure allows for more nuanced estimation of how varying dosage levels influence the outcome \(Y(t)\) over the course of the day.
\end{enumerate}

By avoiding frequent redefinition of \(X\) at each intake event, these two formulations capture realistic clinical scenarios where daily dosage plans are decided once and then implemented. This also facilitates studying how a single day-level treatment choice affects the dynamic evolution of symptoms \(Y(t)\). Each observed \((X,Y)\) can be augmented with covariates \(V(t)\), which may vary at different time scales (e.g., age, activity levels, diurnal patterns). Under mild smoothness assumptions, we propose estimators that measure the effect of these treatment variables on the functional outcomes \(\{Y(t)\}_{t \in [0,1]}\). The rest of the paper is structured as follows: in Section~\ref{sec:prelims}, we provide a brief overview of the causal inference methods applicable to our setup. We discuss binary treatment with functional outcomes in Section~\ref{sec:dynamic-outcomes-binary-treatment}, while in Section~\ref{sec:dynamic_treatment_func_outcome} we consider kernel-based causal effect estimators which also facilitate nonparametric inference in the case of continuous treatments. We present empirical comparison of different causal effect estimators in Section~\ref{sec:experiments} and also evaluate the feasibility of the digital monitoring of PD with the novel causal estimators. Some limitations of the proposed approach and future directions are summarized in Section~\ref{sec:Discussion-Future-Directions}.

\section{Casual Inference Preliminaries}

\label{sec:prelims}

To motivate our setting of function-valued treatment and function-valued outcomes, we first define causal effects in the scalar case using potential outcome notation. Treatments are denoted as $X = x$, where $x$ can be either $0$ or $1$ for binary treatment exposure, or $x \in \mathbb{R}^{+}$ for dose-response exposure. Let $\bm{V} \in \mathbb{R}^{d}$ denote the observed baseline covariates that influence both the treatment $X$ and the outcome $Y$. For each treatment level $x$, the potential outcome $Y^{(x)}$ represents the outcome that would have been observed had the subject received treatment $x$. The potential outcome $Y^{(x)}$ is counterfactual and distinct from the observed conditional expectation $\mathbb{E}[Y \mid X = x]$, which depends on observed treatment assignments.  Inferring $\mathbb{E}[Y^{(x)}]$ from observational data requires assumptions on how the baseline covariates $\bm{V}$ influence both treatment and outcome \citep{hernan2010causal, rosenbaum1983central, robins1994estimation}. Specifically, $\bm{V}$ must satisfy either the \textit{back-door criterion} or the \textit{front-door criterion} with respect to $(X,Y)$ (see Appendix \ref{subsec-causal-graphical-models} for details). 

If $\bm{V}$ satisfies the back-door criterion with respect to $(X,Y)$, the potential outcome expectation can be identified via covariate adjustment:
\begin{equation}
\mathbb{E}[Y^{(x)}] = \int_{\mathcal{V}} \mathbb{E}[Y \mid X = x, \bm{V} = \bm{v}] \, dP_{\bm{V}}(\bm{v}).
\label{eq:back-door-criterion}
\end{equation}
If instead a set of covariates $\bm{V}$ satisfies the front-door criterion, identification follows from a different factorization \citep{pearl1995causal}:
\begin{equation}
\mathbb{E}[Y^{(x)}] = \int_{\mathcal{X}} \left( \int_{\mathcal{V}} \mathbb{E}[Y \mid X = x', \bm{V} = \bm{v}] \, dP_{\bm{V} \mid X}(\bm{v} \mid x) \right) P_{X}(dx'),
\label{eq:front-door-criterion}
\end{equation}
where the treatment probability is marginalized due to its causal effect on $\bm{V}$. For binary treatments $x \in \{0,1\}$, this simplifies to a weighted sum of the expected outcomes in the treated and non-treated groups, adjusted for intermediate confounders.

Causal effects are derived from counterfactual potential outcomes. Assume $n$ independent samples drawn from a superpopulation of $(X, \bm{V}, Y)$, denoted as $(X_i, \bm{V}_i, Y_i)$ for $i = 1, \dots, n$. The covariates matrix is:
\[
\bm{V} = \begin{pmatrix} \bm{V}_{1} \\ \vdots \\ \bm{V}_{n} \end{pmatrix} \in \mathbb{R}^{n \times d},
\]
where each row represents an individual and each column a covariate. For binary exposure $X\in\{0,1\}$, we measure the \textit{average
treatment effect} as a contrast measure between the potential outcomes:
\begin{equation}
\begin{aligned}
\varphi^{ATE} &= \mathbb{E}[Y^{(1)}] - \mathbb{E}[Y^{(0)}] \\
&= \int \varphi^{ATE}(\bm{v}) \, dP_{\bm{V}}(\bm{v}), \quad \text{(Binary treatment effect)}
\end{aligned}
\label{eq:ate}
\end{equation}
where the \textit{heterogeneous treatment effect} for covariates $\bm{V} = \bm{v}$ is:
\begin{equation}
\varphi^{ATE}(\bm{v}) = \mathbb{E}[Y^{(1)} \mid \bm{V} = \bm{v}] - \mathbb{E}[Y^{(0)} \mid \bm{V} = \bm{v}].
\label{eq:ate-heterogeneous}
\end{equation}
In the case of continuous exposure $X\in\mathbb{R}^{+}$ and $Y^{(x)}\in\mathbb{R}$,
simple contrast measures like the above ATE score are not appropriate
since the difference between two continuous levels $x_{1}$ and $x_{2}$
may be influenced by intermediate levels in a cumulative manner. Instead,
the dose-response function for each level $x$ of $X$ is the key
effect of interest: 
\begin{equation}
\begin{aligned}\varphi^{DS}(x) & =\mathbb{E}[Y^{(x)}]\,\quad\text{(Dose-response effect)}\end{aligned}
\label{eq:DS}
\end{equation}
where the formula for $\varphi^{DS}(x)$ in terms of observational terms
takes the form of \eqref{eq:back-door-criterion} or \eqref{eq:front-door-criterion}
depending on whether back-door or front-door criterion for $\bm{V}$
on $\left(X,Y\right)$ are being met.

Estimating causal effects from observational data involves translating the population-level formulas (e.g., \eqref{eq:ate}) into empirical estimators that use observed samples. The validity of these estimators depends on whether back-door or front-door conditions hold and on certain regularity assumptions. A common set of assumptions is the so-called \textit{unconfoundedness}
(or \textit{ignorability}) condition, which is typically aligned with
the back-door criterion: \begin{assumption} (Ignorability) \label{assu:(Ignorability)}
$X\perp Y^{(x)}\mid\bm{V}$, for any $X=x$. \end{assumption}

\begin{assumption} (Positivity) \label{assu:(Positivity)} The propensity
score $\pi(\bm{V})=p(X=1\mid\bm{V})$ is bounded away from 0 and 1,
almost everywhere. That is, $\exists\epsilon>0$ such that $\epsilon<\pi(\bm{V})<1-\epsilon$.
\end{assumption}
In the back-door scenario, Assumption \ref{assu:(Ignorability)}
ensures that once we condition on $\bm{V}$, the distribution of the
potential outcomes is no longer confounded by unobserved variables.
Assumption \ref{assu:(Positivity)} ensures we have sufficient variability
in treatment assignment across values of $\bm{V}$ to estimate causal
effects. For front-door identification, a different set of conditions
apply, involving sequential ignorability assumptions that allow factorizing
the observational distribution as in \eqref{eq:front-door-criterion}.
The estimators discussed next are standard approaches primarily under
the back-door setup, though similar logic can be extended or adapted
for front-door adjustments with more complex estimators. Under the back-door criterion, the potential outcome expectations simplify to:
\[
\mathbb{E}[Y^{(1)}] = \mathbb{E} \left[ \frac{X Y}{\pi(\bm{V})} \right], \quad
\mathbb{E}[Y^{(0)}] = \mathbb{E} \left[ \frac{(1-X) Y}{1 - \pi(\bm{V})} \right].
\]
These expressions allow constructing sample-based estimators. Given $n$ independent and identically distributed (i.i.d.) samples $(X_i, \bm{V}_i, Y_i)$, the \textit{Inverse Probability Weighting} (IPW) estimator of the ATE is:
\begin{equation}
\hat{\varphi}^{IPW} = \frac{1}{n} \sum_{i=1}^{n} \left( \frac{X_i Y_i}{\hat{\pi}(\bm{V}_i)} - \frac{(1 - X_i) Y_i}{1 - \hat{\pi}(\bm{V}_i)} \right).
\label{eq:ipw}
\end{equation}

While IPW is consistent if the propensity score model is correct,
it can be inefficient and sensitive to model misspecification. A more
robust approach is the \textit{Doubly Robust} (DR) estimator, which
augments IPW with outcome models: 
\begin{equation}
\hat{\varphi}^{DR-IPW}=\frac{1}{n}\sum_{i=1}^{n}\left(\frac{X_{i}(Y_{i}-\hat{m}_{1}(\bm{V}_{i}))}{\hat{\pi}(\bm{V}_{i})}+\hat{m}_{1}(\bm{V}_{i})-\frac{(1-X_{i})(Y_{i}-\hat{m}_{0}(\bm{V}_{i}))}{1-\hat{\pi}(\bm{V}_{i})}+\hat{m}_{0}(\bm{V}_{i})\right),\label{eq:doubly_robust_estimator}
\end{equation}
where $\hat{m}_{1}(\bm{V}_{i})=\mathbb{E}[Y\mid X=1,\bm{V}=\bm{V}_{i}]$
and $\hat{m}_{0}(\bm{V}_{i})=\mathbb{E}[Y\mid X=0,\bm{V}=\bm{V}_{i}]$
are estimated outcome regression models. The DR-IPW estimator is consistent
if at least one of the models (the propensity score or the outcome
model) is correct, providing a form of robustness not available from
IPW alone.

The estimators above illustrate how population-level causal effects,
derived under certain identification conditions (notably the back-door
criterion here), can be approximated by sample-based empirical estimators.
Similar estimation strategies can be adapted if the front-door conditions
are met, but the factorization in \eqref{eq:front-door-criterion}
typically leads to different weighting or regression-based methods.
Moreover, for continuous treatments, one can estimate the entire dose-response
function $\varphi^{DS}(x)$ by generalizing these strategies, for
example, employing kernel or regression-based estimators of $\mathbb{E}[Y\mid X=x,\bm{V}]$
and integrating out $\bm{V}$ under the appropriate assumptions.

\section{Binary Treatment with Functional Outcomes}

\label{sec:dynamic-outcomes-binary-treatment}

We study \emph{binary treatment with functional outcomes}: each unit has a baseline covariate vector \(\bm{V} \in \mathcal{V}\), a treatment indicator \(X \in \{0,1\}\), and an outcome trajectory \(\bm{Y}\). Outcomes may be observed either as full trajectories \(\bm{Y}(\cdot) \in L_2([0,1])\), or as discretized vectors \(\bigl(\bm{Y}(u_1), \dots, \bm{Y}(u_T)\bigr) \in \mathbb{R}^T\) for a grid \(0 < u_1 < \cdots < u_T \le 1\). In our analysis, we show that stronger consistency results require assuming \(\bm{Y}(\cdot) \in W^{k,2}([0,1], \mathbb{R})\), i.e., that outcomes lie in a 
a Sobolev space.

\subsection{Notation and setup}

Let $\bigl(\Omega,\mathcal{A},\mathbb{P}\bigr)$ be the underlying probability
space, and let $\bm{Y}:(\Omega,\mathcal{A})\to(\mathcal{F},\mathcal{B})$
be a measurable map taking values in a metric space $(\mathcal{F},\phi)$,
where $\mathcal{B}$ is the Borel $\sigma$-algebra on $\mathcal{F}$.
For each treatment level $x\in\{0,1\}$, let $\bm{Y}^{(x)}:(\Omega,\mathcal{A})\to(\mathcal{F},\mathcal{B})$
denote the corresponding potential outcome. Its \emph{distribution}
(or pushforward measure) is denoted by $\eta_{x}$, where for any
${A}\in\mathcal{B}$, $\eta_{x}({A})=\mathbb{P}\bigl(\bm{Y}^{(x)}\in{A}\bigr).$
In practice, the outcome space $(\mathcal{F}, \phi)$ may be chosen to reflect either a finite--dimensional representation e.g., $(\mathbb{R}^T, \|\cdot\|_2)$, or a functional space such as $(L_2([0,1]), \|\cdot\|_2)$ or $(W^{k,2}, \|\cdot\|_{k,2})$ to capture smoothness and differentiability. This flexibility allows us to handle both discretized and continuous functional data in a unified framework. By the law of total probability, the measure $\eta_{x}$ can be expressed by integrating out the covariates $\bm{V}$, with the identification of the distribution $P_{\bm{V}}$ depending on whether the \emph{back-door}
or \emph{front-door} criteria are satisfied, and how the joint observational distribution of the superpopulation $(X,\bm{V},\bm{Y})$ (i.e. \emph{conditional exchangeability/ignorability}) is used.

Causal estimators are based on this superpopulation $(X,\bm{V},\bm{Y})$,
where covariates and treatment are assigned at baseline. The baseline
assignment is critical, as it ensures that covariates and treatment
are predetermined and unaffected by the high-dimensional
outcomes, enabling the formulation of analogous expressions for potential
outcomes and corresponding causal effects. More precisely, under the \textit{back-door} criterion for $\bm{V}$
for the covariates, $\bm{Y}^{(x)}\perp X\,\lvert\bm{V}$, so we write:
\begin{align*}
\eta_{x}({A}) & =\int P(\bm{Y}^{(x)}\in {A}\mid\bm{V}=\bm{v})\,dP_{\bm{V}}(\bm{v})\\
 & =\int P(\bm{Y}\in {A}\mid X=x,\bm{V}=\bm{v})\,dP_{\bm{V}}(\bm{v})\\
 & = E_V\left[E\left[\mathbf{1}\{\bm{Y}^{(x)}\in A\}\lvert V\right]\right].
\end{align*}

Under the \textit{front-door} criterion for $\bm{V}$, the measure
can be viewed as: 
\[
\eta_{x}({A})=\int P\bigl(\bm{Y}^{(x)}\in {A}\mid\bm{V}^{(x)}=\bm{v}\bigr)\underbrace{P_{\bm{V}^{(x)}}(d\bm{v})}_{\text{distribution of \ensuremath{\bm{V}} under do(\ensuremath{X=x})}}.
\]
where front-door factorization is used to identify $P_{\bm{V}^{(x)}}$
and $P\bigl(\bm{Y}^{(x)}\in {A}\mid\bm{V}^{(x)}=\bm{v}\bigr)$ from the
observational distribution $(X,\bm{V},Y)$ with $P_{\bm{V}^{(x)}}=\int_{\mathcal{V}}P(\bm{V}=\bm{v}\mid X=x')P_{X}(dx')$.

Suppose the outcome space \((\mathcal{F},\phi)\) is a complete, separable
metric space, e.g. the finite dimensional space
$\bigl(\mathbb{R}^T,\|\cdot\|_2\bigr)$ or the Hilbert space
\(L_2([0,1])\). For any treatment level \(x\in\{0,1\}\) the \emph{population
Fréchet 2-mean}\footnote{%
We subsequently drop the qualifier “2-’’ and refer to it simply as the
Fréchet mean.}
of the potential outcome distribution $\eta_x$ is:

\begin{equation}
\label{eq:Frechet_mean-1}
\begin{aligned}
F(\eta_x)
&=\arg\min_{f\in\mathcal F}
     \int_{\mathcal F}\phi^{2}(f,g)\,d\eta_x(g) \\
&=\arg\min_{f\in\mathcal F}
     \int\!\Bigl[\,
       \int_{\mathcal F}\phi^{2}(f,g)\,d\eta_x\bigl(g\mid\bm V=\bm v\bigr)
     \Bigr]dP_{\bm V}(\bm v).
\end{aligned}
\end{equation}
When \((\mathcal{F},\phi)\) has non-positive curvature
(e.g.\ any Hilbert space with its induced norm)
or when \(f\mapsto\phi^{2}(f,\cdot)\) is strictly convex,
the minimiser in \eqref{eq:Frechet_mean-1} exists and is unique
\citep{evans2024limit}. Defining the Fréchet mean only requires completeness of
\((\mathcal F,\phi)\); however, our later consistency and
convergence rate results make heavier use of differentiability specifically,
we work with the Sobolev space $W^{k,2}([0,1])$ $(k\ge1)$, equipped with the norm $\|\cdot\|_{k,2}$.

\paragraph{Balancing weight (IPW) representation.}
Consider a fixed treatment arm $x\in\{0,1\}$. Define the propensity score and the \emph{inverse probability weight}:
\begin{equation*}
\pi_x(v) = P(X = x \mid \bm{V} = v), \qquad
\omega_x = \frac{\mathbf{1}\{X = x\}}{\pi_x(\bm{V})}.
\end{equation*}

Let us denote with $Q$ the observational joint distribution of $(X,\bm{V},\bm{Y})$:
\begin{equation*}
Q(dx,dv,dy)=P_{\bm{V}}(dv)P(dx\lvert v)P(dy\lvert x,v).
\end{equation*}
Under the back-door assumptions (consistency, ignorability, positivity), the joint distribution after the intervention $\operatorname{do}(X=x)$ is
\begin{equation*}
    R_x(dv,dy)=P_{\bm{V}}(dv)\,\eta_x(dy)\quad\bigl(\text{product measure}\bigr).
\end{equation*}

\begin{lemma}[IPW weight as likelihood ratio]
Under \emph{back–door} assumptions, the balancing weight $\omega_x$ is the Radon–Nikodým derivative of $R_x$ with respect to $Q$, meaning it re-weights expectations under $Q$ to recover those under $R_x$, i.e.
\(\displaystyle 
\frac{dR_x}{dQ}(x,v,y)=\omega_x
\)
for \(Q\)-almost every \((x,v,y)\). 
\end{lemma}
\begin{proof}
Ignorability gives $P(\bm{Y}\in dy \mid x,v) = P(\bm{Y}^{(x)}\in dy\mid v)$. Let $\lambda(dx, dv, dy)$ be a reference measure dominating both $Q$ and $R_x$, used to define their Radon-Nikodým derivatives:

\begin{equation}
\lambda(dx,dv,dy)
   = \bigl(\text{counting measure on }\{0,1\}\bigr) \otimes
     P_{\bm V}(dv) \otimes \mu_Y(dy),
\end{equation}
where $\mu_Y$ is any $\sigma$-finite reference measure dominating the conditional distribution $P(dy \mid x,v)$, e.g., Lebesgue on $\mathbb{R}^T$ or the Borel measure on a function space. The observational joint distribution  $Q$ of $(X,\bm{V},\bm{Y})$ has density:
\begin{equation}
q(x,v,y) := \frac{dQ}{d\lambda}(x,v,y)
          = P(x \mid v) \cdot P(y \mid x,v).
\label{eq:obs-density}
\end{equation}
Under the back-door criterion, the interventional distribution is the product
measure $R_x(dv,dy) = P_{\bm V}(dv) \cdot \eta_x(dy)$, so:
\begin{equation}
r_x(x,v,y) := \frac{dR_x}{d\lambda}(x,v,y)
            = \mathbf{1}\{x\} \cdot P(y \mid x,v).
\label{eq:int-density}
\end{equation}

Divide \eqref{eq:int-density} by \eqref{eq:obs-density}:
\begin{align*}
\frac{dR_x}{dQ}(x,v,y)
  &= \frac{r_x(x,v,y)}{q(x,v,y)} \\
  &= \frac{\mathbf{1}\{x\} \cdot P(y \mid x,v)}
           {P(x \mid v) \cdot P(y \mid x,v)} \\
  &= \frac{\mathbf{1}\{x\}}{P(x \mid v)} \\
  &= \frac{\mathbf{1}\{x\}}{\pi_x(v)} \\
  &= \omega_x.
\end{align*}

\noindent
Thus, the inverse-propensity weight $\omega_x$ is the Radon--Nikodým derivative that converts the observational $Q$ into the interventional $R_x$.
\end{proof}

For any bounded Borel $f\colon\mathcal{F}\to\mathbb{R}$, 
\begin{align}
\underbrace{\eta_x(f)}_{\text{causal target}}
     &=\iint f(y)\,R_x(dv,dy)
       =E_{Q}\!\bigl[\omega_x\,f(\bm Y)\bigr].
\label{eq:IPW-importance-sampling}
\end{align}

Hence for every $A\in\mathcal B(\mathcal F)$
$\eta_x(A)=E_Q[\omega_x\mathbf 1\{\bm Y\in A\}]$. \eqref{eq:IPW-importance-sampling} is precisely the importance sampling (i.e. where expectations under a target distribution are computed via reweighted draws from a proposal distribution) with target density being the interventional $R_x$, the proposal density being the observational $Q$, and the likelihood ratio being the weight $\omega_x$. Given i.i.d.\ observations \(\{(X_i,\bm V_i,\bm Y_i)\}_{i=1}^n\sim Q\) set the conventional inverse probability weighting estimator of $\eta_x(f)$:
\[
\hat\eta_{x,n}(f)=\frac1n\sum_{i=1}^{n}\omega_{x,i}\,f(\bm Y_i),
\qquad
\omega_{x,i}=\frac{\mathbf 1\{X_i=x\}}{\pi_x(\bm V_i)}.
\]
$\hat\eta_{x,n}(f)$ is the Monte-Carlo importance sampling estimator of $\eta_x(f)$ obtained from a sample $\left(X_i,\bm{V}_i,\bm{Y}_i\right)\sim Q$. The classical properties of importance sampling estimators (see \citep[§3.2]{robert1999monte}) yield the following consistency and limit behavior for $\hat\eta_{x,n}(f)$:
\begin{itemize}
\item \emph{Unbiasedness}: \(E_Q[\hat\eta_{x,n}(f)]=\eta_x(f)\).
\item \emph{SLLN}: if \(E_Q[\omega_x|f(\bm Y)|]<\infty\) then 
      \(\hat\eta_{x,n}(f)\xrightarrow{\text{a.s.}}\eta_x(f)\).
\item \emph{CLT}: if \(E_Q[\omega_x^2f(\bm Y)^2]<\infty\) then  
      \(\sqrt n\,(\hat\eta_{x,n}(f)-\eta_x(f))\overset d\to
        \mathcal N\bigl(0,\operatorname{Var}_Q[\omega_xf(\bm Y)]\bigr)\).
\end{itemize}

Taking $f_y(\cdot)=\phi^2(\,\cdot\,,y)$ in \eqref{eq:IPW-importance-sampling} gives
\begin{equation}
F(\eta_x)
   =\arg\min_{f\in\mathcal F}
        E_Q\!\bigl[\omega_x\,\phi^2(f,\bm Y)\bigr].
\label{ipw-metric-space}
\end{equation}
which generalises the conventional inverse-propensity weighted estimator for the average treatment effect to functional or metric-valued outcomes. Let \(\widehat\pi\) be any uniformly consistent estimator of $\pi_x$, and define stabilised weights
$
\widehat\omega_{x,i}
  =\frac{\mathbf 1\{X_i=x\}}
         {x\,\widehat\pi(\bm V_i)+(1-x)\,[1-\widehat\pi(\bm V_i)]}.
$
The empirical Fréchet mean
\begin{equation}
\hat F_{n}
  =\arg\min_{f\in\mathcal F}
      \sum_{i=1}^{n}\widehat\omega_{x,i}\,\phi^2(f,\bm Y_i)
\label{eq:Frechet_mean-2}
\end{equation}
converges in probability (a.s.\ if the weights are bounded) to
$F(\eta_x)$, the population Fréchet mean of the potential outcomes under treatment $x$; the bias of the \emph{self-normalised} variant  
\(\widehat\omega_{x,i}/\sum_{j}\widehat\omega_{x,j}\)
is \(O(n^{-1})\) \citep[Thm. 3.2]{robert1999monte}.


When the Fréchet means are unique, they belong to a special class
of $M$-estimators \citep{huber2011robust} defined by so-called \emph{$\rho$-loss}
(i.e.,  here substituting $\rho=\phi^{2}$), allowing us to leverage well-developed theory to establish their consistency, robustness, and rates of convergence.
More general cases where the solution from \eqref{eq:Frechet_mean-1}
forms \emph{Fréchet mean sets} inherit only weaker forms of consistency \citep{schotz2022strong,evans2024limit}.
\begin{assumption} (Uniqueness) \label{assu:(uniqueFrechet)} We
assume that \eqref{eq:Frechet_mean-1} and \eqref{eq:Frechet_mean-2}
are unique up to probability 1 with respect to probability measure
$\mathbb{P}$. 
\end{assumption} 
\begin{defn}
\label{def:dynamic-ate-1} Let $\bm{Y}^{(0)}$ and $\bm{Y}^{(1)}$
be random elements taking values in a metric space $(\mathcal{F},\phi)$,
with corresponding probability distributions $\eta_{0}$ and $\eta_{1}$.
Assume that each has \emph{the unique} Fréchet mean in $\mathcal{F}$,
defined by \eqref{eq:Frechet_mean-1} for $x\in\{0,1\}$. We define
the scalar \emph{dynamic average treatment effect} by 
\begin{equation}
\varphi^{dATE}=\phi\bigl(F(\bm{Y}^{(1)}),F(\bm{Y}^{(0)})\bigr).\label{eq:varphi_dATE_general}
\end{equation}
This quantity measures the distance between the conditional Fréchet means of the potential outcome distributions under treatment and control.
\end{defn}

\begin{rem}
For sets of generalized Fréchet means \citep{aveni2024uniform}, the
minimizer defined in \eqref{eq:Frechet_mean-1} is non-unique, we
can define the dynamic treatment effect using the Hausdorff distance $d_{H}\left(\cdot,\cdot\right)$
between the Fréchet mean sets:

\begin{equation}
\varphi^{dATE}=d_{H}\big(F_{M}(\bm{Y}^{(1)}),F_{M}(\bm{Y}^{(0)})\big),
\end{equation}
where $F_{M}(\bm{Y}^{(x)})$ denotes the Fréchet mean set for treatment
$x$ to extend our approaches. Unfortunately, such set-valued estimators
arising as Fréchet means are less amenable to classical asymptotic
analysis because the Hausdorff distance in set space is not differentiable,
which impedes the application of central limit theorems without imposing
additional structure. 
\end{rem}

The consistency of $\varphi^{dATE}$ and the asymptotic normality
of the estimation residual which would allow us to leverage existing theory
when estimating the significance of an estimated effect require a
few additional assumptions which we specify more formally below.

\subsection{Discretisation, interpolation and consistency} 
In practical applications, we typically observe each outcome trajectory $\bm{Y}_i \in \mathcal{F}$ only at a finite number of time points. That is, for each unit $i=1,\dots,n$, we have access to the discretized values
\[
\bigl(\bm{Y}_i(u_1), \bm{Y}_i(u_2), \ldots, \bm{Y}_i(u_T)\bigr) \in \mathbb{R}^T
\]
evaluated on a fixed grid of time points. We fix an integer $T\ge 1$ and define the grid as
\[
\mathcal U_T = \{ 0 < u_1 < u_2 < \cdots < u_T \le 1 \},
\qquad \text{with } u_0 := 0,\; u_{T+1} := 1.
\]

Note that arbitrary functional spaces $(\mathcal{F},\phi)$ do not necessarily admit a minimal norm interpolant. Therefore, we will restrict our attention to spaces where we do have an interpolant with a good convergence property. 

\begin{defn}
\label{def:convergent-interpolant}
Let $(\mathcal F,\phi)$ be a metric space of real–valued functions on
\([0,1]\).
A map
\(\mathsf{Interpolant}\colon\bigcup_{T\ge1}\mathbb R^{T}\to\mathcal F\)
is \emph{convergent} if, for every $f\in\mathcal F$,
\[
\phi\!\bigl(\mathsf{Interpolant}(f(u_1),\dots,f(u_T)),\,f\bigr)
   \;\longrightarrow\;0
\quad\text{as }\;
\max_{0\le i\le T}\,|u_{i+1}-u_i| \;\to 0.
\]
\end{defn}

The most prominent example of a convergent interpolant is the standard piecewise linear interpolant on $(C([0, 1]), ||\cdot||_{\infty})$, the space of continuous functions equipped with the supremum norm.  While this is a Banach space, it is not a Hilbert space. In contrast, $L^2([0, 1])$ is a Hilbert space with inner product $\langle f,g \rangle = \int_0^1 f(u)g(u) du$. However, interpolants in $L^2$ do not exhibit good convergence properties unless additional regularity is imposed on the functions. For this reason, we consider Sobolev spaces of functions \citep{Adams2003SobolevSpaces}.

Let $W^{k, p}([0, 1], \mathbb{R})$ denote the subspace of $L^p([0, 1])$ containing functions $f$ such that the function $f$
and its weak derivatives up to order $k$ have a finite $L^p$ norm. When equipped with the natural norm 
\begin{align*}
    ||f||_{k, p} = \left(\sum_{i=0}^{k} ||\partial^{(i)} f||_p^{p} \right)^{\frac{1}{p}},
\end{align*}
where $\partial^{(i)}f$ denotes the $i$-th derivative of $f$, the Sobolev space 
$(W^{k, p}([0, 1], \mathbb{R}), ||\cdot||_{k,p})$ turns into a Banach space. With a slight abuse of notation, we will treat this Banach space as a metric space, using the norm-induced metric $\|\cdot\|_{k,p}$ and denoting it by the same symbol. We will mostly focus on the spaces $W^{k, 2}([0, 1], \mathbb{R})$ for $k \ge 1$.  For $k\ge 1$, and $p=2$, the space $(W^{k, 2}([0, 1], \mathbb{R}), ||\cdot||_{k,2})$ is also a Hilbert space with the inner product
   \begin{align*}
       \langle f, g \rangle_{W^{k, 2}} = \sum_{i=0}^{k} \langle \partial^{(i)} f, \partial^{(i)} g \rangle_{L_2},
   \end{align*}
   where $\langle \cdot, \cdot \rangle_{L_2}$ is the standard inner product in the $L_2([0, 1])$ space. This Hilbert space structure will be crucial for our purposes.
   
Given a convergent interpolant (Definition~\ref{def:convergent-interpolant}), define
\[
\hat{\bm Y}_{i,T}
   :=\mathsf{Interpolant}\!\bigl(\bm Y_i(u_1),\dots,\bm Y_i(u_T)\bigr),
   \qquad i=1,\dots,n.
\]
\medskip
Fix a treatment arm \(x\in\{0,1\}\) and stabilised weights
\[
\widehat\omega_{x,i}
     =\frac{\mathbf 1\{X_i=x\}}{\widehat\pi_x(\bm V_i)}
     =\frac{\mathbf 1\{X_i=1\}}{\widehat\pi(\bm V_i)}
      \;+\;
      \frac{\mathbf 1\{X_i=0\}}{1-\widehat\pi(\bm V_i)};
\]
recall we have used the shorthand notation
\(\widehat\pi_1=\widehat\pi\) and
\(\widehat\pi_0=1-\widehat\pi\).
Set
\[
\widehat\eta^{\mathrm{IPW}}_{x,n,T}
   =\frac1n\sum_{i=1}^{n}\widehat\omega_{x,i}\,\delta_{\hat{\bm Y}_{i,T}},
\qquad
\widehat\eta^{\mathrm{IPW}}_{x,n}
   =\frac1n\sum_{i=1}^{n}\widehat\omega_{x,i}\,\delta_{\bm Y_i}.
\]

\medskip

Because \((\mathcal F,\phi)\) is CAT(0) (Hilbert in our main
application) and
\(\widehat\eta^{\mathrm{IPW}}_{x,n,T}\in\mathcal P^{(2)}(\mathcal F)\),
Sturm’s Proposition 4.3 (with the 1-Lipschitz bound in
$W_2$, Theorem 6.3) ensures the barycentre
\(F\!\bigl(\widehat\eta^{\mathrm{IPW}}_{x,n,T}\bigr)\) exists and is
\emph{unique} \citep{Sturm2003NPC}.  Concretely,
\begin{equation}
\hat F_{x,n,T}
   :=F\!\bigl(\widehat\eta^{\mathrm{IPW}}_{x,n,T}\bigr)
   =\arg\min_{f\in\mathcal F}
        \sum_{i=1}^{n}\widehat\omega_{x,i}\,
        \phi^{2}\!\bigl(f,\hat{\bm Y}_{i,T}\bigr).
\label{eq:weighted-Frechet-mean-discrete}
\end{equation}
If \(\mathcal F\) is a Hilbert space equipped with its norm metric,
\(\hat F_{x,n,T}\) coincides with the usual \emph{vector} mean of the
weighted sample when $\hat{\omega}_{x,i}\equiv 1$. Theorem~\ref{thm:IPW-interpolant-consistency} below shows that the weighted Fréchet
mean \eqref{eq:weighted-Frechet-mean-discrete} converges almost surely
to the population mean $F(\eta_x)$ as the grid is refined and the sample
size increases.

\begin{thm}[IPW Fréchet mean LLN]
\label{thm:IPW-interpolant-consistency}
Let $(\mathcal F,\phi)=(W^{k,2}[0,1],\|\cdot\|_{k,2})$ with $k\ge1$ and
assume

\begin{enumerate}
\item Back–door assumptions of \emph{consistency} $\bm Y=\bm Y^{(X)}$, \emph{conditional ignorability} $\bm Y^{(x)}\!\perp\! X\mid\bm V$ and \emph{positivity} $0<\pi_x(\bm V)<1$ a.s.
\item $E\bigl[\phi^2(\bm Y,y_0)\bigr]<\infty$ for some (hence every) $y_0\in\mathcal F$.
\item The map $\mathsf{Interpolant}$ satisfies Definition~\ref{def:convergent-interpolant}.
\item A uniformly consistent estimator $\widehat\pi$ exists with $\sup_{v\in\mathcal V}|\widehat\pi(v)-\pi(v)|\xrightarrow{a.s.}0$.
\end{enumerate}

For each sample size $n$, let $\mathcal{U}_T = \{0 < u_1 < \dots < u_T \le 1\}$ denote a grid of $T$ time points. We assume that the grid becomes dense as $T \to \infty$, in the sense that
\[
\max_{1 \le i \le T-1} |u_{i+1} - u_i| \xrightarrow{T \to \infty} 0.
\]

\noindent
Then the population Fréchet mean $F(\eta_x)$ is unique and
\begin{equation}
\displaystyle
\lim_{n\to\infty}\;
\lim_{T\to\infty}\;
      \hat F_{x,n,T}
      \;=\;
      F(\eta_x)
      \quad\mathbb P\text{-almost surely.}
\;\;
\end{equation}

\end{thm}

\begin{proof}
By construction
\[
\widehat\eta^{\mathrm{IPW}}_{x,n,T}
      =\frac1n\sum_{i=1}^{n}\widehat\omega_{x,i}\,
        \delta_{\hat{\bm Y}_{i,T}},
\qquad
\widehat\eta^{\mathrm{IPW}}_{x,n}
      =\frac1n\sum_{i=1}^{n}\widehat\omega_{x,i}\,
        \delta_{\bm Y_i},
\]
and $\sum_i\widehat\omega_{x,i}=n$ almost surely (for details see Appendix \ref{sec:Balancing-weights}). Because $(\mathcal F,\phi)$ is Hilbert and
\emph{Assumption 2} gives a finite second moment, both weighted measures lie in
$\mathcal P^{(2)}(\mathcal F)$.  
\citet[Proposition 4.3]{Sturm2003NPC} then guarantees that the barycentres
\[
\hat F_{n,T}:=F(\widehat\eta^{\mathrm{IPW}}_{x,n,T}),
\qquad
\hat F_{n,\infty}:=F(\widehat\eta^{\mathrm{IPW}}_{x,n})
\]
exist and are \emph{unique.} 

By the fundamental contraction property \cite[Theorem 6.3]{Sturm2003NPC}, we have 
\[
\phi(\hat F_{n,T},\hat F_{n,\infty})
   \;\le\;
   W_1\!\bigl(\widehat\eta^{\mathrm{IPW}}_{x,n,T},
              \widehat\eta^{\mathrm{IPW}}_{x,n}\bigr).
\]
Couple the two weighted measures by
\(\tfrac1n\sum_{i=1}^{n}\widehat\omega_{x,i}
                     \delta_{(\hat{\bm Y}_{i,T},\bm Y_i)}\).
Then, from the definition of the Wasserstein-1 distance
\[
W_1\!\bigl(\widehat\eta^{\mathrm{IPW}}_{x,n,T},
           \widehat\eta^{\mathrm{IPW}}_{x,n}\bigr)
   \;\le\;
   \frac1n\sum_{i=1}^{n}\widehat\omega_{x,i}\:
          \bigl\|\hat{\bm Y}_{i,T}-\bm Y_i\bigr\|_{k,2}.
\tag{A}
\]

Weights are bounded in probability under positivity and the uniform
consistency of $\widehat\pi$ (\emph{Assumption 4}).  
Because the chosen interpolant is convergent
(\emph{Assumption 3}),
the Bramble–Hilbert lemma for Sobolev spaces
\citep{bramble1970estimation}
gives
\[
\bigl\|\hat{\bm Y}_{i,T}-\bm Y_i\bigr\|_{k,2}
      \;=\;O\!\bigl(\Delta_{T}^{\,k}\bigr)
      \;\xrightarrow[T\to\infty]{}0
      \quad\text{for each }i,
\]
where
$\Delta_{T}:=\max_{1\le j\le T-1}|u_{j+1}-u_j|$ is the mesh width.
Hence the right–hand side of \textup{(A)} converges almost surely to
\(0\) for every fixed \(n\), and therefore
\[
\phi(\hat F_{n,T},\hat F_{n,\infty})
   \;\xrightarrow[T\to\infty]{\;a.s.\;}0.
\tag{B}
\]

For any bounded Borel $f$ define
\[
S_n \;=\; \frac1n\sum_{i=1}^{n}\widehat\omega_{x,i},
\qquad
T_n(f) \;=\; \frac1n\sum_{i=1}^{n}\widehat\omega_{x,i}f(\bm Y_i).
\]
The summands are independent and have finite variance, so by
Kolmogorov’s strong law for independent variables
\citep[Chap.~2, §2.5, “Kolmogorov’s SLLN”]{durrett2019probability}
\[
S_n \xrightarrow{\text{a.s.}} 1,
\qquad
T_n(f) \xrightarrow{\text{a.s.}} \mathbb E[\omega_x f(\bm Y)].
\]
Consequently the self-normalised ratio
$T_n(f)/S_n$ converges almost surely to the same limit, whence the IPW
empirical measure
\[
  \widehat\eta^{\mathrm{IPW}}_{x,n}
  \;=\;\sum_{i=1}^{n}\tilde\omega_{x,i}\,\delta_{\bm Y_i},
  \quad
  \tilde\omega_{x,i}=\widehat\omega_{x,i}/\sum_j\widehat\omega_{x,j},
\]
converges to $\eta_x$ in the Kantorovich–Rubinstein metric $W_{1}$.
(Recall $W_{1}\le W_{2}$ when second moments are finite.)
\begin{rem}
While the self-normalized weights $\tilde\omega_{x,i}$ are not independent, the convergence of the empirical measure follows from showing that both numerator and denominator converge almost surely, and invoking a Slutsky-type argument for a.s.\ convergence of ratios.
\end{rem}

Sturm’s contraction with $\mu=\widehat\eta^{\mathrm{IPW}}_{x,n}$ and
$\nu=\eta_x$ gives
\[
\phi\!\bigl(\hat F_{n,\infty},F(\eta_x)\bigr)
   \;\le\;
   W_1\!\bigl(\widehat\eta^{\mathrm{IPW}}_{x,n},\eta_x\bigr)
   \;\xrightarrow[n\to\infty]{\;a.s.\;}0,
\tag{C}
\]
because the Kantorovich–Rubinstein distance $W_{1}$ between the two
measures goes to zero by the weighted SLLN just proved.

For every \(n\) and \(T\) the ordinary triangle inequality in
\((\mathcal F,\phi)\) gives
\[
\phi\!\bigl(\hat F_{n,T},F(\eta_x)\bigr)
   \;\le\;
   \underbrace{\phi\!\bigl(\hat F_{n,T},\hat F_{n,\infty}\bigr)}_{\text{(B) interpolation error}}
   +\underbrace{\phi\!\bigl(\hat F_{n,\infty},F(\eta_x)\bigr)}_{\text{(C) sampling error}}.
\tag{D}
\]
Take first the limit \(T\to\infty\) (with \(n\) fixed) and apply (B),
then let \(n\to\infty\) and apply (C).  Inequality (D) yields
\[
\lim_{n\to\infty}\;\lim_{T\to\infty}
      \phi\!\bigl(\hat F_{n,T},F(\eta_x)\bigr)=0
\quad\text{almost surely}.
\]
Because \(\phi\) is a metric, this is equivalent to
\(
\displaystyle
\lim_{n\to\infty}\,\lim_{T\to\infty}\hat F_{x,n,T}=F(\eta_x)
\)
$\mathbb P$–a.s., completing the proof.
\end{proof}

The proof of Theorem~\ref{thm:IPW-interpolant-consistency} follows by separating the two sources of errors: one due to interpolation and one due to finite sample sizes. As such, the reader will realize there is no real reason to restrict to the Sobolev space as long as the space under consideration allows for both of these errors to be controlled sufficiently well. Therefore, the statement of the theorem should be true for a much larger class of functions. However, for our purposes, we confine ourselves to the class of Sobolev functions for its Hilbert space structure.

We also consider a more challenging case where the Fréchet mean estimators are defined under the Fisher--Rao metric, which requires constraining the domain $\mathcal{F}$ of the outcomes $\bm{Y}$. In the infinite-dimensional setting, assuming that outcomes lie in a Hilbert subspace of $L_2([0,1])$, the uniqueness of the Fréchet mean has been well studied \citep{bridson2013metric}. However, laws of large numbers under different metric choices $\phi$ are beyond the scope of our work. We direct the reader to \citet{Sturm2003NPC} and \citet{Afsari2011LpCentreofMass} for further discussion on this topic.

\subsection{Important special cases \label{subsec:Euclidean-Metric-on}}
\subsubsection{Sobolev space $\mathcal{F}$ with Euclidean $\phi$}

We now examine an important special case where the functional outcomes lie in a Sobolev space \(\mathcal{F} = W^{k,2}([0,1])\), but are represented through discretization on a regular grid \(0 < u_1 < \cdots < u_T \le 1\). That is, each trajectory \(\bm Y_i(\cdot) \in \mathcal{F}\) is observed at finitely many time points and mapped to a vector \((\bm Y_i(u_1), \dots, \bm Y_i(u_T)) \in \mathbb{R}^T\). In this discretized setting, the metric \(\phi\) on \(\mathcal{F}\) is taken to be the Euclidean distance in \(\mathbb{R}^T\), i.e.,
\[
    \phi(f, g) := \|f - g\|_2 = \left( \sum_{j=1}^T |f(u_j) - g(u_j)|^2 \right)^{1/2}.
\]
This effectively equips the discretized Sobolev trajectories with a finite--dimensional Euclidean structure, making \((\mathcal{F},\phi)\) isometric to \((\mathbb{R}^T, \|\cdot\|_2)\).

Under this metric, the Fréchet mean reduces to the classical weighted mean in \(\mathbb{R}^T\), and the dynamic treatment effect can be defined pointwise across the time grid. Specifically, we define the pointwise inverse-probability-weighted average treatment effect at time \(t = u_j\) by:

\begin{defn}
\label{def:dynamic-ate} For random variables $\bm{Y}^{(0)}$ and
$\bm{Y}^{(1)}$ with independent distributions $\eta_{0}$ and $\eta_{1}$
respectively, and taking values in $\mathbb{R}^{T}$, assume further that
$\mathcal{F}$ is also endowed with a vector space structure. We can
define the \textit{dynamic average treatment function effect} as the
estimated pointwise difference: 
\begin{align}
\Delta(t) & =F\left(\bm{Y}^{(1)}\right)(t)-F\left(\bm{Y}^{(0)}\right)(t), \quad t=1,\cdots,T
\end{align}
The quantity $\Delta\in\mathcal{F}$ captures the \emph{dynamic} (pointwise) difference
between the Fréchet means. The scalar $\varphi^{dATE}$ is recovered
from the vector norm: 
\begin{equation}
\varphi^{dATE}=\phi\big(F(\bm{Y}^{(1)}),F(\bm{Y}^{(0)})\big)=\|\Delta\|_{\phi}.\label{eq:dynamic_causal_effect}
\end{equation}
\end{defn}
In the case of Euclidean distance $\phi$, the norm above is simply $\varphi^{dATE}=\|\Delta\|_{2}$. This leads to the following result: 
\begin{thm}
\label{thm:l2_dte_asymp_normality} Assuming $\bm{Y}^{(1)}$ and $\bm{Y}^{(0)}$
are identifiable (Assumptions \ref{assu:(Ignorability)}, and \ref{assu:(Positivity)}),
and associated with independent distributions $\eta_{0}$ and $\eta_{1}$,
over finite $\mathbb{R}^{T}$. Then, the residuals between the effect $\Delta=\left(\Delta(1),\dots,\Delta(T)\right)^{T}$
and the point-wise estimator for effect $\hat{\Delta}\in\mathbb{R}^{T}$
are asymptotically normal: 
\[
\sqrt{n}\,\bigl(\hat{\Delta}-\Delta\bigr)\xrightarrow{d}\mathcal{N}(\bm{0},\bm{K}),
\]
where $\bm{K}=\mathbb{E}\left[\Delta\Delta^{T}\right]$ reflects the
covariance structure of the expected effect. For non-zero population
effect $\|\Delta\|_{2}>0$, we further get 
\[
\sqrt{n}\,\Bigl(\hat{\varphi}^{dATE}-\varphi^{dATE}\Bigr)\;=\;\sqrt{n}\,\bigl(\|\hat{\Delta}\|_{2}-\|\Delta\|_{2}\bigr)\;\xrightarrow{d}\;\mathcal{N}\bigl(0,\sigma^{2}\bigr),
\]
where 
\[
\sigma^{2}\;=\;\frac{1}{\|\Delta\|_2^{2}\Bigl(\|\zeta^{(1)}\Delta\|_{2}^{2}+\|\zeta^{(0)}\Delta\|_{2}^{2}\Bigr)}.
\]
and each $\zeta^{(x)}$ being the covariance operator associated with
mean-zero Gaussian process characterizing the “limiting fluctuation”
of the empirical distribution around the true distribution of $\bm{Y}^{(x)}$.
Under a zero vector population effect assumption $\|\Delta\|_{2}=0$
we get 
\[
\sqrt{n}\,\Bigl(\hat{\varphi}^{dATE}-\varphi^{dATE}\Bigr)\;=\;\sqrt{n}\,\|\hat{\Delta}\|_{2}\;\xrightarrow{d}\;\|\mathcal{Z}\|,
\]
where $\mathcal{Z}$ is a mean zero Gaussian vector in $\mathbb{R}^{T}$
with covariance matrix $\bm{K}$; the distribution of $\|\mathcal{Z}\|^{2}$
is a generalized $\chi^{2}$-distribution. 
\end{thm}

\begin{proof}
See Appendix \ref{subsec:Asymptotic-normality-L2}. 
\end{proof}

\begin{rem}
The pointwise residuals are also asymptotically normal with: 
\[
\sqrt{n}\,\bigl(\hat{\Delta}(t)-\Delta(t)\bigr)\;\;\xrightarrow{d}\;\;\mathcal{N}(0,\,K_{t}),\quad t=1,\cdots,T
\]
where $K_{t}=\mathbb{E}\left[\Delta(t)^{2}\right]$. 
\end{rem}

\begin{rem}
The asymptotic distribution of the residuals allows us to derive closed
form confidence intervals under both scenarios. For strictly positive
norms $\varphi^{dATE}>0$, a $(1-\alpha)$ confidence interval (CI) for $\|\Delta\|_2$ is: 
\[
\text{CI: }\Bigl[\hat{\varphi}^{dATE}-z_{\alpha/2}\sqrt{\frac{\hat{\sigma}^{2}}{n}},\;\hat{\varphi}^{dATE}+z_{\alpha/2}\sqrt{\frac{\hat{\sigma}^{2}}{n}}\Bigr],
\]
where $\hat{\sigma}^{2}$ is an estimator of $\sigma^{2}$, obtained
by substituting $\hat{\Delta}$ for $\Delta$ in the variance formula.
The asymptotic normality of a doubly robust version of $\Delta$ was
also recently proven in \citet{testa2025doubly}.

If we relax the assumptions to include $\varphi^{dATE}>0$ zero norm
assumption, we can write the confidence interval in terms of the spectral
decomposition of $\mathcal{K}$ ($\|\mathcal{Z}\|_2=\sqrt{\sum_{i=1}^{T}\lambda_{i}Y_{i}^{2}}$,
with $\lambda_{i}$ as the eigenvalues of $\mathcal{K}$ and $Y_{i}\overset{\mathrm{iid}}{\sim}\mathcal{N}(0,1)$).
A $(1-\alpha)$ confidence interval (CI) for $\|\Delta\|_2$ is: 
\[
\text{CI: }\Biggl[\frac{\chi_{T,\alpha/2}\cdot\sqrt{\sum_{i=1}^{T}\lambda_{i}}}{\sqrt{n}},\;\frac{\chi_{T,1-\alpha/2}\cdot\sqrt{\sum_{i=1}^{T}\lambda_{i}}}{\sqrt{n}}\Biggr],
\]
where $\chi_{T,\alpha/2}$ and $\chi_{T,1-\alpha/2}$ are the lower
and upper $\alpha/2$ quantiles of the $\chi_{T}$ distribution. 
\end{rem}

\subsubsection{Outcomes with phase-shifts\label{subsec:Fisher--Rao-Metric}}

In this section, we are interested in equipping the space of outcomes $\left(\mathcal{F},\phi\right)$ with the flexibility of domain warping along the $x$-axis. One common way of facilitating this is if we endow the metric space $\left(\mathcal{F,\phi}\right)$ of the population level outcomes with the Fisher--Rao metric $\phi$ \citep{srivastava2011registration}. 

Let  
\[
\Gamma
\;=\;
\bigl\{\,
  \gamma:[0,1]\!\longrightarrow\![0,1]\;\big|\;
  \gamma(0)=0,\;
  \gamma(1)=1,\;
  \dot\gamma>0
\bigr\}
\]
be the space of \emph{time-warping} actions. For any absolutely continuous curve
$f:[0,1]\!\to\!\mathbb R$ and $\gamma\in\Gamma$ we write $f\circ\gamma$ for the warped trajectory; note that more specifically we are concerned with curves $f\in W^{k,2}([0,1],\mathbb{R})$ in Sobolev space. Fix the outcome space: 
\[
  \mathcal F
  \;=\;
  \bigl\{
      f\in AC[0,1]\;:\;
      \dot f\in L_1[0,1],\;
      \dot f(t)\neq0 \text{ a.e.}
    \bigr\},
\]
and denote by \(T_f\mathcal F\) the linear space of first-–order perturbations \(\eta:[0,1]\!\to\!\mathbb R\) with \(\int_0^1|\eta(t)|\,dt<\infty\). Throughout, a dot \(\dot{\;}\) indicates the weak derivative.

For \(\eta_1,\eta_2\in T_f\mathcal F\) define the Fisher--Rao metric between two continuous functions
\begin{equation}
  g_{FR,f}(\eta_1,\eta_2)
  \;=\;
  \int_{0}^{1}
    \frac{\dot\eta_1(t)\,\dot\eta_2(t)}
         {\,|\dot f(t)|\,}\;dt ,
  \qquad f\in\mathcal F .
  \label{eq:elastic-FR}
\end{equation}
This inner product generates a Riemannian metric
\((\mathcal F,g_{FR})\) that is \emph{invariant to time-warpings}:
\(
  g_{FR,f\circ\gamma}(\eta_1\!\circ\!\gamma,\eta_2\!\circ\!\gamma)
  = g_{FR,f}(\eta_1,\eta_2)
\)
for all \(\gamma\in\Gamma\).  Hence distances computed with
\(g_{FR}\) depend only on the equivalence class
\([f]=\{f\circ\gamma:\gamma\in\Gamma\}\). While this is a complex Riemannian metric, its computation and induced geometry is vastly simplified by defining the \emph{square-root slope
function (SRSF)} transformation:
\[
   Q:\mathcal F\longrightarrow L_2([0,1]),
   \qquad
   Q(f)(t)
   = q_f(t)
   = \mathrm{sgn}\!\bigl(\dot f(t)\bigr)\,
     \sqrt{\,|\dot f(t)|\,}.
\]
A key fact is that the mapping \(f\mapsto q_f\) is an isometry \emph{from the Fisher--Rao Riemannian manifold \((\mathcal F_{\!FR},g_{FR})\) onto the flat Hilbert space \(L_2([0,1])\)} \citep{srivastava2016functional}. 
Consequently the geodesic distance inherited from
\eqref{eq:elastic-FR} reduces to the plain \(L_2\) norm
in the SRSF domain:
\[
   d_{FR}(f_1,f_2)
   \;=\;
   \bigl\|\,q_{f_1}-q_{f_2}\bigr\|_{L_2([0,1])}.
\]

Assuming the Fisher-Rao geometry, the set of phase-equivalence classes
carries the metric
\[
  (\mathcal F/\Gamma,\;d_{FR}),
  \qquad
  d_{FR}\bigl([f_1],[f_2]\bigr)
  =\inf_{\gamma\in\Gamma}\,
     \bigl\|\,q_{f_1}-q_{f_2\circ\gamma}\bigr\|_{L_2([0,1])}.
\]

The quotient space \(\mathcal F/\Gamma\) is \emph{no longer a Hilbert (nor
even a complete) space}: orbits may accumulate without converging to a
valid orbit (Appendix \ref{app:incompleteness}). As a result: (i) the squared distance functional $f\mapsto\sum_{i}w_{x,i}\, d_{FR}^2\bigl([f],[Y_i]\bigr)$ is a \emph{weighted least squares} problem on a non-linear, non-complete manifold; (ii) uniqueness of its minimiser (i.e. which is a weighted Fréchet mean) and consistency of the empirical estimator $\hat{F}(\hat{\bm{Y}}^{(x)}_1, \hat{\bm{Y}}^{(x)}_2, \ldots, \hat{\bm{Y}}^{(x)}_n)$ can no longer be taken for granted.

In practice, we obtain a unique solution by: (i) working in the SRSF domain \(L_2([0,1])\),
(ii) enforcing a constraint on estimated time-warpings (e.g.\ fixing its value at certain locations, or using a centred “Karcher mean’’ alignment), and (iii) verifying numerically that the optimisation stays inside a
geodesically convex neighbourhood where the objective is strictly
convex. This amounts to selecting a \emph{slice} from the quotient space \(\mathcal{F}/\Gamma\), such as: (i) fixing the phase value corresponding to the first local maximum or other landmark feature of the functions (e.g., identified by a subject-matter expert), or (ii) aligning to a prespecified template function (or the Karcher mean function) with the constraint that the set of time warpings is centered at the identity warping ($\gamma_{id}(t)=t$). The image of such a slicing rule is a subset of \(L_2([0,1])\) that inherits Hilbert structure. All Fréchet mean computation can then be carried out in this space, regaining the usual convexity and \(n^{1/2}\)-consistency results-at the price of introducing a preprocessing alignment step. The theoretical analysis of such constrained Fréchet means is beyond our present scope, but see \citet{srivastava2016functional,needham2020simplifying} for conditions
under which uniqueness and \(n^{1/2}\)-consistency are recovered.

In Appendix \ref{sec:square-root-slope-fucntions}, we explore strategies for restricting either the outcome space (Appendix \ref{app:restricting-outcome-space}) $\mathcal{F}$ or the group action $\Gamma$ (Appendix \ref{app:hilbert-vs-group-actions}) to recover theoretical guarantees such as consistency and uniqueness. However, we show that such restrictions can be overly stringent and may limit the practical utility of the Fisher-Rao framework in real data applications. 
 
\section{Kernel-based Inference of Causal Effects}

\label{sec:dynamic_treatment_func_outcome} 

We extend kernel-based estimation \citep{singh2020kernel} to settings with path-valued outcomes, formulating closed-form estimators for dynamic causal effects. These estimators efficiently capture dependencies across time by operating on structured output spaces and avoid the need for explicit pointwise propensity score estimation. As shown by \citet{singh2020kernel}, embedding the mapping from treatments and covariates to outcomes in an RKHS ensures boundedness and resolves key issues in nonparametric causal estimation with continuous treatments. We build on this framework by constructing kernel estimators for dynamic treatment effects and their dose-response extensions, allowing for continuous treatments $x \in \mathbb{R}^+$ and functional outcomes or covariates. To generalize further, we adopt \emph{operator-valued kernels} \citep{kadri2016operator}, which support direct learning in settings with function-valued inputs and outputs. Additionally, we integrate tools from \emph{elastic functional data analysis} to align functional trajectories and address phase variability across subjects, plus improve downstream causal inference. Our framework unifies both finite-dimensional ($\mathbb{R}^T$) and infinite-dimensional (e.g., $L^2([0,1])$ or $W^{k,2}$) outcome representations under a common RKHS-based causal estimation pipeline.

\paragraph{Phase-shifts via alignment assumption.}
Throughout, we assume that for outcomes with flexibility of domain warping (e.g. equipped with Fisher-Rao metric $\phi$), each of the trajectories $\bm{Y}_i$ has been first \emph{phase-aligned} to a fixed template  prior to kernel learning (i.e. in the Fisher-Rao case via square root slope transform as described in Appendix \ref{sec:square-root-slope-fucntions}). 

Formally, let $S : \mathcal{F} \to W^{k,2}([0,1])$
be a measurable slicing map that selects a canonical representative $S(f)$ from each warping equivalence class $[f] \in \mathcal{F} / \Gamma$. The aligned observations are then
\[
\bm{Y}_i^{\mathrm{align}} := S(\bm{Y}_i) \in \mathcal{H}_{\mathcal{Y}} := W^{k,2}([0,1]),
\]
ensuring that all kernel-based learning is performed in a Hilbert space with well-behaved geometric properties.

\begin{rem}
Attempting to apply kernel ridge regression directly on the quotient space $(\mathcal{F}/\Gamma, d_{FR})$ is problematic: the space is not Hilbert and generally not even complete (see Appendix~\ref{app:incompleteness}). As a result, key properties such as the Representer theorem, uniqueness of ridge-regularized solutions, and consistency (Theorem~\ref{thm:IPW-interpolant-consistency}) may fail. By aligning trajectories and working in $W^{k,2}$, we preserve the Hilbert structure required for standard kernel methods. Developing kernel learning techniques that operate directly on $\mathcal{F}/\Gamma$ remains an open problem.
\end{rem}

\subsection{Vector-Valued Kernel Estimators}

\label{subsec:kernel-estimators-functional}

\citet{singh2020kernel} show that by defining the relevant causal functionals in a reproducing kernel Hilbert space (thereby ensuring their boundedness and resolving technical issues with continuous treatments), the estimation of causal parameters such as $\varphi^{\text{ATE}}, \varphi^{\text{CATE}},$ and $\varphi^{\text{DS}}$ can indeed be cast as a nonparametric RKHS regression problem for scalar outcomes and vector covariates. One of the main practical advantages of this approach is that it does not require explicit estimation of the propensity function  $\pi(\bm{V}) = p(X = x \mid \bm{V} = \bm{v})$. In the continuous-treatment setting, this function is sometimes called 
the \emph{generalized propensity score}, and while it can be estimated in the literature 
(see, e.g., \citealp{imai2004causal}), 
it is subject to potential model misspecification and can be  difficult to estimate accurately, especially if one does not impose strong parametric assumptions.

Denote the reproducing kernel Hilbert spaces (RKHS) $\mathcal{H}_{\mathcal{X}}$ and $\mathcal{H}_{\mathcal{V}}$ associated with kernels $k_{\mathcal{X}}$ and
$k_{\mathcal{V}}$ respectively. Define the
feature maps:

\[
\begin{aligned}\psi_{\mathcal{X}}: & \ \mathcal{X}\rightarrow\mathcal{H}_{\mathcal{X}},\quad x_{i}\mapsto\psi_{\mathcal{X}}(x_{i}),\\
\psi_{\mathcal{V}}: & \ \mathcal{V}\rightarrow\mathcal{H}_{\mathcal{V}},\quad\bm{v}_{i}\mapsto\psi_{\mathcal{V}}(\bm{v}_{i}).
\end{aligned}
\]

The feature maps $\psi_{\mathcal{X}}$ and  $\psi_{\mathcal{Y}}$ ``collect'' the points $x_{i}$ and $\bm{v}_{i}$ mapping them from discrete sample points into their corresponding RKHS. If we denote the true regression $f\in\mathcal{H}_{\mathcal{X}\times\mathcal{V}}$ of the expected outcomes given treatment and covariates defined on $\mathcal{X}\times\mathcal{V}$, by the reproducing property, for any $(x,\bm{v}) \in \mathcal{X} \times \mathcal{V}$,
\begin{equation}\label{eq:conditional-expectation-vector}
    f(x, \bm{v}) 
    \;=\; \langle f,\; \psi_{\mathcal{X}}(x) \,\otimes\, \psi_{\mathcal{V}}( \bm{v})\rangle_{\mathcal{H}_{\mathcal{X}\times\mathcal{V}}},
\end{equation}
where $\otimes$ denotes the tensor (or Kronecker) product, and 
$\langle \cdot,\cdot\rangle_{\mathcal{H}_{\mathcal{X}\times\mathcal{V}}}$ denotes the inner product in the RKHS $\mathcal{H}_{\mathcal{X}\times\mathcal{V}}$. Analogously, we can estimate the expected potential outcome $\mathbb{E}[Y^{(x)} \mid X = x, \bm{V} = \bm{v}]$ 
nonparametrically: if we assume the usual \emph{no unmeasured confounding}, the typical integral involved in estimating $\mathbb{E}[Y^{(x)} \mid X = x, \bm{V} = \bm{v}]$ (i.e.,  \eqref{eq:back-door-criterion}) takes the form of the inner product:
\begin{equation}
\varphi(x,\bm{v}) = \langle \varphi, \psi_{\mathcal{X}}(x)\otimes\mu_{\bm{v}} \rangle_{\mathcal{H}},\text{ with }\mu_{\bm{v}}=\int\psi_{\mathcal{V}}(\bm{v})dP_{\bm{V}}(\bm{v}).
\end{equation}
For sub-population specific (i.e., potential outcome for previously exposed to strata ox $x$) or conditional (i.e.,  strata of $\bm{v}$ determined by additional covariates) causal effects, one would change the integrating measure $P_{\bm{V}}(\bm{v})$ from a marginal over the covariates to an appropriate conditional.

Building on this framework, consider the outcomes are discretized samples from functions
representing the outcomes $\bm{Y}=\left(\bm{Y}(u_{1}),\cdots,\bm{Y}(u_{T})\right)\in\mathbb{R}^{T}$. Following this, we adopt the convention: 
\[
\bm{Y}=\begin{pmatrix}\bm{Y}_{1}(u_{1}) & \cdots & \bm{Y}_{1}(u_{T})\\
\vdots &  & \vdots\\
\bm{Y}_{n}(u_{1}) & \cdots & \bm{Y}_{n}(u_{T})
\end{pmatrix}\in\mathbb{R}^{n\times T},\quad\text{vec}(\bm{Y})=\begin{pmatrix}\bm{Y}_{1}(u_{1})\\
\vdots\\
\bm{Y}_{1}(u_{T})\\
\vdots\\
\bm{Y}_{n}(u_{1})\\
\vdots\\
\bm{Y}_{n}(u_{T})
\end{pmatrix}\in\mathbb{R}^{nT\times1}.
\]

Although this effectively treats each functional curve $\bm{Y}_i$ as a $T$-dimensional vector 
(which can have drawbacks; see \citealp{ramsay2005}), 
it is conceptually simple and directly compatible with standard RKHS techniques 
for vector-valued kernels. Let denote $k_{\mathcal{X}\times{\mathcal{V}}}: (\mathcal{X}\times\mathcal{V}) \times (\mathcal{X}\times\mathcal{V}) \to \mathbb{R}$ the separable kernel defined by $k_{\mathcal{X}\times\mathcal{V}}\bigl((x,v),(x',v')\bigr) 
    \;=\; 
    k_{\mathcal{X}}(x,x')\,k_{\mathcal{V}}(v,v')$.  
Denote by $\mathbf{K}_{XV}\in\mathbb{R}^{n\times n}$ the associated Gram matrix
evaluated on the training data 
$\{(x_i,v_i)\}_{i=1}^n$ with $ (\mathbf{K}_{XV})_{ij} 
    \;=\; 
    k\bigl((x_i,v_i),(x_j,v_j)\bigr)$. To capture dependencies among the $T$ outputs in each functional curve,  let $\mathbf{K}_\mathcal{Y}\in\mathbb{R}^{T\times T}$ encode  the correlation structure among measurement points (e.g., time dependence $u_1,\dots,u_T$). 
We then form a Kronecker-structured kernel on $\mathcal{X}\times\mathcal{V}$ 
with multi-dimensional output:
\[
    \mathbf{K}_{XX} 
    \;=\; 
    \mathbf{K}_{XV}\;\otimes\;\mathbf{K}_\mathcal{Y} 
    \;\in\; 
    \mathbb{R}^{nT\times nT},
\]
see \cite{luo2022nonparametric} for kernel methods designed to incorporate more advanced between-curve dependence. Assuming we have $n$ samples $\{\bm{v}_{i}\}_{i=1}^{n}$ from the distribution
of $\bm{V}$, we can approximate the point-wise expected potential outcomes $\mathbb{E}\{\bm{Y}^{(x)} \lvert \bm{V}\}$ via averaging over $\hat{\varphi}(x,\bm{v})$: 
\begin{equation}
\hat{\varphi}(x) = \frac{1}{n}\sum_{i=1}^{n}\hat{\varphi}(x,\bm{v}_i) =
\frac{1}{n}\sum_{i=1}^{n} \mathbf{K}_{(x,\bm{v}_{i})X} \left(\mathbf{K}_{XX}+\lambda\mathbf{I}_{nT}\right)^{-1} \text{vec}(\bm{Y}).
\label{eq:approximate_potential_outcome}
\end{equation}
Alternatively, we can also represent the expected potential outcome estimator over average covariate values in terms of the average kernel: 
\begin{equation}
\hat{\varphi}(x)=\bar{\mathbf{K}}_{xX}\left(\mathbf{K}_{XX}+\lambda\mathbf{I}_{nT}\right)^{-1}\text{vec}(\bm{Y}),\label{eq:average_kernel}
\end{equation}
where $\bar{\mathbf{K}}_{xX}=\frac{1}{n}\sum_{i=1}^{n}\mathbf{K}_{(x,v_{i})X}$ and 

\[
    \mathbf{K}_{(x,v),X} 
    \;=\; 
    \underbrace{\bigl[k\bigl((x,v),(x_1,v_1)\bigr),\dots,k\bigl((x,v),(x_n,v_n)\bigr)\bigr]}_{\in \mathbb{R}^{1\times n}}
    \;\otimes\;
    \underbrace{\mathbf{K}_\mathcal{Y}}_{\in \mathbb{R}^{T\times T}}
    \;\in\; \mathbb{R}^{T\times nT}.
\]
Given a new input \((x, \bm{v})\), the predicted \(T\)-dimensional response is
\[
\hat{\varphi}(x, \bm{v}_j)^{new} = \mathbf{K}_{(x,\bm{v}_j),X} \, \hat{\varphi}(x) \in \mathbb{R}^{T}.
\]
This construction exploits both input similarity (via $\mathbf{K}_{XV}$) 
and the output dependence structure (via $\mathbf{K}_\mathcal{Y}$) 
in a unified kernel ridge regression framework. In Bayesian extensions, one may require additional conditions on the kernels (e.g., nuclear dominance \citep{chau2021bayesimp}, non-stationary kernels \citep{noack2024unifying} or low-rank assumptions \citep{luo2022sparse}) 
or spectral representations of the Hilbert space \citep{dance2024spectral}.

Recall that $\hat{\varphi}(x)$ becomes our kernel estimator for \eqref{eq:Frechet_mean-1} from above. To recover expressions for estimators for the point-wise binary treatment effect $\hat{\Delta}(t)$, the scalar $\hat{\varphi}^{dATE}$ and the continuous treatment extension $\hat{\varphi}^{dDS}(x)$, we compute:
\begin{align}
\label{eq:kernel_estimators_effects}
\hat{\Delta}(t) &= \hat{\varphi}(1)(t)  - \hat{\varphi}(0)(t), \; t=1,\dots,T \quad \text{(Pointwise binary effect)}\\
\hat{\varphi}^{dATE} &= \lVert \hat{\varphi}(1) - \hat{\varphi}(0) \rVert_{\phi}\quad \text{(Dynamic binary treatment effect)}\\
\hat{\varphi}^{dDS} &= \lVert\hat{\varphi}(x)\rVert_{\phi}\quad \text{(Dynamic dose response)}\label{eq:kernel_estimators_effects-DS}
\end{align}

\subsection{Operator-Valued Kernel Estimators}
\label{subsec:operator-valued-kernels}

In the previous subsection, we discretized the functional outcomes into vectors in $\mathbb{R}^T$ and employed a multi-output (matrix-valued) kernel. Consequently, the associated RKHS contains maps $f:\mathcal{X} \to \mathbb{R}^T$; each $f(x)$ yields a $T$-dimensional vector whose coordinates correspond \emph{only} to the pre-selected grid $\{t_1, \dots, t_T\}$.  
While one can interpolate these vectors after the fact, such a step lies outside the native representation: the discrete model cannot \emph{directly} produce the causal effect at an unseen time point $t^* \notin \{t_1, \dots, t_T\}$. In other words, its limitation is not generalization in the covariate direction $x$, but rather in \emph{representing} the full outcome trajectory.

To overcome this, we next consider \emph{operator-valued kernels}~\cite{kadri2016operator}. These kernels map inputs into bounded linear operators on a Hilbert space $\mathcal{H}$ of time-indexed functions, thereby extending real- and matrix-valued kernels to the fully functional setting and giving an RKHS whose elements $f: \mathcal{X} \to \mathcal{H}$ can be evaluated at every $t \in \mathcal{T}$ without any additional interpolation scheme.

Let \(\bm{Y}_i\) denote an observed outcome curve. After interpolation we regard each curve as an element of the \emph{separable Hilbert space} $\mathcal H_{\mathcal Y} \;:=\; W^{k,2}\!\bigl([0,1],\mathbb R\bigr)
  \quad\text{with } k \ge 1$,
equipped with the Sobolev inner product
$\langle\cdot,\cdot\rangle_{k,2}$. We define
\[
  \mathcal K :
  (\mathcal X \times \mathcal H_{\mathcal V})
  \times
  (\mathcal X \times \mathcal H_{\mathcal V})
  \;\longrightarrow\;
  \mathcal L(\mathcal H_{\mathcal Y}),
\]
where $\mathcal L(\mathcal H_{\mathcal Y})$ denotes the bounded linear operators on $\mathcal H_{\mathcal Y}$. For inputs $i$ and $j$, we consider the separable (tensor-product) operator-valued kernel
\begin{align}\label{eq:operator-valued-kernel-example}
\mathcal{K}\left((X_{i},\bm{V}_{i}),(X_{j},\bm{V}_{j})\right) 
&= k_{\mathcal{X}}\bigl(X_{i}, X_{j}\bigr) \cdot 
   k^*_{\mathcal{V}}\bigl(\bm{V}_{i}, \bm{V}_{j})\bigr) \cdot k_{\mathcal{H}_{\mathcal{Y}}}, \notag\\
&= k_{\mathcal{X}}\bigl(X_{i}, X_{j})\bigr) \cdot 
   k^*_{\mathcal{V}}\bigl(\bm{V}_{i}, \bm{V}_{j})\bigr) \cdot I_{\mathcal{H}_{\mathcal{Y}}}.
\end{align}

where $k_{\mathcal X} : \mathcal X \times \mathcal X \to \mathbb{R}$ is a positive-definite scalar kernel on the treatment space (as distinct from Section~\ref{subsec:kernel-estimators-functional}); $k^*_{\mathcal V} : \mathcal H_{\mathcal V} \times \mathcal H_{\mathcal V} \to \mathbb{R}$ is a positive-definite kernel on the embedded covariate space\footnote{If raw covariates $\bm{V}_i \in \mathcal{V}$ are not already elements of $\mathcal H_{\mathcal V}$, we embed them via a feature map $\psi_{\mathcal V} : \mathcal V \to \mathcal H_{\mathcal V}$ and write $\bm{v}_i = \psi_{\mathcal V}(\bm V_i)$.}; $I_{\mathcal H_{\mathcal Y}}$ is the identity operator on $\mathcal H_{\mathcal Y}$, ensuring the kernel values lie in $\mathcal L(\mathcal H_{\mathcal Y})$.

Because $k_{\mathcal X} \cdot k^*_{\mathcal V}$ is positive‑definite and $I_{\mathcal H_{\mathcal Y}}$ is self‑adjoint, the kernel $\mathcal K$ is positive‑definite in the sense of operator‑valued kernels \citep{micchelli2005learning}. It therefore induces a vector-valued RKHS of functions $(x, \bm{v}) \mapsto f_{(x, \bm{v})} \in \mathcal H_{\mathcal Y}$, which we use for operator-valued kernel ridge regression. 


The regression of the conditional expectation \(\mathbb{E}[\bm{Y} \mid X,\bm{V}]\) from \eqref{eq:conditional-expectation-vector} naturally extends to the operator-valued RKHS setting. By the reproducing property of an operator-valued RKHS \(\mathcal{H}_{\mathcal{K}}^{\text{op}}\), the true regression function can be written as: 
\[
    m(x,\bm{v})
    \;=\;
    \left\langle\,
      m,\;
      \Psi\bigl((x,\bm{v})\bigr)
    \right\rangle_{\!\mathcal{H}_{\mathcal{K}}^{\text{op}}},
    \quad
    \text{for } \Psi : \mathcal{X} \times \mathcal{H}_{\mathcal{V}} \to \mathcal{H}_{\mathcal{K}}^{\text{op}} \text{ the operator-valued feature map}.
\]
Following the same ridge regression objective as in Section~\ref{subsec:kernel-estimators-functional}, we estimate the conditional expectation operator 
\begin{equation}
m : \mathcal{X} \times \mathcal{H}_{\mathcal{V}} \rightarrow \mathcal{H}_{\mathcal{Y}}, \qquad (x,\bm{v}) \mapsto \mathbb{E}[\bm{Y} \mid X = x, \bm{V} = \bm{v}],
\end{equation}
by solving the regularized least squares problem:
\begin{equation}
\min_{m \in \mathcal{H}_{\mathcal{K}}^{\text{op}}} \sum_{i=1}^{n} \left\| \bm{Y}_{i} - m(X_i, \bm{v}_i) \right\|_{\mathcal{H}_{\mathcal{Y}}}^2 + \lambda \left\| m \right\|_{\mathcal{H}_{\mathcal{K}}^{\text{op}}}^2,
\end{equation}
where $\bm{v}_i := \psi_{\mathcal{V}}(\bm{V}_i)$ denoting the appropriately interpolated to $\mathcal{H}_{\mathcal{V}}$ covariate and \(\lambda > 0\) is a regularization parameter. By the Representer Theorem for operator-valued kernels, the minimizer admits the form:
\[
m(\cdot, \cdot) = \sum_{j=1}^{n} \mathcal{K} \bigl((\cdot, \cdot), (X_j, \bm{v}_j)\bigr) \, \alpha_j,
\]
for coefficients \(\alpha_j \in \mathcal{H}_{\mathcal{Y}}\). Substituting the separable kernel from \eqref{eq:operator-valued-kernel-example}, the estimator becomes:
\[
m(\cdot, \cdot) = \sum_{j=1}^{n} k_{\mathcal{X}}(\cdot, X_j) \cdot k^*_{\mathcal{V}}(\cdot, \bm{v}_j) \cdot \alpha_j.
\]
Letting \(\bm{\alpha} = [\alpha_1, \dots, \alpha_n]^\top\) and defining the kernel matrix \(\mathcal{K}_{ij} := k_{\mathcal{X}}(X_i, X_j) \cdot k^*_{\mathcal{V}}(\bm{v}_i, \bm{v}_j)\), the optimization reduces to:
\begin{equation}
\min_{\bm{\alpha}} \sum_{i=1}^{n} \left\| \bm{Y}_i - \sum_{j=1}^{n} \mathcal{K}_{ij} \, \alpha_j \right\|_{\mathcal{H}_{\mathcal{Y}}}^2 + \lambda \sum_{i,j=1}^{n} \left\langle \alpha_i, \mathcal{K}_{ij} \, \alpha_j \right\rangle_{\mathcal{H}_{\mathcal{Y}}}.
\end{equation}

We can summarise our estimators for potential outcomes in the case where the outcomes lie in a \emph{separable Hilbert space} $\bigl(\mathcal{F},\langle\!\cdot,\cdot\!\rangle_{\mathcal{F}}\bigr)$ (e.g. a Sobolev space $W^{k,2}([0,1])$) with the induced metric $\phi(f,g)=\|f-g\|_{\mathcal{F}}$, and treatments $X\in\mathbb{R}$, as an operator-valued map
\[
  \hat{\varphi}(x):\mathbb{R}\longrightarrow\mathcal{F}.
\]
This map integrates over covariates $\bm{V}$ to yield an estimate of the (average) dose–response at treatment level $x$:
\begin{align}
\hat{\varphi}(x)
&= \int_{\mathcal{V}} F\bigl(\bm{Y}\mid\bm{V}= \bm{v},\,X=x\bigr)\;dP_{\bm{V}}(\bm{v}) \notag\\
&= \sum_{j=1}^{n}
   \mathcal{K}\!\bigl(
     (\psi_{\mathcal{X}}(x),\psi_{\mathcal{V}}(\bm{V}_j)),
     (\psi_{\mathcal{X}}(X_j),\psi_{\mathcal{V}}(\bm{V}_j))
   \bigr)
   \bigl(\mathcal{K}+\lambda I_n\bigr)^{-1}\bm{Y},
\label{eq:operator-value-estimator}
\end{align}
where $\psi_{\mathcal{X}}(\cdot)$ and $\psi_{\mathcal{V}}(\cdot)$ are the feature maps induced by kernels $k_{\mathcal{X}}$ and $k_{\mathcal{V}}$, respectively. Here $P_{\bm{V}}$ is the distribution over the covariate domain $\mathcal{V}$, and the sum corresponds to the kernel ridge solution in the operator-valued RKHS induced by $\mathcal{K}$. The use of kernel regression allows us to approximate the integral over $P_{\bm{V}}$ via a data-adaptive weighted average, without requiring explicit knowledge of the generalized propensity score or covariate density; this parallels the balancing weight estimator from \eqref{eq:Frechet_mean-2}. The estimator \eqref{eq:operator-value-estimator} can then be substituted into the expressions for dynamic binary treatment effects and dose response (e.g.\ \eqref{eq:kernel_estimators_effects}), generalizing those estimators to the fully functional setting.

\begin{rem}[Fisher-Rao kernel]
When phase variability remains after the global slice $S$, one can inject additional \emph{warp-invariance} directly into the operator-valued kernel by replacing the output factor
$I_{\mathcal H_{\mathcal Y}}$ in \eqref{eq:operator-valued-kernel-example} with a Fisher–Rao similarity:

\[
k_{\mathrm{FR}}(f,g)
  \;=\;
  \exp\!\bigl\{-\zeta\,d_{\mathrm{FR}}(f,g)^{2}\bigr\},
  \quad
  \zeta>0,
\qquad
k_{\mathcal H_{\mathcal Y}}(\bm Y_i,\bm Y_j)
  = k_{\mathrm{FR}}(\bm Y_i,\bm Y_j)\,I_{\mathcal H_{\mathcal Y}}.
\]
$k_{\mathrm{FR}}$ is positive-definite (proof follows immediately from \citet{schoenberg1938metric}, included for completeness in Appendix~\ref{subsec:proof-operator-kernel}), so the resulting operator-valued kernel remains positive‑definite and kernel ridge regression is well posed. However, since $d_{\mathrm{FR}}$ is defined on the quotient $\mathcal F/\Gamma$, the induced RKHS lives on a \emph{non-Hilbert} manifold. Solutions of the ridge problem therefore exist but need not be unique; standard Representer theorem guarantees
require the extra alignment step discussed above. Algorithm~\ref{alg:Iterative-Kernelized-Causal}
implements an iterative scheme (producing an estimator $\hat{\varphi}^{dATE}$ that enjoys existence and empirical stability) that alternates
\emph{(i)} SRSF alignment of covariates and outcomes and
\emph{(ii)} operator‑valued kernel ridge regression with the
$k_{\mathrm{FR}}$ factor.  

\end{rem}

\paragraph{Uniform consistency and rates.}
Let  
$k_{\!\mathcal X}:\mathcal X\times\mathcal X\!\to\!\mathbb R$
and  
$k_{\!\mathcal V}:\mathcal V\times\mathcal V\!\to\!\mathbb R$
be continuous, bounded, positive-definite kernels with associated scalar RKHSs $\mathcal{H}_{\mathcal{X}},\mathcal{H}_{\mathcal{V}}$. We work with the product kernel  
$k\!=\!k_{\!\mathcal X}\cdot\,k_{\!\mathcal V}$  
on
$\mathcal X\times\mathcal V$  
and its RKHS  
$\mathcal{H}:=\mathcal{H}_{\mathcal X}\otimes \mathcal{H}_{\mathcal V}$.
Denote by $\zeta_j(\mathcal{H})$  
the non-increasing eigenvalues of the corresponding
convolution operator.  
\begin{assumption}[Smoothness]
\label{ass:input-regularity-smoothness}
The true regression operator  
\(m_0\)  
satisfies the source condition  
\(m_0\in \mathcal{H}^{\,c}\)  
for some  
\(c\in(1,2]\).
\end{assumption}
\begin{assumption}[Spectral decay]
\label{ass:input-regularity-spectral-decay}
There exists a constant $C > 0$ such that the eigenvalues $\{\zeta_j(\mathcal{H})\}_{j \ge 1}$ of the convolution operator associated with the kernel $k = k_{\!\mathcal{X}} \cdot k_{\!\mathcal{V}}$ satisfy the polynomial decay condition
\[
\zeta_j(\mathcal{H}) \le C\, j^{-b}
\quad\text{for all } j \ge 1,
\]
for some decay rate \(b \ge 1\).
\end{assumption}

\noindent
For the heterogeneous-effect estimators, we additionally impose
smoothness-decay conditions analogous to Assumption~\ref{ass:input-regularity-spectral-decay}
on the conditional mean-embedding operators
$E_1 : \mathcal{H}_{D} \to \mathcal{H}_{\mathcal X}$
and
$E_2 : \mathcal{H}_{\mathcal V} \to \mathcal{H}_{\mathcal X}$,
with respective exponents $(c_1,b_1)$ and $(c_2,b_2)$;
cf.\ \citet[Assumption 6.3]{singh2024kernel}.

On the \emph{output} side, the response trajectories are modelled in the
Sobolev Hilbert space
$\mathcal{H}_{\mathcal Y}=W^{k,2}([0,1]),\; k\ge 1,$
equipped with its canonical inner product
$\langle\!\cdot,\cdot\rangle_{k,2}$.
We choose the identity operator $I_{\mathcal H_{\mathcal Y}}$ as the output
kernel; hence no additional spectral assumption on
$\mathcal H_{\mathcal Y}$ is required.

Let $\lambda>0$ be the ridge regularisation parameter that enters our
operator-valued kernel ridge estimator (cf.\ ~\eqref{eq:operator-value-estimator}).
Under Assumptions~\ref{ass:input-regularity-smoothness}–
\ref{ass:input-regularity-spectral-decay} and their
mean-embedding counterparts, choose
\[
  \lambda
  \;=\;
  n^{-\frac{1}{c+1/b}},
  \qquad
  \lambda_\ell
  \;=\;
  n^{-\frac{1}{c_\ell+1/b_\ell}},
  \;\; \ell\in\{1,2\},
\]
(which coincides with the rule of
\citealp[Thm.\;6.1]{singh2024kernel})
yields the uniform rates
\[
  \bigl\|
    \widehat{\theta}_{\mathrm{ATE}}
    -
    \theta^{0}_{\mathrm{ATE}}
  \bigr\|_{\infty}
  \;=\;
  \mathcal O_{\mathbb P}\!\Bigl(
    n^{-\frac{c-1}{2(c+1/b)}}
  \Bigr),
  \qquad
  \text{and similarly for ATT, CATE, DS.}
\]
The operator-valued nature of the regression only changes constants (in the Hilbert-Schmidt norm) but \emph{not} the exponent; hence our method inherits the minimax‐optimal
$n^{-(c-1)/\{2(c+1/b)\}}$ behaviour established by \citet{singh2024kernel}.  
Uniqueness of the estimator is guaranteed because learning is performed after the alignment map
$S:\mathcal F\!\to\!W^{k,2}$ projects each trajectory to a Hilbert slice, where the Representer theorem and convexity apply; attempting to learn directly on the quotient space $\mathcal F/\Gamma$ would forfeit these
guarantees (as mentioned in Section~\ref{subsec:Fisher--Rao-Metric}).

\begin{algorithm}
\include{main_alg}

\caption{\label{alg:Iterative-Kernelized-Causal}Iterative SRVF based Kernel Causal Effect Estimation}
\end{algorithm}

\section{Experimental Results}
\label{sec:experiments}

To assess the impact of the proposed methods, we conducted synthetic experiments simulating functional data scenarios where traditional causal inference approaches struggle. These experiments illustrate how the estimators introduced in Sections \ref{sec:dynamic-outcomes-binary-treatment} and \ref{sec:dynamic_treatment_func_outcome}, including functional treatment effect estimators, alignment techniques, and operator-valued kernels, handle challenges posed by functional outcomes and covariates. We analyze two key scenarios: (1) binary treatment with functional outcomes and (2) continuous treatment with functional co-variates and outcomes which are temporally misaligned. 

\subsection{Binary Treatment with Synthetic Data}\label{subsec:Binary-treatment-synthetic}

The first experiment simulates functional outcomes $Y(t)$ generated as time-dependent curves influenced by a binary treatment $X \in \{0,1\}$ and baseline covariates $V$. We consider two scenarios:
\begin{itemize}
    \item Functional outcomes modeled as:
    \begin{equation}
    Y(t) = \mu_{0}(t) + \beta_{X}(t)X + \epsilon(t),
    \end{equation}
    where $\mu_{0}(t)$ represents the baseline curve, $\beta_{X}(t)$ is the time-varying treatment effect, and $\epsilon(t)$ is independent Gaussian noise.
    \item Monotonic functional outcomes defined as the cumulative sum of an underlying process $Z(t)$:
    \begin{equation}
    Z(t) = \mu_0(t) + \beta_X(t) X + \epsilon(t),
    \end{equation}
    leading to the observed outcome function:
    \begin{equation}
    Y(t) = \sum_{\tau=1}^{t} Z(\tau).
    \end{equation}
\end{itemize}

In both cases, the function $\beta_{X}(t) = \sum_{i=1}^{3} a \cdot \exp\left(-\frac{(t - c_i)^2}{2w^2}\right)$ is parameterized by $a$, $c_i$, and $w$, where $c_i$ are the centers of three equally spaced peaks, controlling the treatment effect’s amplitude, location, and spread. Both cases also simulate the challenge of temporal misalignment, where the outcome curves $Y(t)$ exhibit random shifts in their peak locations across individuals samples. Covariates $V$ are simulated to correlate with both $X$ and $Y(t)$, introducing realistic confounding structures. Theoretical results from Theorem \ref{thm:fr_dte_asymp_normality} guarantee the consistency of SRVF-based estimators in the monotonic outcome case. However, we empirically evaluate the performance of different estimators in both scenarios. We compare the following estimators: Inverse Probability Weighting (IPW) ATE estimator \cite{imai2004causal}; Doubly Robust ATE estimator; Kernel ATE estimator \cite{singh2020kernel}; proposed Operator Kernel ATE estimator \eqref{eq:kernel_estimators_effects}; SRVF-based Operator Kernel ATE estimator, which registers path outcomes to their Fréchet mean.

\begin{table}[h]
    \centering
    \begin{tabular}{lccc} 
        \toprule
        Method & $n=50$ & $n=100$ & $n=250$ \\
        \midrule
        IPW ATE & 22.70 (14.93) & 23.18 (15.31) & 22.64 (15.09) \\
        Doubly Robust ATE & 21.74 (14.84) & 21.98 (15.10) & 21.65 (14.88) \\
        Kernel ATE & 22.63 (15.25) & 23.14 (15.47) & 22.61 (15.10) \\
        Operator Kernel ATE & 15.18 (10.50) & 15.42 (10.70) & 13.32 (9.79) \\
        SRVF Operator Kernel ATE & 15.15 (10.32) & 15.60 (10.49) & 12.65 (9.60) \\
        \bottomrule
    \end{tabular}
    \caption{Mean absolute error and standard deviation (in brackets) of causal effect estimators for binary treatment and \textit{monotonic} path outcomes $\bm{Y}$.}
    \label{tab:methods_comparison_monotonic}
\end{table}

\begin{table}[h]
    \centering
    \begin{tabular}{lccc} 
        \toprule
        Method & $n=50$ & $n=100$ & $n=250$ \\
        \midrule
        IPW ATE & 0.88 (0.77) & 0.89 (0.77) & 0.89 (0.79) \\
        Doubly Robust ATE & 0.67 (0.67) & 0.68 (0.67) & 0.66 (0.70) \\
        Kernel ATE & 0.87 (0.78) & 0.89 (0.78) & 0.89 (0.79) \\
        Operator Kernel ATE & 0.65 (0.62) & 0.62 (0.60) & 0.62 (0.60) \\
        SRVF Operator Kernel ATE & 0.65 (0.61) & 0.64 (0.58) & 0.62 (0.57) \\
        \bottomrule
    \end{tabular}
    \caption{Mean absolute error and standard deviation (in brackets) of causal effect estimators for binary treatment and \textit{nonmonotonic} path outcomes $\bm{Y}$.}
    \label{tab:methods_comparison_nonmonotonic}
\end{table}

For each scenario, we simulate five datasets from the described super-population $\left\{\bm{Y}, X, \bm{V}\right\}$ and evaluate the estimators for sample sizes $n \in \{50, 100, 250\}$. We use squared exponential kernel for the covariate kernel (and outcome kernel in the operator setting), and binary kernel for the treatments as proposed originally in \cite{singh2020kernel}. The kernel parameters for the kernel estimators are set using the \textit{median inter-point heuristic} and the regularization terms are set using hyperparameter grid search optimizing the out-of-sample performance holding $20\%$ of the data for testing.  Figures \ref{fig:binary_box1} and \ref{fig:binary_box2} present box plots of ATE estimation errors across simulations, while Tables \ref{tab:methods_comparison_monotonic} and \ref{tab:methods_comparison_nonmonotonic} summarize mean absolute errors and standard deviations across the time grid. Dynamic estimation errors over time of the pointwise estimators $\Delta(t)$ are shown in Figures \ref{fig:binary_line1} and \ref{fig:binary_line2}, respectively. To estimate standard deviation across time, we first average across simulation datasets for each sample size. Our results indicate that the proposed estimators incorporating outcome structure achieve lower mean absolute errors and reduced standard deviations across the time grid. As expected, the doubly robust ATE achieves lower estimation error for smaller $n$, even without properly accounting for the outcome structure. Functional alignment further reduces error variance, particularly in more complex outcome scenarios, where it allows to flatten the estimation error across the time-grid.

\begin{figure}[h]
\centering
\includegraphics[width=16cm]{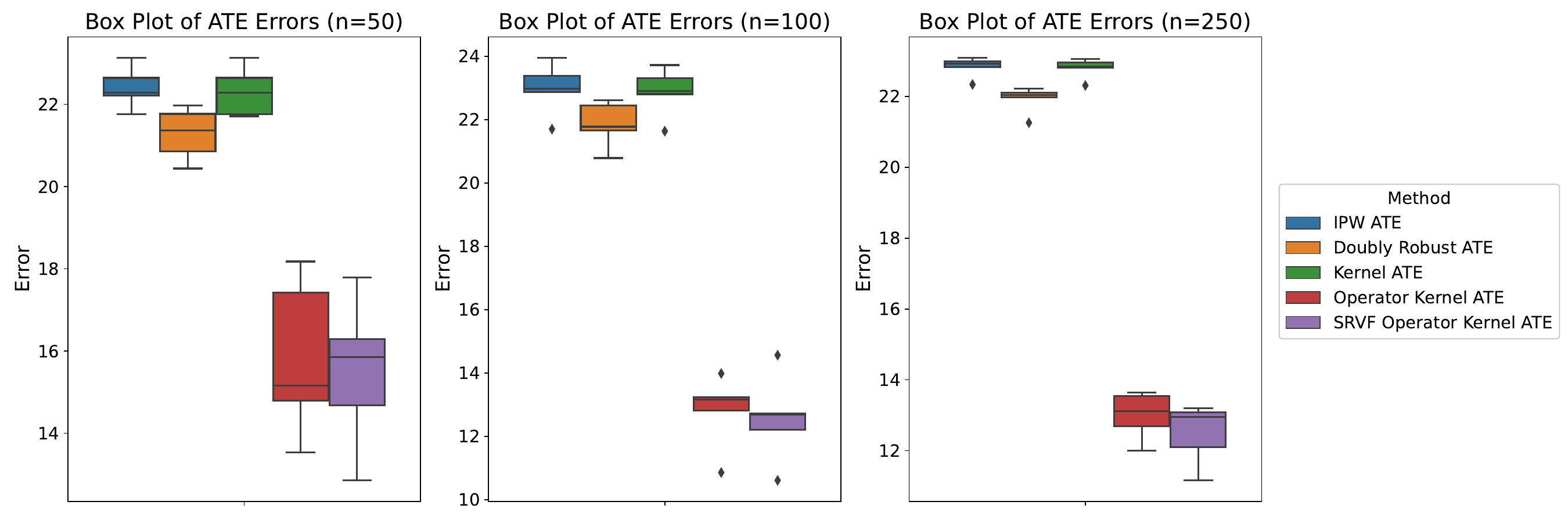}
\caption{Box plot of ATE estimation errors for \textit{monotonic} path-valued outcomes, across different training sample sizes ($n=50, 100, 250$). The IPW ATE, Doubly Robust ATE, and Kernel ATE estimators ignore the multivariate structure, whereas the SRVF-based approach accounts for amplitude and phase variability.}
\label{fig:binary_box1}
\end{figure}

\begin{figure}[h]
\centering
\includegraphics[width=16cm]{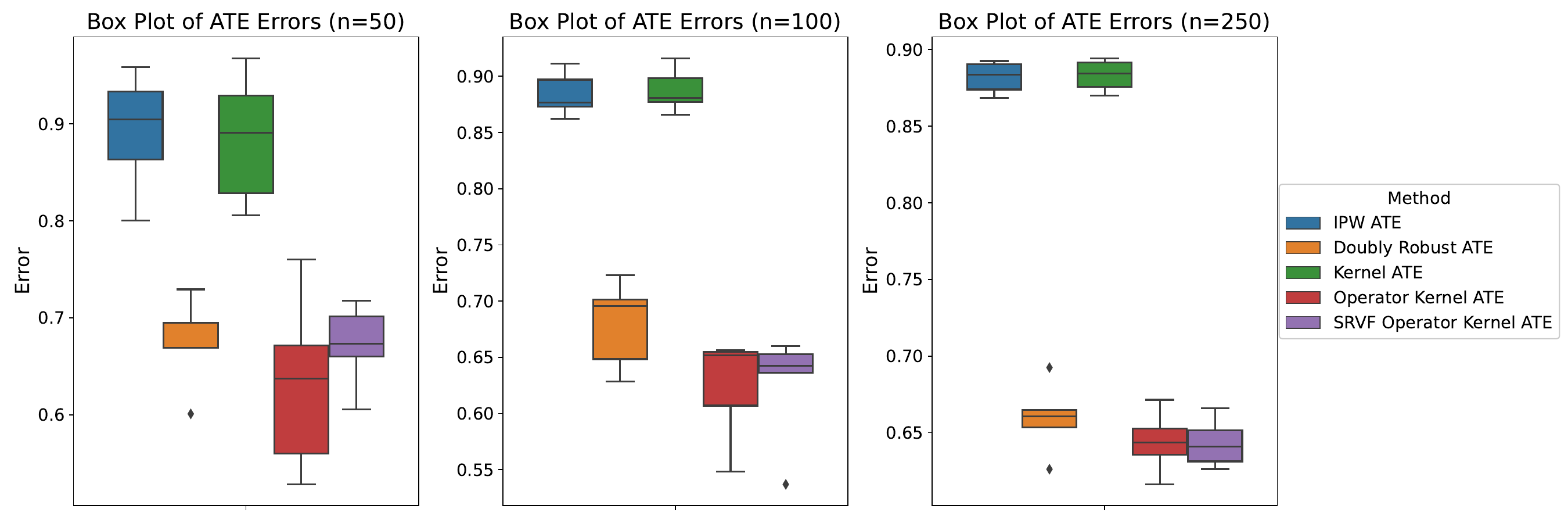}
\caption{Box plot of ATE estimation errors for \textit{nonmonotonic} path-valued outcomes, across different training sample sizes ($n=50, 100, 250$). The IPW ATE, Doubly Robust ATE, and Kernel ATE estimators ignore the multivariate structure, whereas the SRVF-based approach accounts for amplitude and phase variability.}
\label{fig:binary_box2}
\end{figure}

\begin{figure}[h]
\centering
\includegraphics[width=16.5cm]{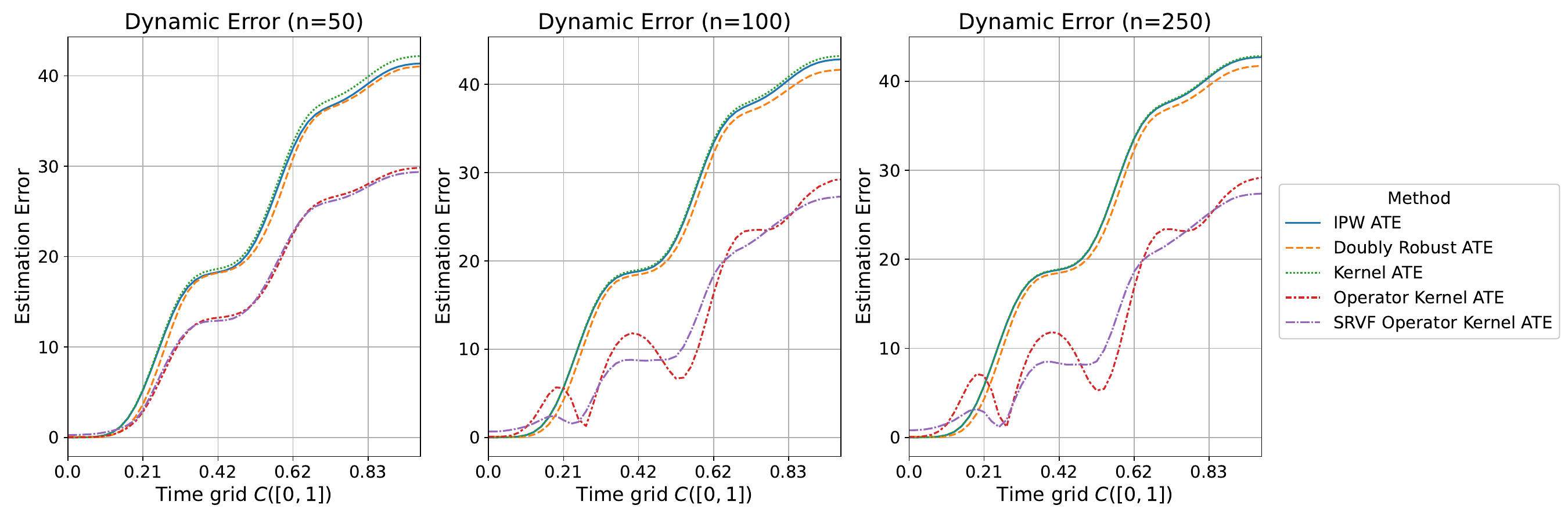}
\caption{Dynamic error as a function of time for different causal inference estimators for the ATE, assuming \textit{monotonic outcomes}. The plots display the average error (i.e., estimated across 5 draws from the super-population) for varying sample size $n$.}
\label{fig:binary_line1}
\end{figure}

\begin{figure}[h]
\centering
\includegraphics[width=16.5cm]{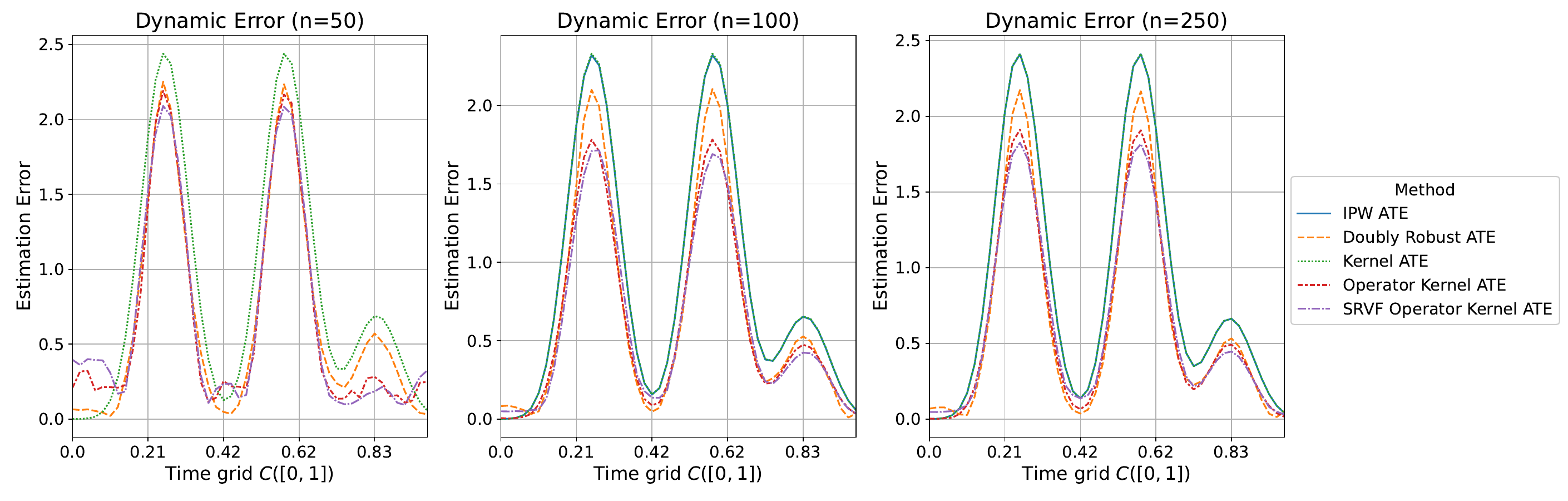}
\caption{Dynamic error as a function of time for different causal inference estimators for the ATE, assuming \textit{nonmonotonic outcomes}. The plots display the average error (i.e., estimated across 5 draws from the super-population) for varying sample size $n$.}
\label{fig:binary_line2}
\end{figure}
\subsection{Continuous Treatment with Synthetic Data  \label{subsection:Temporal-Misalignment}}

The next experiment simulates functional outcomes $Y(t)$ generated as time-dependent curves influenced by
a \textit{continuous} treatment $X\in\mathbb{R}$ and time-dependent baseline covariates $V(t)$. Kernel causal estimators \citep{singh2020kernel} and our proposed extensions from Section \ref{sec:dynamic_treatment_func_outcome} can readily deal with the continuous treatment $X\in\mathbb{R}$ by a simple change of kernel $k_{\mathcal{X}}$ describing the feature maps $\psi_{\mathcal{X}}$. We simulate functional outcomes as:
\begin{equation}
Y(t) = \mu_{V}(t) + \beta_{X}(t)X + \epsilon(t),
\end{equation}
where $\mu_{V}(t)$ represents a curve effect (arc parameterized with expected peak location and hight) dependent on the covariates, $\beta_{X}(t)$ is the time-varying treatment effect, and $\epsilon(t)$ is independent Gaussian noise. The function $\beta_{X}(t) = \sum_{i=1}^{3} a \cdot \exp\left(-\frac{(t - c_i)^2}{2w^2}\right)$ is again parameterized by $a$, $c_i$, and $w$, where $c_i$ are the centers of three equally spaced peaks controlling the treatment effect's amplitude, location, and spread. However, in this setup the outcomes are further modulated by time-dependent covariates $V(t)$ which introduce realistic time-dependent confounding. Independent random shifts are introduced in the peak locations for both $Y(t)$ and $V(t)$ across individual samples from the super-population $\{\bm{Y}, X, \bm{V}\}$. We compare the different kernel causal effect estimators (Kernel DS estimator as in \cite{singh2020kernel}, proposed Operator Kernel DS estimator \eqref{eq:kernel_estimators_effects-DS}, and the Iterative SRVF-based Kernel DS estimator) for their mean absolute error in estimating the dose-response effect in the setup of dynamic outcomes and covariates. Mimicking Section \ref{subsec:Binary-treatment-synthetic}, we simulate five datasets from the described super-population $\{\bm{Y}, X, \bm{V}\}$ for the scenarios of sample size $n=50$, $n=100$ and $n=250$. The squared exponential kernel parameters\footnote{Replacing the covariates kernel \( k_{\mathcal{V}} \) with kernels explicitly designed for sequential data (e.g., the signature kernel \cite{lee2023signature}) would be a sensible approach if we wish to capture more complex temporal features of \( \bm{V} \), but we defer this to future work due to the challenges in characterizing the approximation error, convergence properties, and statistical efficiency of using truncated signature features in high-dimensional settings.} used for covariates and treatment are set using the median inter-point heuristic with regularization terms set using hyperparameter grid search optimizing the out-of-sample
performance holding 20$\%$ of the data for testing. Figure \ref{fig:continuous_box1} present a box plot of the DS estimation error across simulations, while Table \ref{tab:methods_comparison_continuous_treatment} summarizes the mean absolute error and standard deviation across the time grid. Dynamic estimation errors over time of the pointwise estimators $\Delta(t)$ of the DS are shown in Figures \ref{fig:continuous_line1}. The standard deviation across time is computed over the average estimation error across the simulation datasets for selected sample size. 

Our results indicate that isolating and estimating the treatment effects in this more complex scenarios does require more data samples to converge to a robust estimator. The proposed iterative algorithm and operator-valued kernel approach both reduce the standard deviation of the estimation error across the time grid with (i.e.,  flattens the estimation error) and  
as we increase the number of samples from $\{\bm{Y}, X, \bm{V}\}$ which can be explained with the increase complexity of the effect function (i.e.,  continuous treatment and time-dependent confounding). Incorporating the outcome structure eventually reduces the error, but requires more samples to achieve a robust estimate.  
\begin{table}[h]
    \centering
    \begin{tabular}{lccc} 
        \toprule
        Method & $n=50$ & $n=100$ & $n=250$ \\
        \midrule
        Kernel DS & 0.32 (0.16) & 0.32 (0.16) & 0.37 (0.16) \\
        Operator Kernel DS & 0.33 (0.15) & 0.25 (0.15) & 0.27 (0.15) \\
        Iterative SRVF Operator Kernel DS & 0.30 (0.15) & 0.33 (0.16) & 0.26 (0.15) \\
        \bottomrule
    \end{tabular}
    \caption{Mean absolute error and standard deviation (in brackets) of causal effect estimators for continuous treatment and time-depedent outcomes $\bm{Y}$ and covariates $\bm{V}$.}
    \label{tab:methods_comparison_continuous_treatment}
\end{table}

\begin{figure}
\centering

\includegraphics[width=17cm]{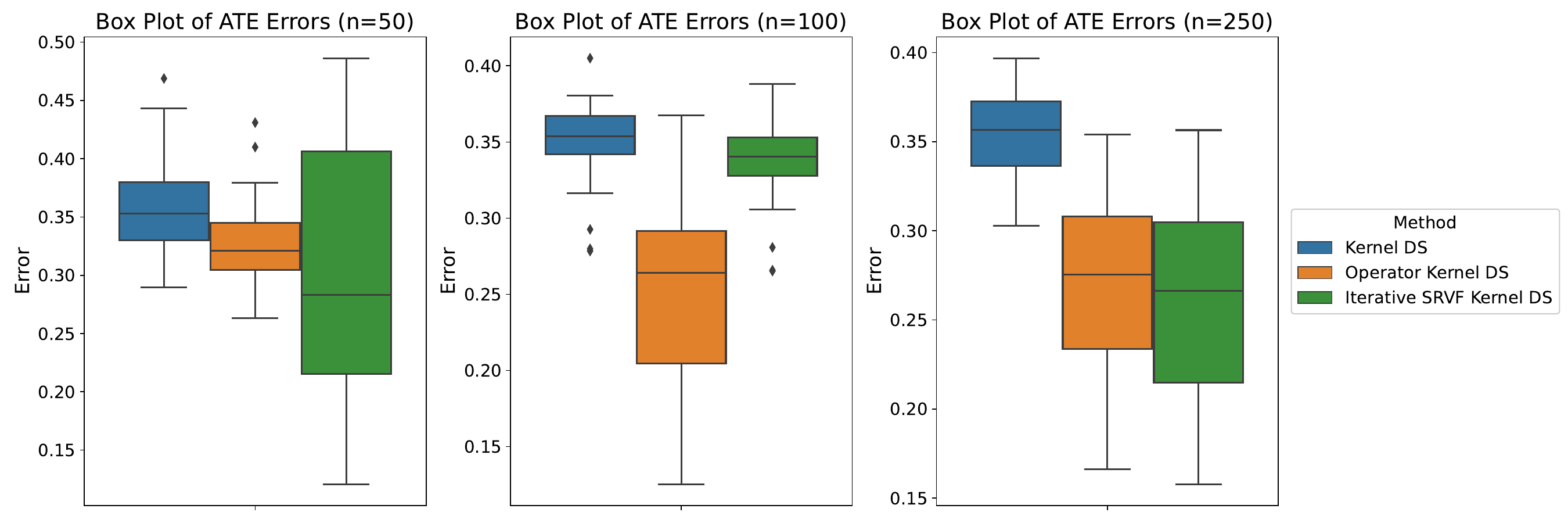}

\caption{Box plot of DS estimation errors for different methods under continuous treatment and time-dependent confounding, computed across different training sample sizes ($n=50, 100, 250$). The Kernel DS estimator ignore the multivariate structure, the Operator Kernel DS accounts for the multivariate outcome structure and the Iterative SRVF-based approach also infers the temporal alignment of $\bm{V}$ and $\bm{Y}$.}
\label{fig:continuous_box1}
\end{figure}

\begin{figure}[h]
\centering
\includegraphics[width=17cm]{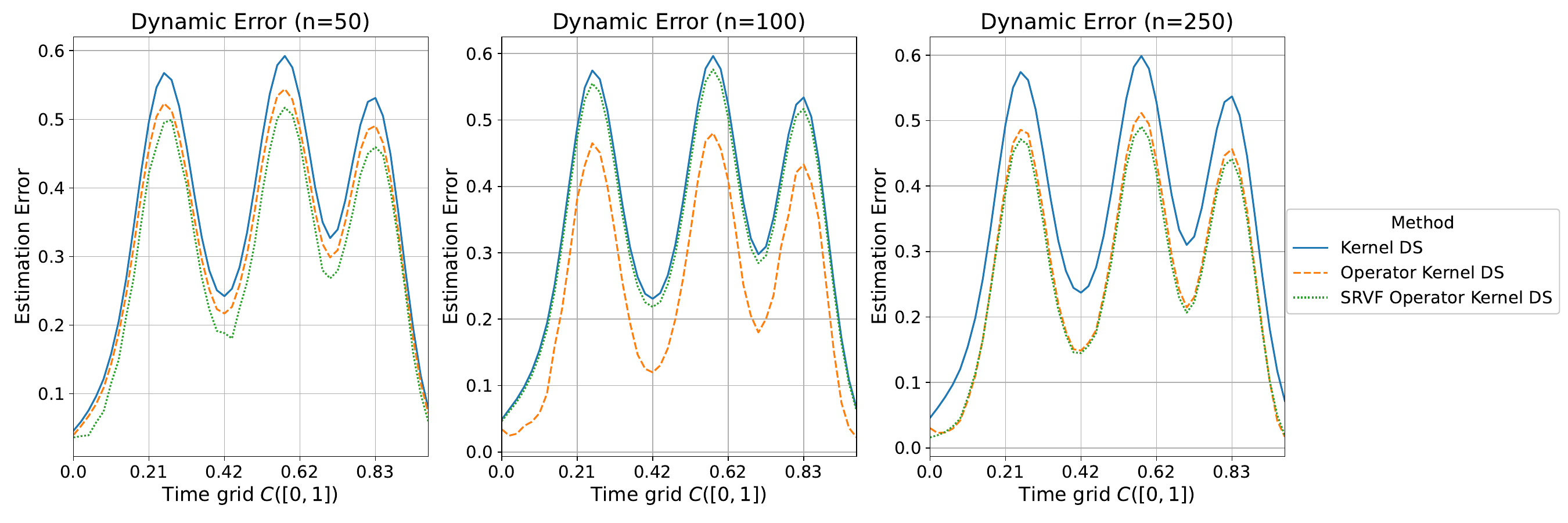}
\caption{Dynamic error as a function of time for different causal inference estimators for the DS,  under continuous treatment and time-dependent confounding. The plots display the average error (i.e., estimated across 5 draws from the super-population) for varying sample size $n=50, 100, 250$.}
\label{fig:continuous_line1}
\end{figure}

\subsection{Causal Effects on Digital Outcomes\label{subsec:Causal-effects-on}}

Digital biomarkers are increasingly recognized as reliable and sensitive clinical endpoints for assessing disease progression and treatment responses in Parkinson’s disease (PD). The ability to collect near real-time health measurements at scale makes these biomarkers particularly well-suited for functional data analysis. In this section, we apply our proposed framework to estimate the causal effects of (i) resting tremors and (ii) bradykinetic gait on their respective digital outcomes: tremor probability and gait energy. Our analysis is based on data from the Parkinson@Home Validation Study \citep{evers2020real}, which involved two weeks of passive monitoring using a wrist-worn device (Moto 360) in participants with and without PD.  

The study included 25 PD patients with motor fluctuations and 25 age-matched non-PD controls. Participants were monitored during both standardized clinical assessments (including MDS-UPDRS Part III) and unscripted daily life activities at home. The PD group underwent two monitoring sessions: once after overnight withdrawal of dopaminergic medication and again one hour post-medication. Additionally, all participants were continuously monitored in free-living conditions for two weeks. For our analysis, we estimate causal effects based on predicted symptom trajectories over this 2-week period. 

For PD patients with tremor, we selected data from the most affected arm (determined by MDS-UPDRS Part III, items 3.15 \& 3.17). For those without tremor, we selected the corresponding side matched for hand dominance. Due to technical issues with sensor devices in one PD patient and one non-PD control, the final dataset for training symptom detection models included 24 PD patients and 24 controls (further details about the tremor data are available in \citet{evers2025passive}). Performance of the trained symptom prediction classifiers is available in Appendix \ref{sec:digital-outcomes}. Out-of-sample symptom profiles were estimated for:  

\begin{itemize}
    \item 16 PD participants (8 with annotated tremor, 8 without) and 8 age-matched controls to assess the effect of disease status on tremor probability.
    \item 13 PD participants and 8 age-matched controls to estimate the effect of disease status on gait energy.
\end{itemize}

For further details on cohort demographics, please refer to Table~\ref{tab:Demographics}. Based on video recordings, trained research assistants annotated the main activities (e.g., walking, sitting, standing still) and symptoms (e.g., tremor and freezing of gait) occurring during unscripted activities. A movement disorders expert reviewed the symptom annotations. We used tremor presence annotations from the arm with the most severe tremor (i.e.,  the same side as the accelerometer sensor) and annotations for gait activity. In 8 PD patients, tremor was observed during unscripted daily life activities, whereas in the remaining 16 PD patients, no tremor was observed.

As a first step, we down-sampled the three-axial accelerometer data from the Physilog devices from 200Hz to 50Hz after \textit{anti-aliasing} with a fourth-order moving average filter. To remove the effect of orientation changes of the device, we applied $l_{1}$-trend filtering to each individual axis, assuming piecewise linear changes (\cite{raykov2020probabilistic}, i.e.,  setting $\lambda$ to 10,000). We then segmented the accelerometer data into non-overlapping 5-second windows and extracted features from each axis of the pre-processed accelerometer data (resulting in a total of 84 features for all axes combined, see Appendix \ref{sec:digital-outcomes}). Following common symptom detection practices, we trained a logistic classifier on home-based video annotations from 48 participants to predict tremor and gait episodes. Classification performance was evaluated against expert annotations and is reported in Table \ref{tab:detection_metrics}. The two classifiers were then used to estimate digital markers reflective of tremor and gait: tremor probability and bout energy during walking, respectively. 

These digital markers serve as proxies for symptom severity, allowing us to model their causal relationships with disease status, as outlined in Figure \ref{fig:DAGs-tremor-gait}. In this figure, we list the causal assumptions regarding known factors expected to affect disease category and the corresponding digital marker.\footnote{The example presents a simplified version of the problem and does not reflect the full list of factors affecting each digital outcome (e.g., geographic location, disease severity, and other contextual factors).}

\paragraph{Effect of PD on Gait Energy}

First, we quantify the effect of PD on expected gait bout energy during walking periods using 2-week follow-up data from 22 participants (14 PD, 8 non-PD age-matched controls). Gait bouts are predicted out-of-sample via the trained logistic classifier. The 2-week wearable outcomes are averaged to obtain a single daily profile per participant, where outcome curves are smoothed over 15-minute intervals. Daily measurements range between 5 and 12 hours, and the final outcome curves are constructed by aggregating across available days per time bin. 

We model the diagnostic category (binary PD vs. non-PD control) as the intervention variable $D$ while controlling for the confounding effects of \textit{hours awake} and \textit{gender}. Figure \ref{fig:gait-causal-function} presents both the average treatment effect estimated using the IPW estimator and the proposed dATE estimator from Section \ref{subsection:Temporal-Misalignment}. While we observe a strong overlap in potential outcome distributions (z-scored), many PD participants exhibit gait bouts with energy levels indistinguishable from non-PD controls. However, the PD group also displays a higher frequency of low-energy bouts. A Welch t-test on the IPW-estimated distributions fails to reject the null hypothesis that the two gait bout distributions originate from the same population. However, time-specific conditioning, as seen in Figure \ref{fig:Conditional-on-time-effects} (Appendix \ref{sec:digital-outcomes}), reveals significant fluctuations, which are better captured by the dATE estimator (Figure \ref{fig:gait-causal-function}) and confirmed by the corresponding Welch t-test. We estimate significant variations in PD gait bout energy, which can be at least partially attributed to dopaminergic therapy (average of four daily levodopa doses among the 14 included PD participants).

\paragraph{Effect of Diagnostic Category on Tremor}

Next, we quantify the effect of PD tremor diagnosis on expected tremor probabilities using 2-week follow-up data from 24 participants (8 PD with annotated tremor, 8 PD without annotated tremor, and 8 non-PD age-matched controls). Tremor probabilities are predicted out-of-sample via the trained logistic classifier. We follow the same setup as above to obtain single daily tremor profiles per participant.  Based on the assumptions in Figure \ref{fig:DAGs-tremor-gait}, we condition on the presence of non-gait activity to improve the precision of tremor score estimation. We first define the intervention variable $D$ as a binary tremor vs. non-tremor classification, where non-tremor includes both non-tremor PD and non-PD controls. Figure \ref{fig:tremor-causal-function} presents both the average treatment effect estimated using the IPW estimator and the proposed dATE estimator from Section \ref{subsection:Temporal-Misalignment}. Unlike gait energy, we observe much less overlap between the potential outcome distributions for tremor and non-tremor groups, a finding confirmed by the Welch t-test. The dynamic ATE is also significant but exhibits less structured temporal variation, likely due to tremor symptoms being less responsive to dopaminergic medication in most participants. This is further supported by the time-conditioned ATE estimates in Figure \ref{fig:Conditional-on-time-effects} (Appendix \ref{sec:digital-outcomes}). If we redefine $D$ by grouping both PD cohorts together and comparing them to non-PD controls, Figure \ref{fig:tremor-causal-function-dynamic} shows a natural reduction in ATE differences. This aligns with the presence of low but detectable tremor probabilities in PD participants without annotated tremor. Simultaneously, the dynamic ATE effect becomes more pronounced, reflecting the periodic levodopa-induced fluctuations expected in PD tremor cases.

\begin{figure}[h!]
    \centering

    \begin{minipage}{0.48\textwidth}
    \centering
    \begin{tikzpicture}[
        node distance=1cm and 1.5cm,
        every node/.style={draw, circle, minimum size=0.6cm, align=center, inner sep=1pt},
        every edge/.style={->, thick}
    ]
    \node (D1) [label={[label distance=-0mm]above:{\shortstack{Disease \\ Category}}}] {$D$};
    \node (Y1) [right=of D1, label={[label distance=0mm]above:{\shortstack{Bout \\ Energy}}}] {$Y$};
    \node (X11) [below left=0.8cm and 0.8cm of D1, label=below:{Gender}] {$X_1$};
    \node (X21) [below right=0.8cm and 0.8cm of Y1, label=below:{\shortstack{Hours \\ Awake}}] {$X_2$};

    \draw[->] (D1) -- (Y1);
    \draw[->] (X11) -- (D1);
    \draw[->] (X11) -- (Y1);
    \draw[->] (X21) -- (D1);
    \draw[->] (X21) -- (Y1);

    \end{tikzpicture}
    \end{minipage}
    \hfill
    \begin{minipage}{0.48\textwidth}
    \centering
    \begin{tikzpicture}[
        node distance=1cm and 1.5cm,
        every node/.style={draw, circle, minimum size=0.6cm, align=center, inner sep=1pt},
        every edge/.style={->, thick}
    ]
    \node (D2) [label=above:{\shortstack{Disease \\ Category}}] {$D$};
    \node (Y2) [right=of D2, label=above:{\shortstack{Tremor \\ Probability}}] {$Y$};
    \node (X12) [below left=0.8cm and 0.8cm of D2, label=below:{Gender}] {$X_1$};
    \node (X22) [below right=0.8cm and 0.8cm of Y2, label=below:{\shortstack{Hours \\ Awake}}] {$X_2$};
    \node (X32) [below=0.8cm of Y2, label=below:{\shortstack{Gait \\ Activity}}] {$X_3$};

    \draw[->] (D2) -- (Y2);
    \draw[->] (X12) -- (D2);
    \draw[->] (X12) -- (Y2);
    \draw[->] (X22) -- (D2);
    \draw[->] (X22) -- (Y2);
    \draw[->] (X32) -- (Y2);

    \end{tikzpicture}
    \end{minipage}
\caption{\label{fig:DAGs-tremor-gait}Causal Directed Acyclic Graphs (DAG) displaying the assumed association and its direction between estimated digital outcomes and clinical annotations. (Left) displays the minimal assumptions made for the gait bout energy, a known marker of bradykinetic gait in PD; (Right) displays the minimal assumptions made for factors affecting tremor probability estimated from a wrist-worn device.}
\end{figure}
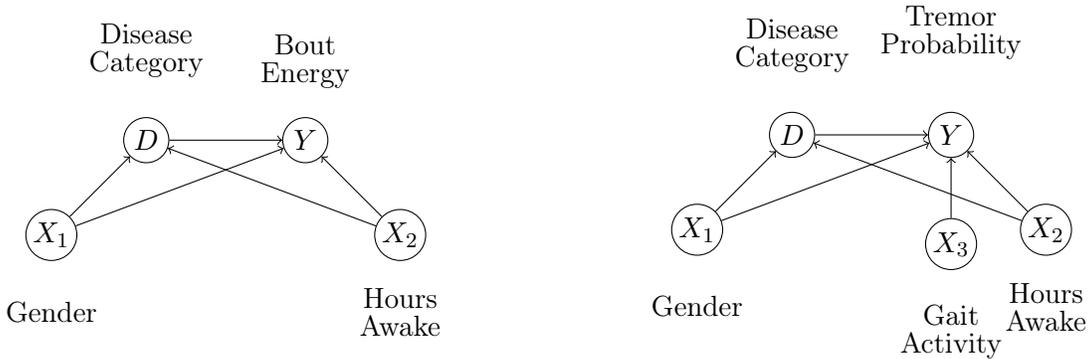

\begin{figure}[h!]
\centering

\includegraphics[height=4.4cm]{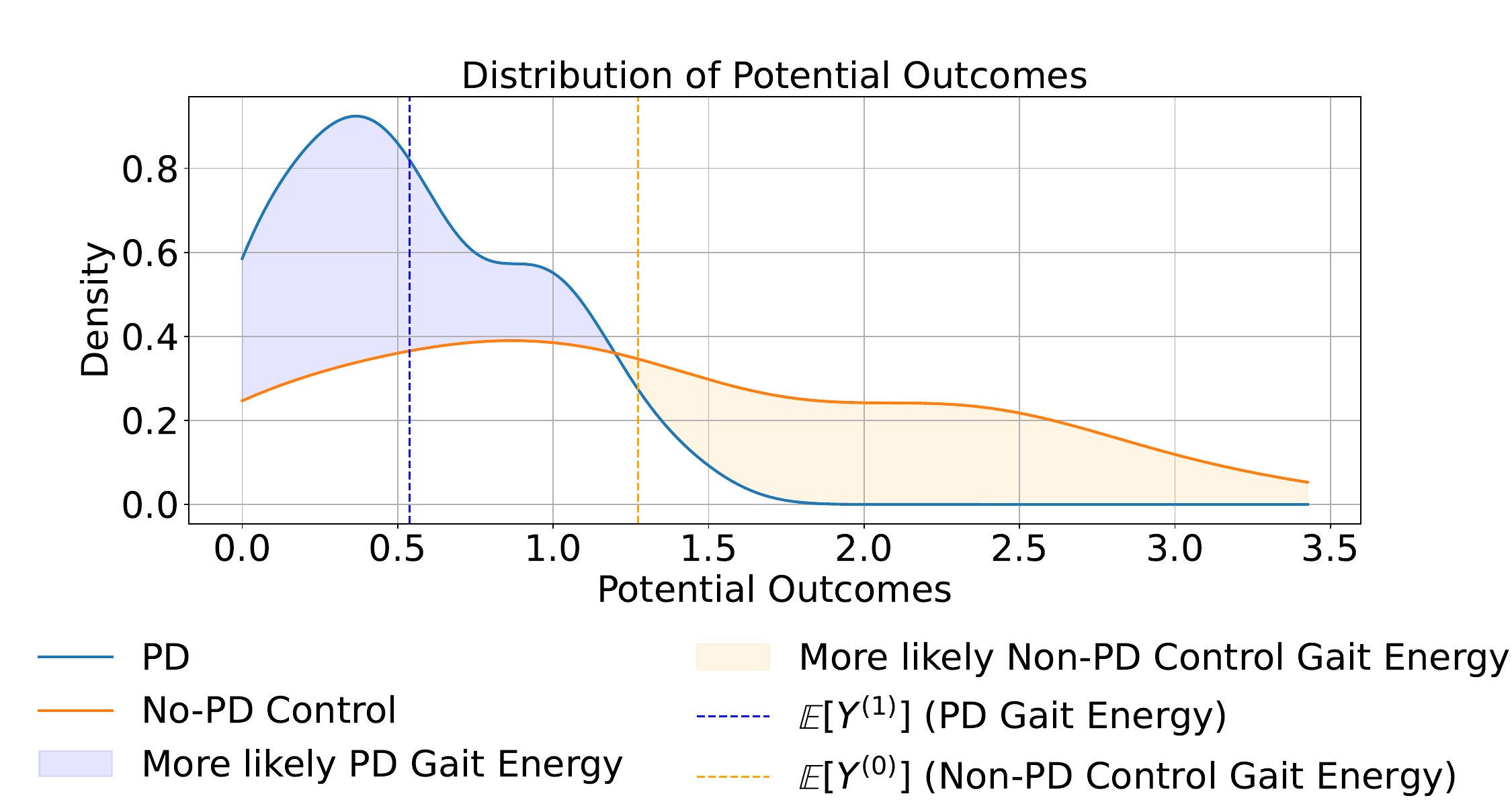}
\includegraphics[height=4.4cm]{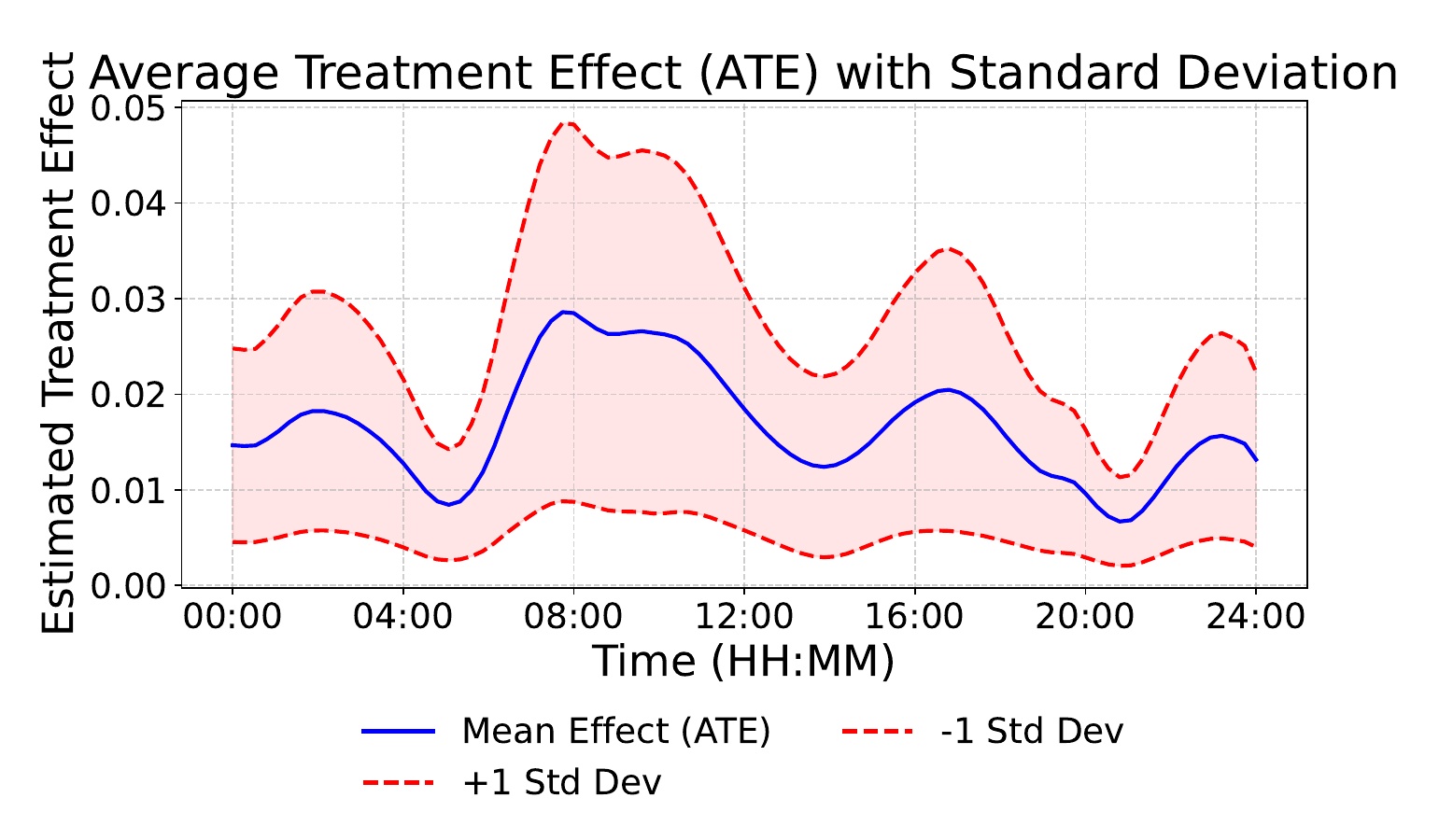}
\caption{\label{fig:gait-causal-function} Effect of disease category on gait
energy, estimated from average weekly digital outcomes. (Left) shows
a density plot of the potential outcomes of daily outcome values for
the two groups after controlling for the confounding effect of gender
and hours awake using \textit{IPW} estimator; (Right) shows the estimated
dynamic average treatment effect using the proposed kernel estimator
in Section \ref{subsection:Temporal-Misalignment}. Welch t-test is
performed to compare the difference in between the potential outcomes
in both settings: \textbf{p-value=0.12} is obtained for the IPW estimator
when comparing daily outcomes; \textbf{p-value<0.01} is obtained for
the averaged daily potential outcomes estimated from the kernel estimator.}
\end{figure}

\begin{figure}
\centering

\includegraphics[height=4.4cm]{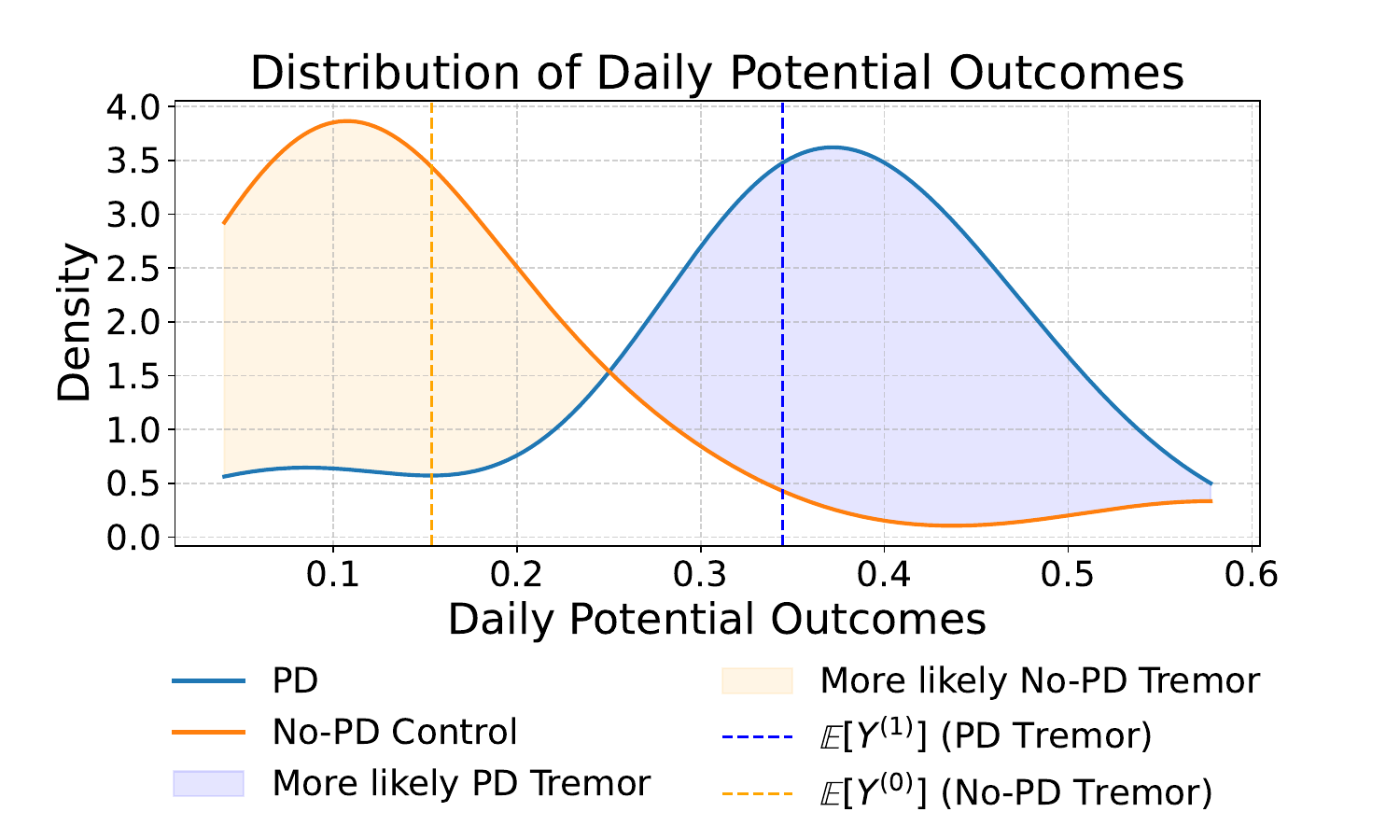}
\includegraphics[height=4.2cm]{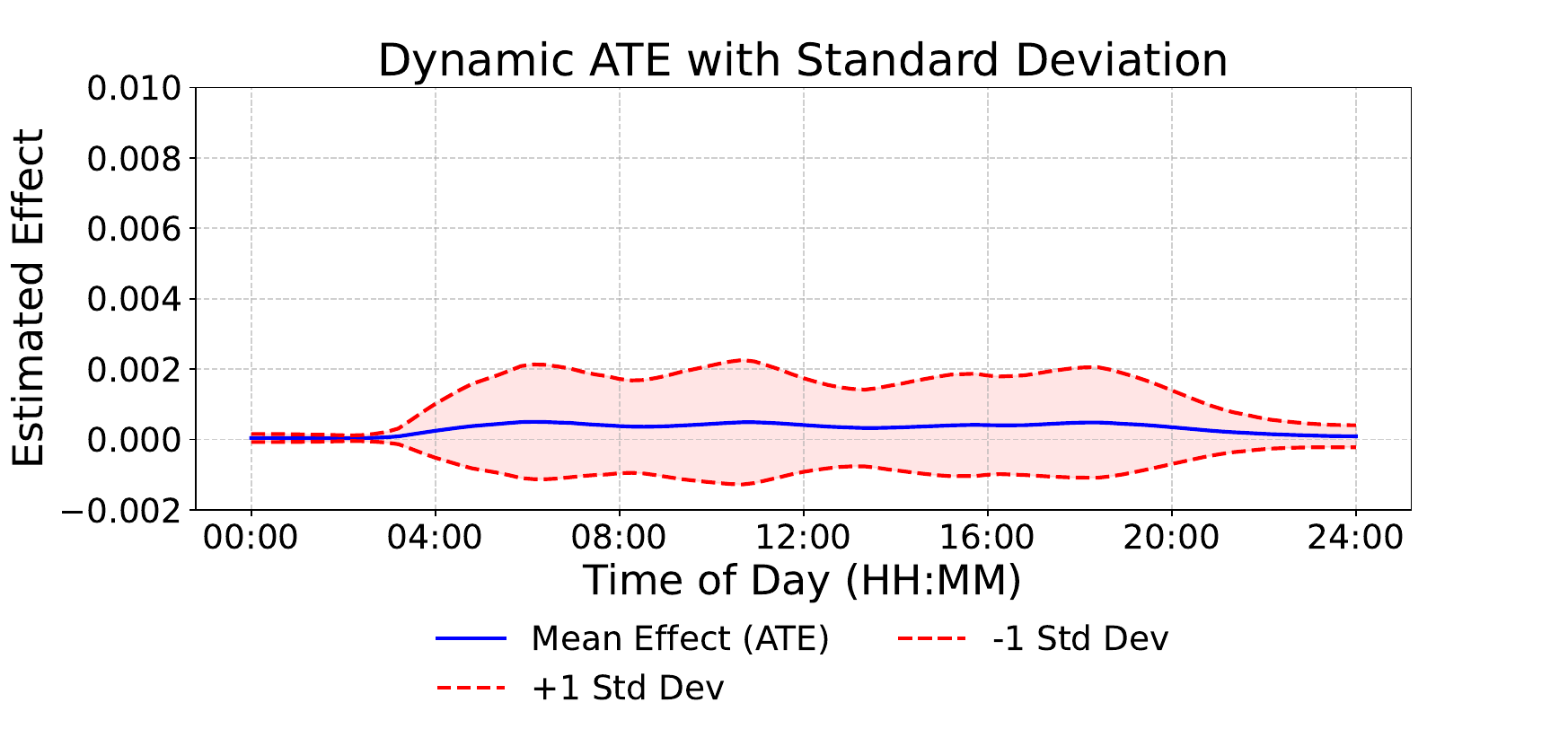}

\caption{\label{fig:tremor-causal-function} Effect of tremor annotation status
on tremor probabilities at-home, estimated from average weekly digital
outcomes. (Left) shows a density plot of the potential outcomes of
daily outcome values for the two groups after controlling for the
confounding effect of gender and hours awake using \textit{IPW} estimator;
(Right) shows the estimated dynamic average treatment effect using
the proposed kernel estimator in Section \ref{subsection:Temporal-Misalignment}.
Welch t-test is performed to compare the difference in between the
potential outcomes in both settings: \textbf{p-value<0.01} is obtained
for the IPW estimator when comparing daily outcomes; \textbf{p-value<0.01}
is obtained for the averaged daily potential outcomes.}
\end{figure}

\begin{figure}[h!]
\centering

\includegraphics[height=4.9cm]{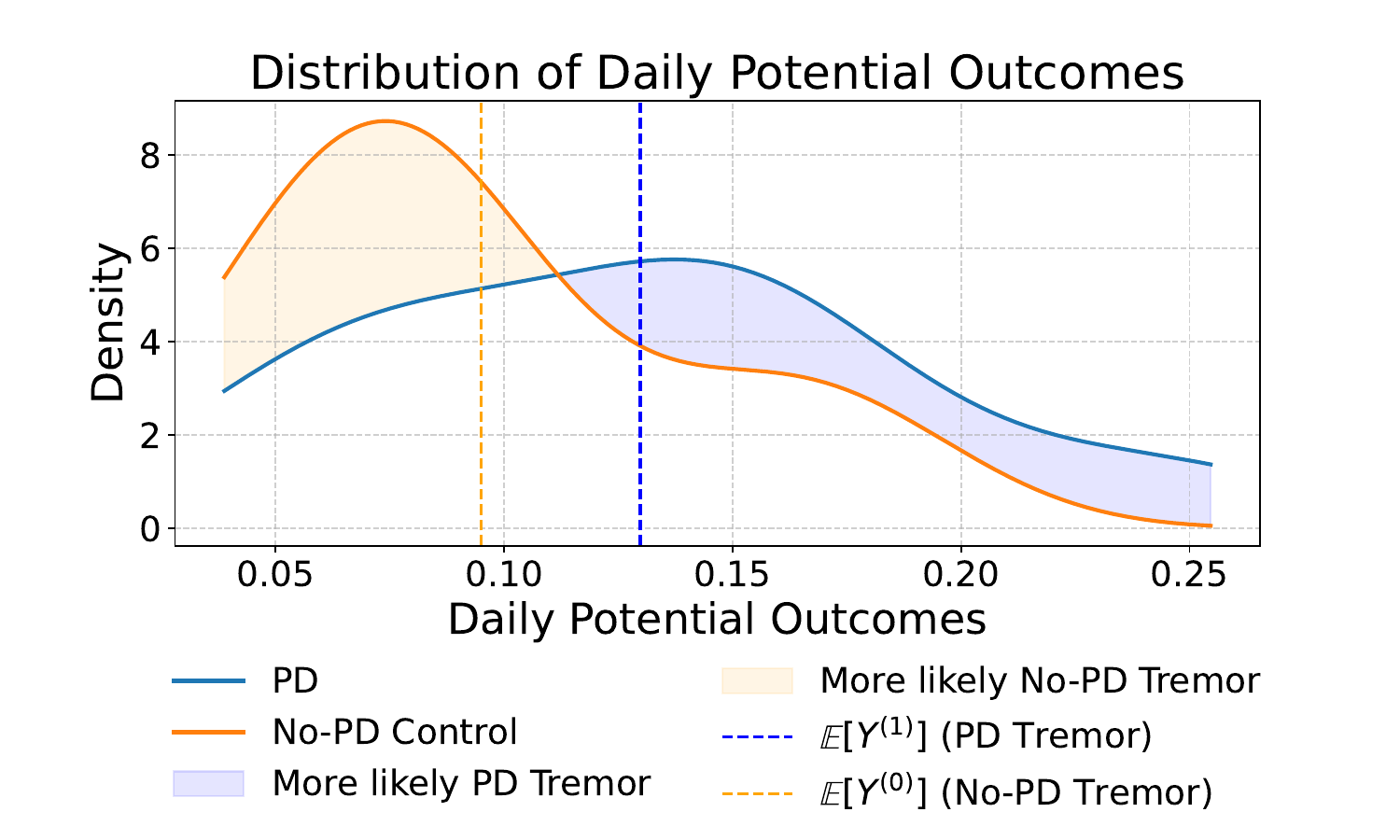}
\includegraphics[height=4.9cm]{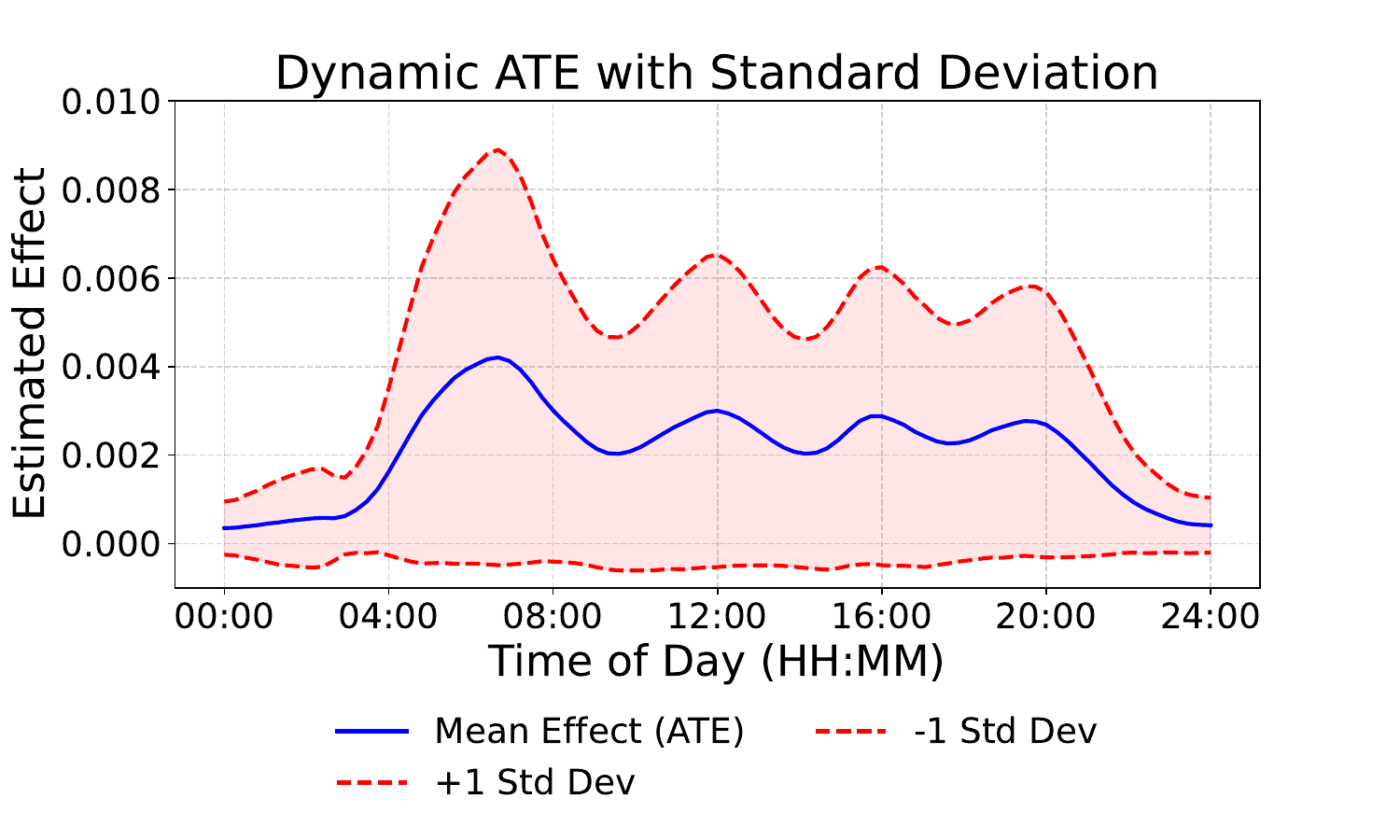}

\caption{\label{fig:tremor-causal-function-dynamic} Effect of PD status on
tremor probabilities at-home, estimated from average weekly digital
outcomes. (Left) shows a density plot of the potential outcomes of
daily outcome values for the two groups after controlling for the
confounding effect of gender and hours awake using \textit{IPW} estimator;
(Right) shows the estimated dynamic average treatment effect using
the proposed kernel estimator in Section \ref{subsection:Temporal-Misalignment}.}
\end{figure}

\section{Discussion}\label{sec:Discussion-Future-Directions}

This work introduces a novel framework for causal inference in settings involving functional data, extending traditional methods to accommodate dynamic and non-linear domains. 

The primary contribution lies in studying multivariate causal effects and developing appropriate kernel methods that integrate structural assumptions about the outcome space to derive closed-form estimators of the causal effects. Specifically, the framework addresses both binary
and continuous treatments with functional outcomes, demonstrating how these treatments can be modeled using Fréchet means in metric
spaces such as $L^{2}$ and Fisher--Rao (Section \ref{subsec:Euclidean-Metric-on}
and Section \ref{subsec:Fisher--Rao-Metric}). By leveraging operator-valued kernels, the proposed methods enable the modeling of function-valued covariates and outcomes, effectively capturing complex temporal dynamics
(Section \ref{subsec:kernel-estimators-functional}).

In comparison to existing approaches, this work generalizes prior efforts such as \citet{ecker2024causal}, which employ linear parametric
models to estimate potential outcomes for functional data under specific assumptions, and \citet{belloni2017program}, which focus on approximate
local treatment effect estimators for non-Euclidean outcomes without providing finite sample inference. Unlike these works, the present framework adopts a nonparametric paradigm for estimating causal effects
in infinite-dimensional functional settings, introducing joint aligning-kernelization procedures to handle path-valued random variables. Furthermore, it surpasses the methods in \citet{singh2020kernel}, which rely on scalar
outputs, by extending kernel-based causal estimators to functional outcomes. This contribution provides a more comprehensive approach to estimating causal effects on non-linear and temporally structured data.

Theorem \ref{thm:l2_dte_asymp_normality} establishes the asymptotic normality of the dynamic average treatment effect (dATE) under the $L^{2}$ metric, providing theoretical guarantees for estimator reliability, while Theorem \ref{thm:fr_dte_asymp_normality} formalizes the consistency of Fréchet mean estimators under the Fisher--Rao metric in functional spaces. These advancements represent significant extensions to the
methods proposed by \citet{lin2023causal}, by enabling causal effect estimation for complex functional outcomes and demonstrating their efficacy in high-dimensional settings. Theorem \ref{thm:l2_dte_asymp_normality} complements recent work in \citet{testa2025doubly} with focus towards scalar treatment effect estimators.

The proposed framework is validated through extensive experiments, including synthetic setups and real-world applications to digital monitoring of PD. The experiments illustrate the ability of the framework to analyze
high-dimensional functional data with temporal misalignment, evaluate more complex continuous treatment effects, and ensure appropriate estimation of the empirical estimators. In the presented PD application, we demonstrate how the dynamic ATE, $\varphi^{dATE}$, reveals different ways in which therapy affects symptoms, highlighting its previously underexplored potential 
(Section \ref{subsec:Causal-effects-on}). In summary, this work bridges a gap in the literature by introducing a nonparametric methodology for causal inference in functional data settings, combining theoretical rigor with practical algorithms that enable alignment. It also highlights open problems related to consistency guarantees (see Theorem \ref{thm:fr_dte_asymp_normality}) for causal effect estimators under more complex assumptions about the outcome and covariate spaces.


\section*{Acknowledgements}
The authors wish to thank all participants in the PD@Home validation study for their enthusiasm to contribute to this study and for welcoming us into their homes. We also thank Dr. Luc J.W. Evers for his input in developing clinically relevant hypotheses in Section~\ref{subsec:Causal-effects-on}. Yordan P. Raykov was funded by the Michael J. Fox Foundation for Parkinson’s Research (Grants 10231 and 17369) and supported in part by the EPSRC Horizon Digital Economy Research grant Trusted Data Driven Products under Grant EP/T022493/1. Wasiur R. KhudaBukhsh was supported by the Engineering and Physical Sciences Research Council (EPSRC) under Grant EP/Y027795/1.
 \newpage{}
 \bibliographystyle{plainnat}
\bibliography{main}
\newpage{}
\appendix

\section{Causal graphical models}

\label{subsec-causal-graphical-models} Below, we outline some fundamental
concepts frequently discussed in the causal graphical modeling literature
\citet{pearl2009causal}.

\subsection{Graphs }

A graph $G$ consists of a set of directed edges $E$ and indices
$[m]=\{1,\ldots,m\}$. The edges of $G$ can be represented by a parent
function $\text{pa}:[m]\to2^{[m-1]}$, such that $l\in\text{pa}(j)\iff l\to j$
in $G$. If $j$ has no incoming edges, then $\text{pa}(j)=\emptyset$.
In this case, the graph $G$ can be written as $G=(E,\text{pa})$.

\subsection{Causal DAGs }

A graph $G=(A,\text{pa})$ is a causal directed acyclic graph (DAG)
over random variables $A:=(A_{i})_{i=1}^{n}\sim P$, if the factorization
$P=\prod_{i}P_{i|\text{pa}(i)}$ holds, and the conditionals $(P_{j|\text{pa}(j)})_{j=1}^{n}$
represent actual (e.g., physical) data-generating mechanisms for observations. 
\begin{defn}[Back-Door Criterion, \citet{pearl1995causal}]
Consider a causal DAG $G$ over $(A_{i})_{i=1}^{n}$, and let $(X,V,Y)\subseteq(A_{i})_{i=1}^{n}$.
The set $V$ satisfies the back-door criterion with respect to $(X,Y)$
if: 
\begin{itemize}
\item No node in $V$ is a descendant of $X$. 
\item $V$ blocks every path between $X$ and $Y$ that contains an edge
pointing into $X$. 
\end{itemize}
\end{defn}

\begin{defn}[Front-Door Criterion, \citet{pearl1995causal}]
Consider a causal DAG $G$ over $(A_{i})_{i=1}^{n}$, and let $(X,V,Y)\subseteq(A_{i})_{i=1}^{n}$.
The set $V$ satisfies the front-door criterion with respect to $(X,Y)$
if: (1) $V$ intercepts all directed paths from $X$ to $Y$; (2) there
is no back-door path between $X$ and $V$; (3) every back-door path
between $V$ and $Y$ is blocked by $X$. 
\end{defn}

\section{Fisher--Rao via square-root slope functions\label{sec:square-root-slope-fucntions}}

\subsection{Square root slope functions\label{subsec:square-root-slope-fucntions}}
For function $f_{i}:[0,1]\rightarrow\mathbb{R}$, we define the \emph{square-root
slope function (SRSF)} as follows: 
\begin{equation}
q_{i}(t)=\text{sgn}(\dot{f}_{i}(t))\sqrt{\vert \dot{f}_{i}(t)\vert}\,,\label{eq:srsf}
\end{equation}
where \(\dot{\;}\) denotes the weak derivative.
Note that without loss of generality, $f_{i}$ can be defined on a
common closed interval $[t_{1},t_{m}]$. We also note that given SRSF
$q_{i}$, $f_{i}$ can be recovered up to its starting point $f_{i}(0)$
via $f_{i}(t)=f_{i}(0)+\int_{0}^{t}q_{i}(s)\vert q_{i}(s)\vert\ ds$.
Given two functions $f_{1},f_{2}$ with SRSFs $q_{1},q_{2}$,the elastic
pairwise registration problem means optimizing the following loss
function over the group of warping functions $\Gamma=\{\gamma:[0,1]\rightarrow[0,1]\ \vert\ \gamma(0)=0,\gamma(1)=1,\dot{\gamma}>0\}$:
\begin{equation}
\gamma^{*}=\underset{\gamma\in\Gamma}{\text{argmin}}\ \vert\vert q_{1}-(q_{2}\circ\gamma)\sqrt{\dot{\gamma}}\vert\vert_{L_2}^{2}\,,\label{eq:elastic_pairwise}
\end{equation}
where $\vert\vert\cdot\vert\vert_{L_2}$ denotes the $L^{2}$ norm, and
$\circ$ denotes function composition. This is traditionally solved
using dynamic programming or gradient-descent algorithms \citep{ramsay1998curve}.
Once estimated, let $q_{2}^{*}=(q_{2}\circ\gamma^{*})\sqrt{\dot{\gamma}^{*}}$
be the optimally-warped SRSF of $f_{2}$, and $f_{2}^{*}=f_{2}\circ\gamma^{*}$
be the corresponding optimally-warped $f_{2}$ which is best aligned
to $f_{1}$. Estimation of $\gamma^{*}$ can be sensitive to function
noise, as taking derivatives to compute the SRSF will exacerbate the
noise. As a remedy, one can add an additional penalty term to the
loss function in \eqref{eq:elastic_pairwise}, to penalize properties
of the estimated warping function (e.g., roughness).

For a collection of curves $\{f_{i}\}_{i=1}^{n}$, joint elastic registration
requires one to have a template function to which all functions are
jointly matched to. In the context of this work, those are selected
to be estimated from the data and we adopt the \emph{Karcher mean
function} with respect to the elastic metric\footnote{The Karcher mean is a specific case of the Fréchet 2-mean applied
in the context of Riemannian manifolds}. Karcher means are suitable for any metric space, but are particularly
advantageous when closed-form sample means are not easily specified
for the collection of data objects. We define the Karcher mean SRSF
as: 
\begin{equation}
\tilde{\mu}_{i}=\underset{q\in\mathbb{L}^{2}([0,1])}{\text{argmin}}\ \sum_{i=1}^{n}\vert\vert q-(q_{i}\circ\gamma_{i}^{*})\sqrt{\dot{\gamma}_{i}^{*}}\vert\vert_{L_2}^{2}\,\label{eq:karch_mean_fn}
\end{equation}
where $\gamma_{i}^{*}$ solves the pairwise registration optimization
in \eqref{eq:elastic_pairwise}. A two-step iterative algorithm is
required here, as one iterates between averaging functions to form
an estimate of the mean SRSF, and then performing separate pairwise
alignments to align functions to this current mean SRSF estimate.

\paragraph{Discrete approximation.}

In data we observe only the grid values
$Y_{i}(u_1),\dots,Y_{i}(u_T)$.
Write $\Delta u_t=u_{t+1}-u_t$.
After a smoothing step (e.g. monotone splines) we approximate
$\dot Y_i(u_t)\approx\Delta Y_{i,t}/\Delta u_t$ with
$\Delta Y_{i,t}=Y_i(u_{t+1})-Y_i(u_t)$
and define the discrete SRSF vector
\[
  q_{i,t}
  \;=\;
  \operatorname{sgn}\!\bigl(\Delta Y_{i,t}\bigr)\,
  \sqrt{\bigl|\Delta Y_{i,t}\bigr|/\Delta u_t},
  \qquad t=1,\dots,T-1 .
\]
The inner product
\(
  \langle q_f,q_g\rangle_{L_2}
\)
is consistently approximated by the Riemann sum
\(
  \sum_{t=1}^{T-1} q_{f,t}\,q_{g,t}\,\Delta u_t ,
\)
turning \eqref{eq:elastic-FR} into a computable formula on $\mathbb R^{T-1}$.

\subsection{Incompleteness of the phase-quotient space}
\label{app:incompleteness}

\begin{prop}
Let \(\Gamma\) act on \(\mathcal F\) by right composition
\(f\mapsto f\circ\gamma\).  
Endow \(\mathcal F/\Gamma\) with the
Fisher-–Rao metric
\(d_{FR}([f_1],[f_2])
   =\inf_{\gamma\in\Gamma}\|q_{f_1}-q_{f_2\circ\gamma}\|_{L_2}\).
Then \((\mathcal F/\Gamma,d_{FR})\) is \emph{not complete}.
\end{prop}

\begin{proof}[Proof Sketch]
Fix \(f^\star(t)=t\) and define the warps
\(\gamma_n(t)=t^{\,n}\) (\(n\ge1\)).
Set \(f_n=f^\star\circ\gamma_n=t^{\,n}\).
One checks  
\(d_{FR}\bigl([f_n],[f_{n+1}]\bigr)=\|q_{f_n}-q_{f_{n+1}}\|_{L_2}
  \le C n^{-3/2}\),
hence \(\bigl([f_n]\bigr)\) is Cauchy. The pointwise limit of \(f_n\) is
\(f_\infty(t)=\mathbf 1\{t=1\}\), whose derivative vanishes a.e.,
so \(f_\infty\notin\mathcal F\).
Therefore the sequence has no limit in \(\mathcal F/\Gamma\).
A more delicate construction shows that even allowing post-composition
by \(\Gamma\) cannot “repair’’ the limit, proving incompleteness.
\end{proof}

\noindent
Because completeness is a necessary condition for the existence of
Fréchet means and for the strong law of large numbers on metric spaces
\citep{Sturm2003NPC}, additional constraints or a regularisation of \(\Gamma\) are required if one wishes to obtain global uniqueness and consistency results.

\subsection{Phase-quotient space with Hilbert Structure}
\label{app:hilbert-vs-group-actions}

In this appendix, we investigate under which conditions the quotient space \(\mathcal{F}/\Gamma\), arising from group actions on space in $\mathcal{F}$ with some Hilbert space structure (e.g. Sobolev space $W^{k,2}([0,1],\mathbb{R})$), inherits a Hilbert space structure. This is important for establishing convexity, uniqueness, and consistency guarantees for Fréchet mean estimators under warping invariance.
 
The quotient of a Hilbert space by a closed linear subspace is itself a Hilbert space. However, domain warpings of the form \(f \mapsto f \circ \gamma\), where \(\gamma \in \Gamma\), are nonlinear transformations on \(L_2([0,1])\). As soon as \(\gamma\) depends nontrivially on \(t\), the orbit \(\{\,q_f \circ \gamma : \gamma \in \Gamma\,\}\) becomes a curved manifold in \(L_2\), and the quotient cannot inherit a linear or Hilbert space structure. Consider the restricted action group of phase shifts:
\[
\Gamma_c \;=\; \left\{\, \gamma_c(t) = t + c \mod 1 : c \in [0,1) \,\right\}.
\]
Here, the group acts by isometric translations on the circle. The associated quotient space \(\mathcal{F}/\Gamma_c\) now satisfies some improved properties:
\begin{itemize}
    \item The quotient \(\mathcal{F}/\Gamma_c\) is \emph{complete} — the counterexample from Appendix~\ref{app:incompleteness} no longer applies since \(\gamma_c' \equiv 1\).
    \item The Fisher--Rao metric reduces to the \emph{orbit distance}:
  \[
  d([q],[p]) \;=\; \inf_{c \in [0,1)} \| q - p(\cdot + c) \|_{L_2},
  \]
  which is well-known in signal processing and time series alignment.
\end{itemize}

Despite these improvements, the orbit distance \(d\) is not induced by an inner product and it violates the parallelogram law: therefore \(\mathcal{F}/\Gamma_c\) is still not a Hilbert space. Uniqueness of the Fréchet mean may fail if multiple shifts yield the same minimum.

However, the space is now a complete, proper length space with non-positive curvature (a CAT(0) “quotient cylinder”), allowing for geodesic convexity and weaker versions of consistency and uniqueness to hold under suitable regularity assumptions.

Let $\mathcal{H}$ be a Hilbert space, and let a group $G$ act on $\mathcal{H}$ by isometries. We now characterise when the quotient $\mathcal{H}/G$ is itself a Hilbert space.

\begin{prop}
Let $\mathcal{H}$ be a Hilbert space and $G$ a group acting isometrically on $\mathcal{H}$. Then the quotient space $\mathcal{H}/G$ is a Hilbert space if and only if:

\begin{enumerate}
    \item Every orbit is a closed affine subspace of $\mathcal{H}$; and
    \item The action of $G$ is by translations along a closed linear subspace $\mathcal{S} \subset \mathcal{H}$.
\end{enumerate}

In this case, $\mathcal{H}/G \cong \mathcal{H}/\mathcal{S}$, with the quotient norm given by
\[
\|[h]\| \;=\; \inf_{s \in S} \| h - s \|_{\mathcal{H}}.
\]
\end{prop}

\begin{proof}[Proof Sketch]
“Only if” direction: If $\mathcal{H}/G$ is Hilbert, then it satisfies the parallelogram identity. This implies convexity of orbits, which in turn implies they must be affine subspaces. Preservation of midpoints forces the group to act by translations. Completeness implies the orbits are closed.

“If” direction: If $G$ acts by translation along a closed subspace $\mathcal{S}$, then the quotient $\mathcal{H}/S$ is a Hilbert space under the usual norm.
\end{proof}

\paragraph{Example: Constant shifts} 
Let $\Gamma_c$ act on $\mathcal{H}$ by shifts \(T_c q = q(\cdot + c)\). This is a unitary action, but the orbit \(\{T_c q\}\) is not affine unless \(q \equiv 0\). Thus the quotient $\mathcal{H}/\Gamma_c$ is not Hilbert. For any non-constant warp $\gamma$, the operator $T_\gamma: q \mapsto q \circ \gamma$ is nonlinear, and the conditions of the proposition fail entirely.

\subsection{Restricting the original space of outcomes $\mathcal{F}$}
\label{app:restricting-outcome-space}

An alternative approach we explore is considering the $T$-dimensional vector domain $\mathcal{T}$ of the realizations of the outcomes $f(u_t)$ on a fixed grid $u_1,\dots,u_T$ which is sufficiently restricted so we can state strong consistency results for the empirical Fréchet mean estimators over a finite grid. The key assumption our analysis is best on is monotonicity constraint for the evaluations $f(u_1),\dots,f(u_T)$ over the grid: consider a subset 
$\mathcal{T} \subset \mathcal{D} \subset \mathbb{R}^T$. Consequently, we focus on functions that are bounded over $\mathcal{T}$. Formally, we restrict our original space $\mathcal{F}$ to  
\[
\mathcal{G} \;=\; B(\mathcal{T}) \;\subseteq\; \mathcal{F},
\]
the space of bounded functions on $\mathcal{T}$. Likewise, we restrict the pushforward measure $\eta_x$ to $\mathcal{G}$, obtaining a new measure $P_x$ defined by  
\[
P_x(A) \;=\; \eta_x(A \cap \mathcal{G}) 
\quad\text{for all measurable sets }A \subseteq \mathcal{G}.
\]
In other words, $P_x$ is a more specialized (conditional) version of $\eta_x$, supported only on those functions in $\mathcal{F}$ that lie in $\mathcal{G}$.

\begin{thm}[Fréchet means under the \emph{extended} Fisher–Rao metric]
\label{thm:fr_dte_asymp_normality}
Fix an integer $T\ge 2$ and let
\[
\mathcal D
   =\bigl\{\bm Y=(Y(1),\dots,Y(T))^\top\in\mathbb R^{T}
                : Y(1)<\dots<Y(T)\bigr\},
\qquad
\mathcal T
   =\mathcal D\cap[\,0,\;T^{-1}-\delta T^{-1}\,]^{T},
\]
for some $0<\delta<1$.  Assume $n$ i.i.d.\ observations
$\bm Y_1,\dots,\bm Y_n\in\mathcal T$ with
$\bm Y_i=(Y_i(1),\dots,Y_i(T))^\top$.  
Each $\bm Y_i$ is obtained by evaluating an underlying positive
function $f_i\in\mathcal G$
at the design points $\{1/T,\dots,T/T\}$, where  
\(
  \mathcal G=B(\mathcal T)
\)
is the Banach space of bounded functions on $\mathcal T$.  Denote by
$P_x$ the distribution of $f_i$ on $\bigl(\mathcal G,d_{\mathrm{EFR}}\bigr)$,
where $d_{\mathrm{EFR}}$ is the \emph{extended} Fisher–Rao distance

\[
d_{\mathrm{EFR}}(f,g)
  =2\arccos\!\Bigl(
     \frac{\int_0^1\sqrt{f(t)g(t)}\,dt}
          {\sqrt{\int_0^1\!f(t)\,dt}\,
           \sqrt{\int_0^1\!g(t)\,dt}}
    \Bigr),\qquad f,g>0.
\]

\medskip\noindent
\textbf{(1) Population Fréchet mean in $\mathcal G$.}\;
Under Assumption \ref{assu:(uniqueFrechet)} the Fréchet mean of
$\{f_i\}_{i=1}^n$ is unique and equals
\begin{equation}
\overline f
  =\arg\min_{g\in\mathcal G}\,
     \sum_{i=1}^n d_{\mathrm{EFR}}\!\bigl(f_i,g\bigr)^2
  =\arg\min_{g\in\mathcal G}\,
     \sum_{i=1}^n
       \Bigl[
         2\arccos
            \Bigl(
              \frac{\int_0^1\sqrt{f_i(t)g(t)}\,dt}
                   {\sqrt{\int_0^1\!f_i}\;
                    \sqrt{\int_0^1\!g}}
            \Bigr)
       \Bigr]^2.
\label{eq:def_f_0}
\end{equation}

\medskip\noindent
\textbf{(2) Empirical Fréchet mean on the discrete grid.}\;
For $\bm g=(g_1,\dots,g_T)\in\mathcal T$ define the extended FR distance
on~$\mathcal T$ by
\[
\phi(\bm Y_i,\bm g)
   =2\arccos\Bigl(
       \frac{\sum_{j=1}^T\sqrt{Y_i(j)g_j}}
            {\sqrt{\sum_{j=1}^T Y_i(j)}
             \,\sqrt{\sum_{j=1}^T g_j}}
     \Bigr).
\]
The empirical Fréchet mean is
\begin{equation}
\overline f_{\,n}
   =\arg\min_{\bm g\in\mathcal T}\,
         \frac1n\sum_{i=1}^n \phi(\bm Y_i,\bm g)^2
   =\arg\min_{\bm g\in\mathcal T}
       \frac1n\sum_{i=1}^n
       \Bigl[
         2\arccos
            \Bigl(
              \frac{\sum_{j=1}^T\sqrt{Y_i(j)g_j}}
                   {\sqrt{\sum_{j=1}^T Y_i(j)}
                     \sqrt{\sum_{j=1}^T g_j}}
            \Bigr)
       \Bigr]^2.\label{eq:def_f_n-1}
\end{equation}

\medskip\noindent
\textbf{(3) Uniqueness and consistency.}\;
Assume the minimiser $\bm g_\ast$ in
\eqref{eq:def_f_n-1} is not proportional to~$\bm1$.
Under conditions (1)–(4) of
Theorem~\ref{thm:IPW-interpolant-consistency} and mild smoothness of
$d_{\mathrm{EFR}}$ on $\mathcal T$, the empirical mean
$\overline f_{\,n}$ is almost‐surely unique and
\[
\overline f_{\,n}\;\xrightarrow[n\to\infty]{\;P\;}
       \overline f,
\qquad
\text{equivalently }\;
\inf_{\bm g\in\mathbb R^T}\frac1n\sum_{i=1}^n\phi(\bm Y_i,\bm g)
   \;\xrightarrow[n\to\infty]{\;} 
   \inf_{g\in\mathcal G}\,E_{P_x}\bigl[d_{\mathrm{EFR}}(f,g)\bigr].
\]
Thus $\overline f_{\,n}$ is a Huber–type $\rho$-estimator:
it minimises a sample loss, is consistent, and is (a.s.) unique, even
though $\phi(\cdot,\cdot)$ is not convex in $\mathbb R^{T}$.
\end{thm}

\begin{proof}
The idea of our proof is as follows, we restrict the vector domain
$\mathcal{T}$ so small that the $\phi(\bm{Y},\bm{g})$ is injective
and nearly linear in each of these $T$-coordinates in $\bm{Y}$ when
domain $\mathcal{T}$ is assumed; then we verify the (A1)-(A5) in
\citet{huber1967behavior} to elicit the $\bar{f}_{n}$ as a consistent
solution to the finite dimensional minimization problem; then we need
to verify assumptions of Lemma 1 in \citet{cox2020almost} to ensure
that this solution is unique.

\textbf{Step 1: injectivity.} Let $\bm{a} = (a_1, \dots, a_T)$ and $\bm{b} = (b_1, \dots, b_T)$ be two elements of $\mathcal{T} \subset \mathbb{R}^T$.
Define the function
\[
F(\bm{a}) = \frac{\sum_{i=1}^{T} w_i \sqrt{a_i}}{\sqrt{\sum_{i=1}^T a_i}},
\]
where $w_1 > w_2 > \cdots > w_T > 0$ are fixed strictly decreasing weights. We aim to show that $F$ is injective on the domain
\[
\mathcal{D} = \left\{ (a_1, \dots, a_T) \in [0, \infty)^T \mid a_1 \le a_2 \le \cdots \le a_T < \tfrac{1}{T} - \tfrac{\delta}{T} \right\}.
\]

Suppose $F(\bm{a}) = F(\bm{b})$ for $\bm{a}, \bm{b} \in \mathcal{D}$.
Let $S_{\bm{a}} = \sum_{i=1}^{T} a_i$ and $S_{\bm{b}} = \sum_{i=1}^{T} b_i$.
Then
\[
\frac{\sum_{i=1}^{T} w_i \sqrt{a_i}}{\sqrt{S_{\bm{a}}}} = \frac{\sum_{i=1}^{T} w_i \sqrt{b_i}}{\sqrt{S_{\bm{b}}}}.
\]
Multiplying both sides by $\sqrt{S_{\bm{a}}}$ and $\sqrt{S_{\bm{b}}}$ gives
\begin{equation}
\sum_{i=1}^{T} w_i \sqrt{a_i} = r^{-1} \sum_{i=1}^{T} w_i \sqrt{b_i},
\qquad \text{where } r := \frac{\sqrt{S_{\bm{b}}}}{\sqrt{S_{\bm{a}}}} > 0.
\label{eq:extended-proj-A}
\end{equation}

To isolate the contribution of the largest coordinate, rewrite \eqref{eq:extended-proj-A} as:
\begin{equation}
w_T \big( \sqrt{a_T} - r^{-1} \sqrt{b_T} \big)
= \sum_{i=1}^{T-1} w_i \big( r^{-1} \sqrt{b_i} - \sqrt{a_i} \big).
\label{eq:extended-proj-B}
\end{equation}
The right-hand side is bounded in absolute value by $w_{T-1} \sum_{i=1}^{T-1} |\sqrt{b_i} - r \sqrt{a_i}|$.
Since $w_T > w_{T-1}$, the equality in \eqref{eq:extended-proj-B} can only hold if
\begin{equation}
\sqrt{a_T} = r^{-1} \sqrt{b_T}
\quad \Rightarrow \quad
a_T = r^{-2} b_T.
\label{eq:extended-proj-C}
\end{equation}

Summing both sides of \eqref{eq:extended-proj-C} over $i=1,\dots,T$ gives
\[
S_{\bm{a}} = r^{-2} S_{\bm{b}}.
\]
But by definition $r^2 = S_{\bm{b}} / S_{\bm{a}}$, so substituting in gives
\[
S_{\bm{a}} = \frac{1}{r^2} S_{\bm{b}} = \frac{1}{r^2} (r^2 S_{\bm{a}}) = S_{\bm{a}} \quad \Rightarrow \quad r^2 = 1.
\]
Hence $r = 1$ and therefore $a_T = b_T$. Substituting this into \eqref{eq:extended-proj-A} and cancelling the $T$th terms yields
\[
\sum_{i=1}^{T-1} w_i \sqrt{a_i} = \sum_{i=1}^{T-1} w_i \sqrt{b_i}.
\]
Repeating the same argument recursively (using that $w_{T-1} > w_{T-2} > \cdots$) we conclude that $a_i = b_i$ for all $i=1, \dots, T$. Thus, $F$ is injective on the domain $\mathcal{D}$.\\

\textbf{Step 2: almost-sure uniqueness.}
Because $\phi_{\textsc{EFR}}$ is not convex on $\mathbb R^{T}$,
classical $M$–estimator theory cannot guarantee a single minimiser.
We therefore verify the three requirements of Lemma 1 in
\citet{cox2020almost}.

\begin{itemize}

\item \textbf{(Absolute continuity).}\;
      For any fixed $\bm g\in\mathcal T$ the map
      \[
        \bm Y\;\longmapsto\;
        \phi_{\textsc{EFR}}(\bm Y,\bm g)
        =2\arccos\!\Bigl(
           \frac{\sum_{j=1}^{T}\sqrt{Y(j)\,g_j}}
                {\sqrt{\sum_{j=1}^{T}Y(j)}\,
                 \sqrt{\sum_{j=1}^{T}g_j}}
        \Bigr)
      \]
      is $\mathcal B(\mathbb R^{T})$–measurable and takes values in
      the bounded interval $\bigl[0,\,2\arccos(1-\varepsilon)\bigr)$ with
      $\varepsilon=\delta/(1-\delta)>0$.  
      Hence its distribution under $P_x$ is absolutely continuous on
      the full support $\mathcal T$.

\item \textbf{(Manifold).}\;
      The optimisation set is the single compact, convex polytope
      $\mathcal T\subset\mathbb R^{T}$, which is a second–countable
      Hausdorff manifold of dimension $T$.

\item \textbf{(Continuous differentiability).}\;
      Write
      \( \displaystyle
        \Psi(\bm y,\bm g)=
        \frac{\sum_{j}\sqrt{y_j g_j}}
             {\sqrt{\sum_{j}y_j}\sqrt{\sum_{j}g_j}}
      \).
      Both numerator and denominator are $C^{\infty}$ on
      $(0,\infty)^{2T}$ and strictly positive on
      $\mathcal T\times\mathcal T$, so $\Psi$ is $C^{\infty}$.  
      Because $\Psi\le 1-\varepsilon$, the composition
      $\phi_{\textsc{EFR}}=2\arccos\circ\Psi$ is also $C^{\infty}$.
      Consequently:

      \begin{enumerate}
      \item $\phi_{\textsc{EFR}}(\bm Y,\bm g)$ is continuous in
        $(\bm Y,\bm g)\in\mathcal T\times\mathcal T$.
      \item For every $\bm g\in\mathcal T$ the map
        $\bm Y\mapsto\phi_{\textsc{EFR}}(\bm Y,\bm g)$ is
        differentiable, and its gradient is continuous in
        $(\bm Y,\bm g)$.
      \item For every $\bm Y\in\mathbb R^{T}$, $\bm g\in\mathcal T$
        and every direction
        $\Delta\in A(\bm g)$ (the $T$-dimensional tangent cone of
        $\mathcal T$ at $\bm g$),
        the directional derivative
        $D_\Delta\phi_{\textsc{EFR}}(\bm Y,\bm g)$ exists and is
        continuous in $(\bm Y,\bm g)$.
      \end{enumerate}
\end{itemize}

\medskip
\noindent
\emph{Genericity (assumption\,a).}  
Define
\[
\xi(\bm g_{1},\bm g_{2},\bm Y)
   =\phi_{\textsc{EFR}}(\bm Y,\bm g_{1})
    -\phi_{\textsc{EFR}}(\bm Y,\bm g_{2}).
\]
By Step 1 (injectivity) the normalised inner–product term inside the
two arccosine expressions is different whenever
$\bm g_{1}\neq\bm g_{2}$, so
\(\xi(\bm g_{1},\bm g_{2},\bm Y)\neq 0\) for all
\[
(\bm g_{1},\bm g_{2},\bm Y)\in
\Xi=\{(\bm g_{1},\bm g_{2},\bm Y):
      \bm g_{1},\bm g_{2},\bm Y\in\mathcal T,\;
      \bm g_{1}\neq\bm g_{2}\}.
\]

\medskip
All three Cox–Reid assumptions and the genericity condition are thus
satisfied; therefore the empirical minimiser in
\eqref{eq:def_f_n-1}$^\star$ is almost surely \emph{unique}.

\paragraph{Step 3: Consistency.}

Let $\bm Y_{1},\dots,\bm Y_{n}\in\mathcal{T}\subset\mathbb R^{T}$ be
i.i.d.\ draws from a distribution $P_x$ with finite second moment.
Write
\[
\phi_{\textsc{EFR}}(\bm y,\bm g)
   \;=\;
   2\arccos
     \Bigl(
       \frac{\displaystyle\sum_{j=1}^{T}\sqrt{y_j g_j}}
            {\sqrt{\displaystyle\sum_{j=1}^{T} y_j}\,
             \sqrt{\displaystyle\sum_{j=1}^{T} g_j}}
     \Bigr),
   \qquad \bm y,\bm g\in\mathcal{T}.
\]

Define the empirical loss  
\(
  \hat\gamma_n(\bm g)=\tfrac1n\sum_{i=1}^{n}
                      \phi_{\textsc{EFR}}(\bm Y_i,\bm g)^2
\)
and let  
\(
  \bar f_n=\arg\min_{\bm g\in\mathcal T}\hat\gamma_n(\bm g).
\)
We verify Huber’s conditions (A1)–(A5) for
$\rho(\bm Y,\bm g)=\phi_{\textsc{EFR}}(\bm Y,\bm g)^2$; then
$\bar f_n\!\to\!\bar f$ in probability (and a.s.) by
\citeauthor{huber1967behavior}’s theorem.

\begin{itemize}
\item[\textbf{(A1)}] \textbf{Measurability.}
      For fixed $\bm g$, $\phi_{\textsc{EFR}}(\bm Y,\bm g)$ is
      $\mathcal B(\mathbb R^T)$-measurable.  Because
      $\bm y\!\mapsto\!\sum_j\sqrt{y_j g_j}$ and
      $\bm y\!\mapsto\!\sum_j y_j$ are continuous, the arccosine of their
      ratio is measurable as well.

\item[\textbf{(A2)}] \textbf{Lower semicontinuity.}
      $\phi_{\textsc{EFR}}(\bm y,\bm g)$ is $C^{1}$ on
      $\mathcal T\times\mathcal T$; thus it is continuous and
      $P_x$-a.s.\ lower semicontinuous in $\bm y$.

\item[\textbf{(A3)}] \textbf{Finite envelope.}
      On $\mathcal T$ we have $y_j,g_j\le 1/T-\delta/T$ so
      \[
        \frac{\sum_{j}\sqrt{y_j g_j}}
             {\sqrt{\sum_{j}y_j}\sqrt{\sum_{j}g_j}}
        \;\le\;1-\varepsilon,
        \qquad\varepsilon:=\tfrac{\delta}{1-\delta}>0.
      \]
      Hence $\phi_{\textsc{EFR}}(\bm y,\bm g)\in[0,2\arccos(1-\varepsilon)]$
      for all $(\bm y,\bm g)\in\mathcal T\times\mathcal T$.
      Setting $a(\bm y)\equiv0$ gives
      \(\sup_{\bm g}\rho(\bm y,\bm g)<\infty\) and
      \(\mathbb{E}_{P_x}\sup_{\bm g}\rho(\bm Y,\bm g)<\infty\).

\item[\textbf{(A4)}] \textbf{Identification.}
      By Step 1 (injectivity) the map
      $\bm y\mapsto\phi_{\textsc{EFR}}(\bm y,\bm g)$
      is one-to-one for each $\bm g$, hence the expected loss
      \(
        \gamma(\bm g)=\mathbb{E}_{P_x}\rho(\bm Y,\bm g)
      \)
      has a unique minimiser $\bar f$ in $\mathcal T$.

\item[\textbf{(A5)}] \textbf{Compact parameter set.}
      $\mathcal T$ is compact, so any positive continuous
      $b(\bm g)$ e.g.\ $b(\bm g)\equiv1$ satisfies
      \(\inf_{\bm g\in\mathcal T}
          \tfrac{\rho(\bm y,\bm g)-a(\bm y)}{b(\bm g)}\ge0\)
      with $a(\bm y)=0$.
\end{itemize}

All Huber conditions hold; therefore
\[
\frac1n\sum_{i=1}^{n}\rho(\bm Y_i,\bar f_n)
\;-\;
\inf_{\bm g\in\mathcal T}\frac1n\sum_{i=1}^{n}\rho(\bm Y_i,\bm g)
\;\xrightarrow{n\to\infty}0,
\quad\text{and}\quad
\bar f_n\xrightarrow{P_x}\bar f.
\]
Since $\mathcal T$ is compact,
$\bar f_n\to\bar f$ also holds almost surely
(by the usual subsequence argument).
\end{proof}

This sufficiency result is interesting in application, it means that the uniqueness of empirical Fréchet mean is ensured in a narrower and narrower cube, as the number of sample points $T$ over the grid increases. In other words, if we have more and more sample points, there could be identifiability issue, warning us of the ``extra-flexibility'' brought by functional data. However, since our condition is only sufficient,
there could be a wider domain (e.g., removing the monotonicity induced by $\mathcal{D}$) where the uniqueness is ensured.

\begin{cor}
Let
\[
\lambda(\bm g)
   :=\mathbb{E}_{P_x}
        \bigl[\phi_{\textsc{EFR}}(\bm Y,\bm g)\bigr],
   \qquad\bm g\in\mathbb R^{T},
\]
and assume the Jacobian
$\frac{\partial\lambda}{\partial \bm{g}}\mid_{\bm{g}=\bar{f}}
  =:\bm\Lambda$
is non-singular.  Then
\[
\sqrt{n}\,\bigl(\,\overline f_{\,n}-\overline f\,\bigr)
   \;\;\xrightarrow{d}\;\;
   \mathcal N\!\bigl(
      \bm 0,\;
      \bm\Lambda^{-1}\bm C\,\bm\Lambda^{-\top}
   \bigr),
\]
where $\bm C=\operatorname{Cov}_{P_x}
        \bigl[\phi_{\textsc{EFR}}(\bm Y,\bar f)\bigr].$
\end{cor}

\begin{proof}
Since $\overline f_{\,n}$ is a Huber--type $\rho$-estimator
(Theorem~\ref{thm:IPW-interpolant-consistency}) and
$\phi_{\textsc{EFR}}$ is continuously differentiable on
$\mathcal T$, Theorem~6.6 of \citet{huber2011robust}
applies directly, yielding the stated asymptotic normality.
\end{proof}
\section{Properties of balancing weights}\label{sec:Balancing-weights}
Here we illustrate explicitly why under the \emph{back-door} criterion, the resulting \(\mathbb{E}_{Q}[\omega_x]=1\). For a fixed treatment arm \(x\in\{0,1\}\) define  
\[
\omega_x \;=\;
\frac{\mathbf{1}\{X=x\}}{\pi_x(\bm V)},
\qquad
\pi_x(\bm V) \;=\; P(X=x\mid\bm V).
\]

Take expectation under the observational law \(Q\) of \((X,\bm V,\bm Y)\):
\begin{equation}
\mathbb{E}_{Q}[\omega_x]
\;=\;
\mathbb{E}_{Q}
\!\Bigl[
  \tfrac{\mathbf 1\{X=x\}}{\pi_x(\bm V)}
\Bigr]
\;=\;
\mathbb{E}_{\bm V}
\Bigl[
  \mathbb{E}
  \bigl[
    \tfrac{\mathbf 1\{X=x\}}{\pi_x(\bm V)}
    \,\big|\,\bm V
  \bigr]
\Bigr].
\label{eq:expectation-balancing-weights-1}
\end{equation}
Conditional on \(\bm V=\bm v\),
\[
\mathbb{E}\bigl[\mathbf 1\{X=x\}\mid\bm V=\bm v\bigr]
  \;=\;
  P(X=x\mid\bm V=\bm v)
  \;=\;
  \pi_x(\bm v).
\]
Hence
\begin{equation}
\mathbb{E}
\Bigl[
  \tfrac{\mathbf 1\{X=x\}}{\pi_x(\bm V)}
  \,\Big|\,\bm V=\bm v
\Bigr]
\;=\;
\frac{\pi_x(\bm v)}{\pi_x(\bm v)}
\;=\;1.
\label{eq:expectation-balancing-weights-2}
\end{equation}

Substituting \eqref{eq:expectation-balancing-weights-2} into \eqref{eq:expectation-balancing-weights-1} gives
\[
\;
  \mathbb{E}_{Q}[\omega_x]=\mathbb{E}_{\bm V}[1]=1.
\;
\]

Thus the IPW weight integrates to one, confirming
\(\omega_x\) is a valid likelihood ratio that re-weights the
observational law \(Q\) to the interventional law for \(X=x\).
\section{Asymptotic normality of dynamic effects}
\subsection{\label{subsec:Asymptotic-normality-L2} Asymptotic normality of residuals in \texorpdfstring{$L_{2}$}{L2}}

In this section, we establish the asymptotic normality of the residuals of our
pointwise and norm-based estimators of the treatment effect. We first directly prove Theorem~\ref{thm:l2_dte_asymp_normality} for the simpler finite $T$ case, we extend the argument as $T\to\infty$ following a functional Central Limit Theorems (fCLT) argument and point to the additional assumptions that have to be made.

\paragraph{Proof of Theorem~\ref{thm:l2_dte_asymp_normality} for finite $T$}

\begin{proof}
 Each entry \(\hat{\Delta}(t)\) is an average (or
difference of two averages) of i.i.d.\ observations, so the multivariate Central Limit Theorem ensures
\[
  \sqrt{n}\,\bigl(\hat{\Delta} - \Delta\bigr)
  \;\;\xrightarrow{d}\;
  \mathcal{N}\Bigl(\mathbf{0},\mathbf{K}\Bigr).
\]

When \(\|\Delta\|\neq 0\), the map \(g(x)=\|x\|_{2}\) is differentiable around
\(\Delta\), and the usual delta method yields the normal limit with variance
\(\Delta^\top\,\mathbf{K}\,\Delta \,/\,\|\Delta\|_{2}^2\).  

If \(\Delta=0\), then \(g\) is not differentiable at \(0\), but by the continuous
mapping theorem we immediately get \(\sqrt{n}\,\|\hat{\Delta}\|_{2}\xrightarrow{d} \|\mathcal{Z}\|\) where \(\|\mathcal{Z}\|\) is the Euclidean norm of a mean-zero Gaussian vector with covariance \(\mathbf{K}\). Its square \(\|\mathcal{Z}\|_2^2\) follows a generalized \(\chi^2\)-distribution (the distribution of
\(\sum_{j=1}^T \lambda_j \,\nu_j^2\) for \(\nu_j \sim \mathcal{N}(0,1)\), \(\lambda_j\)
the eigenvalues of \(\mathbf{K}\)).
\end{proof}

\paragraph{Infinite-dimensional extension (\texorpdfstring{$T\to\infty$}{T->infinity}).}

Consider $\mathcal{F}=L^{2}([0,1])$ with $\phi(f,g)=\left\langle f-g,f-g\right\rangle $,
each observation $\bm{Y}_{i}$ is associated with a treatment variable $X_{i}$ and covariates $\bm{V}_{i}$. In this setting, the Fréchet mean corresponds to the pointwise, cross-sectional mean of functions, which is unique due to the strict convexity of the squared $L_{2}$ norm.
\begin{proof}
For the infinite dimensional setting, a similar argument for asymptotic normality of the estimator residuals can be made using some recent results from \cite{kennedy2023semiparametric} demonstrating that under standard \textit{identifiability} assumptions outlined above, the defined functional potential outcomes $\bm{Y}^{(1)}$ and $\bm{Y}^{(0)}$ belong to the Donsker class (a formal proof is provided as a special case of \citep[Theorem 3]{kennedy2023semiparametric}). If, instead of a fixed dimension \(T\), the potential outcomes 
\(\bm{Y}^{(1)}\) and \(\bm{Y}^{(0)}\) are functions in a suitable
infinite-dimensional space (e.g.\ \(L_2([0,1])\)), one can invoke a
\emph{functional} CLT. Under conditions ensuring that the empirical processes
\(\hat{F}(\bm{Y}^{(x)})\) converge in distribution in some function space (often
requiring \(\bm{Y}^{(x)}\) to lie in a Donsker class, see
\cite{van2000asymptotic,kennedy2023semiparametric}), one obtains
\[
  \sqrt{n} \Bigl[\hat{F}\bigl(\bm{Y}^{(1)}\bigr) - F\bigl(\bm{Y}^{(1)}\bigr)\Bigr]
  \;\xrightarrow{d}\;\mathcal{GP}^{(1)},
  \qquad
  \sqrt{n} \Bigl[\hat{F}\bigl(\bm{Y}^{(0)}\bigr) - F\bigl(\bm{Y}^{(0)}\bigr)\Bigr]
  \;\xrightarrow{d}\;\mathcal{GP}^{(0)},
\]
where \(\mathcal{GP}^{(1)}, \mathcal{GP}^{(0)}\) are mean-zero Gaussian \emph{processes}.
Hence their difference is also asymptotically Gaussian in the relevant function
space:
\[
  \sqrt{n}\,\Bigl(\hat{\Delta}(\cdot) - \Delta(\cdot)\Bigr)
  \;\;=\;\;
  \sqrt{n}\Bigl[\hat{F}\bigl(\bm{Y}^{(1)}\bigr)-F\bigl(\bm{Y}^{(1)}\bigr)\Bigr]
  \;-\;
  \sqrt{n}\Bigl[\hat{F}\bigl(\bm{Y}^{(0)}\bigr)-F\bigl(\bm{Y}^{(0)}\bigr)\Bigr]
  \;\xrightarrow{d}\;\mathcal{GP}^{(1)} \,-\, \mathcal{GP}^{(0)}.
\]
Applying the \emph{functional} delta method to the map
\(\Delta(\cdot)\mapsto \|\Delta(\cdot)\|_{L^2}\) requires Fréchet
differentiability at non-zero \(\Delta\).  The derivative of
\(\|\cdot\|_{L^2}\) at \(\Delta\) is
\[
  h'_{\Delta}\bigl(g\bigr)
  \;=\; \frac{\langle g,\Delta\rangle_{L^{2}}}{\|\Delta\|_{L^{2}}}
  \quad\text{for}\quad \|\Delta\|\neq 0.
\]
From this, one obtains precisely the same asymptotic normality result for
\(\sqrt{n}\,\bigl(\|\hat{\Delta}\|_{L^2}-\|\Delta\|_{L^2}\bigr)\), but with
\(\mathcal{GP}^{(x)}\) now interpreted as infinite-dimensional Gaussian processes
(and their inner products with \(\Delta\) giving the limiting variance). For the
zero-function case \(\|\Delta\|_{L^2}=0\), one again obtains convergence in
distribution to the norm of a mean-zero Gaussian process, i.e.\ \(\|\mathcal{Z}\|_{L_2}\)
for \(\mathcal{Z}\) in a suitable function space; see
\cite[Theorem 3.9]{testa2025doubly} for a more detailed study of the asymptotic normality of the residuals of $\hat{\Delta}$ in the functional setting. For the infinite-dimensional scenario, the primary additional requirement
is that the \emph{class} of (potential) outcome functions
\(\bm{Y}^{(1)}(\cdot)\) and \(\bm{Y}^{(0)}(\cdot)\) be Donsker (ensuring a
functional CLT), and that the map \(\Delta(\cdot)\mapsto \|\Delta(\cdot)\|_{L^2}\)
be sufficiently regular (Fréchet differentiable away from zero). Under these
assumptions, all the conclusions of
Theorem~\ref{thm:l2_dte_asymp_normality} extend naturally to \(T\to\infty\).
\end{proof}
\subsection{Constructing operator kernel using SRSF embedding\label{subsec:proof-operator-kernel}}

Although one can formally define a Riemannian metric on $\mathcal{F}$  (where $\mathcal{F}$ is a suitable set of curves), a key result from 
\citet{srivastava2016functional} shows that the geodesic distance under the Fisher--Rao metric between any two curves $f,g\in\mathcal{F}$  can be computed directly via their SRSFs. Concretely, we claim:
\[
    d_{\mathrm{FR}}\bigl(f,g\bigr)
    \;=\;
    \|\,q_{f} - q_{g}\|_{L^2}.
\]
To see this, construct a path 
$\{f_\tau\}_{\tau\in[0,1]}\subset\mathcal{F}$ corresponding to a straight line in SRSF space:

\begin{enumerate}
\item Let $q_f(t)$, $q_g(t)$ be the SRSFs of 
$f$ and $g$. A straight line in $L^2$ between $q_f$ and $q_g$ is 
\[
    q_\tau(t) \;=\; (1-\tau)\,q_f(t) \;+\;\tau\,q_g(t).
\]
\item  Define
\[
    f_\tau(t) \;=\; \int_{0}^{t} q_\tau(s)^{2}\,ds.
\]
Then $f_\tau(0)=0$ and $\dot{f}_\tau(s)=q_\tau(s)^2$, so $f_\tau\in\mathcal{F}$ 
for each $\tau$.
\item The Fisher--Rao length of $\tau\mapsto f_\tau$ 
can be shown to equal
\[
    \int_{0}^{1}
      \left\|
         \frac{d}{d\tau} f_\tau
      \right\|_{\mathrm{FR}} d\tau
    \;=\;
    \|\,q_f - q_g\|_{L^2},
\]
matching the length of a straight line in $L^2$.
\item  One shows no other path in $\mathcal{F}$ yields a smaller Fisher--Rao length. Thus $\|\,q_f - q_g\|_{L^2}$ is indeed the geodesic 
distance between $f$ and $g$.
\end{enumerate}
Hence,
\[
  d_{\mathrm{FR}}(f,g)
  \;=\;
  \|\,q_{f} - q_{g}\|_{L^{2}},
  \quad
  \forall\,f,g\in\mathcal{F}.
\]

\smallskip
\noindent
Next we demonstrate the isometric embedding. Define 
\[
    \Phi:\mathcal{F}\;\to\; L^2([0,1]), 
    \quad
    \Phi(f) \;=\; q_f.
\]
By the above claim,
\[
  d_{\mathrm{FR}}(f,g)
  \;=\;
  \|\Phi(f)-\Phi(g)\|_{L^2}.
\]
Hence $\Phi$ is an \emph{isometric embedding} of $\mathcal{F}$, endowed 
with the Fisher--Rao distance, into the Hilbert space $L^2([0,1])$. 
It is injective provided we fix $f(0)=0$ so that no two distinct curves 
have the same SRVF. Finally we utilize that Gaussian (RBF) Kernel is positive definite. Recall the well-known fact: if $\mathcal{H}$ is a Hilbert space (e.g.\ $L^2([0,1])$), 
then
\[
   \psi(h,h') 
   \;=\; 
   \exp\!\bigl(-\alpha\,\|h - h'\|^2\bigr),
   \quad \alpha>0,
\]
is a positive-definite kernel on $\mathcal{H}$. Since 
$\Phi:\mathcal{F}\to L^2([0,1])$ is an isometry, define
\[
   k_{\mathrm{FR}}(f,g)
   \;=\;
   \exp\!\Bigl(
     -\,\alpha\,d_{\mathrm{FR}}(f,g)^2
   \Bigr)
   \;=\;
   \exp\!\Bigl(
     -\,\alpha\,\|\Phi(f)-\Phi(g)\|_{L^2}^2
   \Bigr).
\]
For any finite set $\{f_1,\dots,f_n\}\subset\mathcal{F}$ and real coefficients 
$\{c_1,\dots,c_n\}$, we have
\[
  \sum_{i,j=1}^{n} 
  c_i\,c_j\,
  k_{\mathrm{FR}}(f_i,f_j)
  \;=\;
  \sum_{i,j=1}^{n}
  c_i\,c_j\,
  \exp\!\Bigl(
    -\,\alpha\,\|\Phi(f_i)-\Phi(f_j)\|_{L^2}^2
  \Bigr).
\]
Because $\{\Phi(f_i)\}\subset L^2([0,1])$, the standard RBF kernel in $L^2$ is positive-definite, so the above sum is nonnegative. Therefore, 
$k_{\mathrm{FR}}(\cdot,\cdot)$ is positive definite on $\mathcal{F}$. Thus
\[
   k_{\mathrm{FR}}(f,g)
   \;=\;
   \exp\!\bigl(
     -\,\alpha\,d_{\mathrm{FR}}(f,g)^2
   \bigr)
\]
is a valid positive-definite kernel on $\mathcal{F}$.

\section{Digital outcomes for Parkinson's disease}

\label{sec:digital-outcomes}

\subsection{Features for symptom detection}

Within stationary segments, 28 different features are extracted from
each axis of the pre-processed accelerometer data, based on non-overlapping
5 second windows: the standard deviation; the power in different frequency
bands (0.3-2Hz, 4-8Hz, 8-12Hz, 0.2-14Hz); the frequency and height
of the dominant peak in the PSD within different frequency bands (0.3-2Hz,
4-8Hz, 8-12Hz, 0.2-14Hz); the sample entropy, and the spectral entropy;
13 mel cepstral coefficients. We then perform z-score normalization
of the feature vectors, using the data from the unscripted activities
of all 24 PD patients (both with and without annotated tremor episodes)
and all 24 non-PD controls. These features are then used to train a logistic classifier (i.e.,  applying $l_1$-regularization) to predict the annotated gait and tremor episodes in the home visits segment of the Parkinson@Home cohort. The performance measures for both classifiers are computed using leave-one-subject-out cross validation and are reported in Table \ref{tab:detection_metrics}.

\subsection{Annotation protocol}

If the patient performed significant upper limb activities for more
than 3 seconds, the assistant only annotated whether tremor was present
or not. Otherwise, both the presence and severity of the tremor were
annotated, similar to the MDS-UPDRS part III tremor items. However,
because of the low prevalence of moderate and severe tremor in this
dataset, we aimed to model only the presence of tremor.

\begin{table}
\caption{Demographic and clinical characteristics of PD patients and non-PD
controls included in the analyses. IQR: inter-quartile range. MDS-UPDRS:
Movement Disorder Society-Sponsored Revision of the Unified Parkinson's
Disease Rating Scale. Part 1: non-motor experiences of daily living.
Part 2: motor experiences of daily living. Part 3: motor examination.
Part 4: motor complications. {*}: 1 missing value. \label{tab:Demographics}}
\centering

\begin{tabular}{ccccc}
\hline 
 & %
\begin{tabular}{@{}c}
\textbf{PD patients with}\tabularnewline
\textbf{annotated tremor $(n=8)$}\tabularnewline
\end{tabular}\tabularnewline
\hline 
 & 61.0 (58.3 - 69.0)\tabularnewline
\hline 
\multicolumn{2}{l}{Gender (men), n (\%)}\tabularnewline
\hline 
\multicolumn{2}{l}{Time since diagnosis of PD (years), median (IQR)}\tabularnewline
\hline 
\multicolumn{2}{l}{Hoehn and Yahr stage in off state, n (\%)}\tabularnewline
\multicolumn{2}{l}{\quad{}Stage 1}\tabularnewline
\multicolumn{2}{l}{\quad{}Stage 2}\tabularnewline
\multicolumn{2}{l}{\quad{}Stage 3}\tabularnewline
\multicolumn{2}{l}{\quad{}Stage 4}\tabularnewline
\hline 
\multicolumn{2}{l}{MDS-UPDRS, median (IQR)}\tabularnewline
\multicolumn{1}{l}{\quad{}Part 1} & \multicolumn{1}{l}{(scale range: 0 to 52)}\tabularnewline
\multicolumn{1}{l}{\quad{}Part 2} & \multicolumn{1}{l}{(scale range: 0 to 52)}\tabularnewline
\multicolumn{1}{l}{\quad{}Part 3 (off state)} & \multicolumn{1}{l}{(scale range: 0 to 132)}\tabularnewline
\multicolumn{1}{l}{\quad{}Part 3 (on state)} & \multicolumn{1}{l}{(scale range: 0 to 132)}\tabularnewline
\multicolumn{1}{l}{\quad{}Part 4} & \multicolumn{1}{l}{(scale range: 0 to 24)}\tabularnewline
\hline 
\multicolumn{2}{l}{Tremor sub-score of MDS-UPDRS part III, median (IQR)}\tabularnewline
\multicolumn{1}{l}{\quad{}Off state} & \multicolumn{1}{l}{(scale range: 0 to 40)}\tabularnewline
\multicolumn{1}{l}{\quad{}On state} & \multicolumn{1}{l}{(scale range: 0 to 40)}\tabularnewline
\hline 
\multicolumn{2}{l}{Rest tremor severity (arm of most affected side), n (\%)}\tabularnewline
\multicolumn{2}{l}{\quad{}0: normal (off $|$ on)}\tabularnewline
\multicolumn{2}{l}{\quad{}1: slight (off $|$ on)}\tabularnewline
\multicolumn{2}{l}{\quad{}2: mild (off $|$ on)}\tabularnewline
\multicolumn{2}{l}{\quad{}3: moderate (off $|$ on)}\tabularnewline
\multicolumn{2}{l}{\quad{}4: severe (off $|$ on)}\tabularnewline
\hline 
\end{tabular}
\end{table}

\begin{table}[h]
    \centering
    \renewcommand{\arraystretch}{1.2}
    \setlength{\tabcolsep}{10pt}
    \definecolor{headerblue}{RGB}{0, 153, 204}
    \begin{tabular}{lccc} 
        \textbf{Event detection} & \textbf{{AUROC}} & \textbf{Sensitivity} & \textbf{Specificity} \\
        \midrule
        Tremor detection & 0.89 (0.07) & 0.74 (0.16) & 0.95 (0.07) \\
        Gait detection   & 0.95 (0.03) & 0.73 (0.17) & 0.95 (0.04) \\
        \bottomrule
    \end{tabular}
    \caption{Performance metrics for tremor and gait detection algorithms trained and evaluated from the Parkinson@Home annotated visits. We use leave-one-subject-out cross-validation; mean measures across folds are reported with standard deviation in the brackets.}
    \label{tab:detection_metrics}
\end{table}
\subsection{Conditional Treatment Effects}
To provide context on the digital monitoring application discussed in Section \ref{subsec:Causal-effects-on}, we estimate the conditional ATE using the standard IPW estimator \cite{imai2004causal}, explicitly conditioning on the time covariate. Figure \ref{fig:Conditional-on-time-effects} presents the mean and standard deviation of the potential outcomes for both the treated and non-treated groups. We observe that the more complex causal estimators yield a similar dynamic ATE for the effect of levodopa therapy on gait energy. However, while the IPW conditional ATE estimator obscures the significance of the therapy's effect on tremor episodes, the purpose-designed functional causal kernel ATE estimator better highlights this effect.

\begin{figure}
\centering

\includegraphics[width=0.4\textwidth]{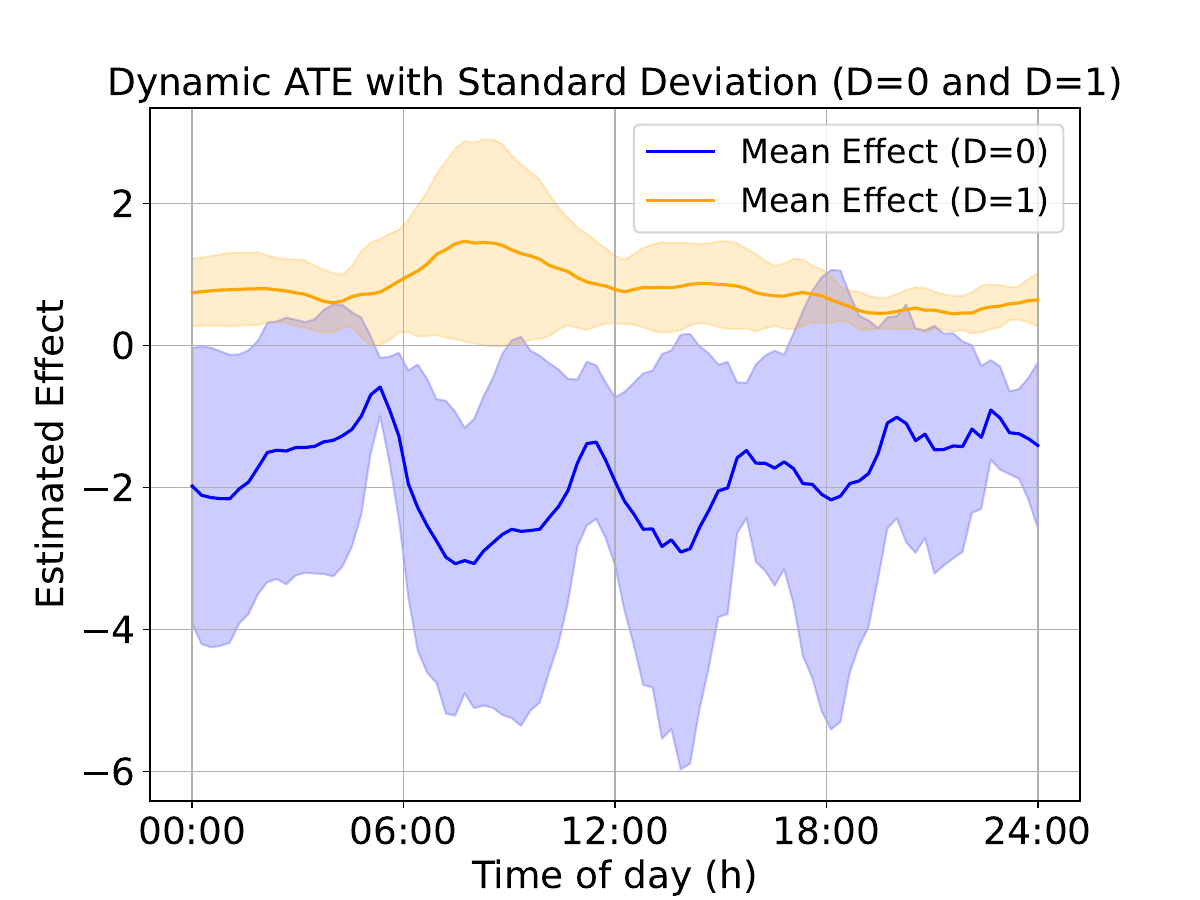}
\includegraphics[width=0.41\textwidth]{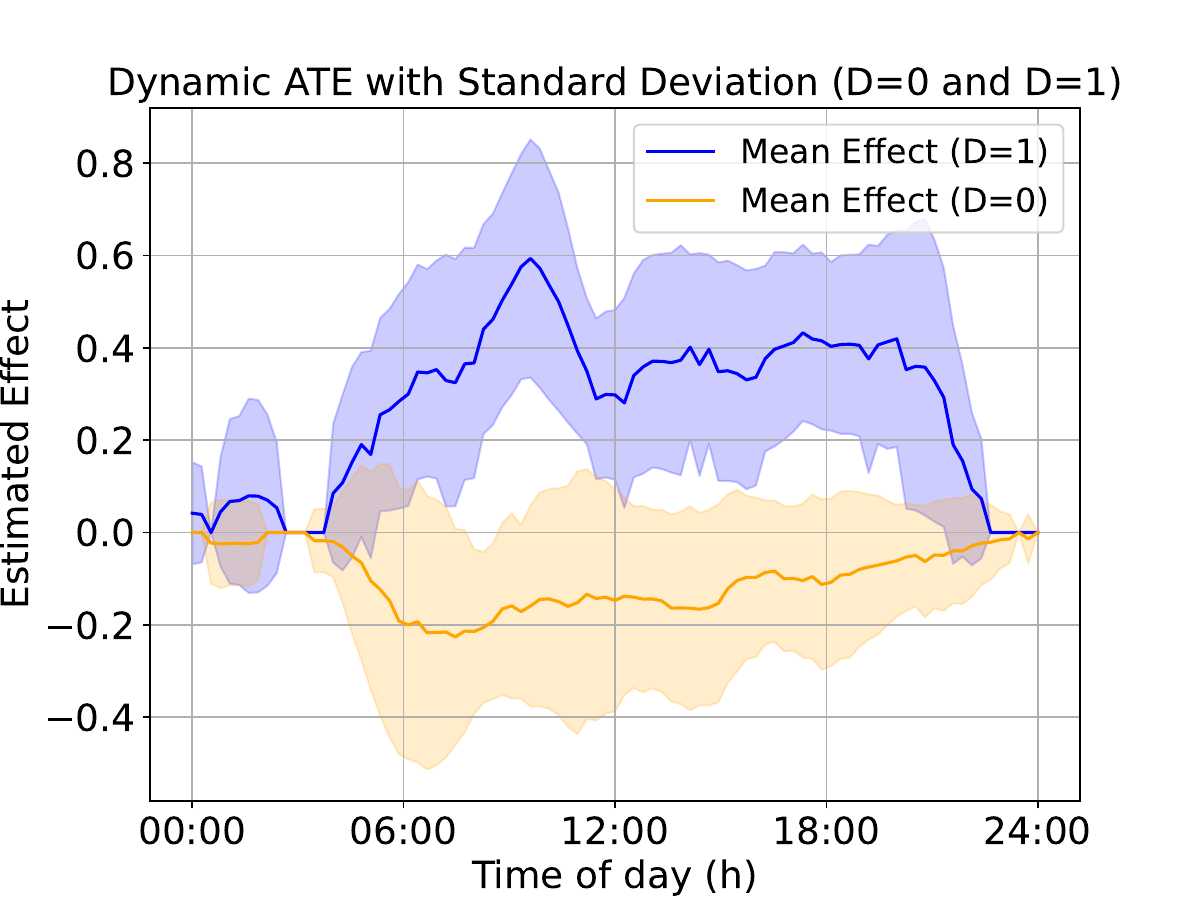}
\caption{\label{fig:Conditional-on-time-effects} Average effect of disease
category on digital outcomes (left) gait, (right) tremor. IPW estimator
is used to condition explicitly on time and display the $\varphi^{cATE}$
estimator on the time grid. }
\end{figure}

\end{document}

%% file: main_alg.tex
\begin{algorithmic}[1] \Require Functional covariates $\{V_{i}(u)\}_{i=1}^{n}$,
outcomes $\{Y_{i}(v)\}_{i=1}^{n}$, treatment $\{X_{i}\}_{i=1}^{n}$, metric $\phi$, kernel functions $k_{\mathcal{X}}$, $k_{\mathcal{V}}$, $k_{\mathcal{Y}}$, corresponding RKHS maps $\psi_{\mathcal{X}}$, $\psi_{\mathcal{V}}$ and $\psi_\mathcal{Y}$
regularization parameter $\lambda>0$, maximum iterations $R_{\text{max}}$,
convergence threshold $\epsilon$. \Ensure Registered curves $\tilde{V}_{i}(u)$,
$\tilde{Y}_{i}(\nu)$; estimated conditional expectation $m(x,v)$;
causal effects $\varphi^{dATE}$ (or $\varphi^{dDS}$ in the last step).
\State Initialize $\tilde{V}_{i}^{(0)}(u)=V_{i}(u)$, $\tilde{Y}_{i}^{(0)}(\nu)=Y_{i}(\nu)$,
and set $r=0$. \While{$r<R_{\text{max}}$ \textbf{or} convergence
criterion not met} 
 \State Compute the mean curve $\mu_{V}$ for functional covariates:
\[
\mu_{V}^{(r)}\gets\arg\min_{\mu}\sum_{i=1}^{n}\phi(\tilde{V}_{i}^{(r)},\mu_{V})
\]
\State Compute the mean curve $\mu_{Y}$ for functional outcomes:
\[
\mu_{Y}^{(r)}\gets\arg\min_{\mu}\sum_{i=1}^{n}\phi(\tilde{Y}_{i}^{(r)},\mu_{Y})
\]

\For{$i=1,\dots,n$} \State Align covariates: 
\[
\tilde{V}_{i}^{(r+1)}(u)\gets\tilde{V}_{i}^{(r)}(u)\circ\arg\min_{\gamma\in\Gamma}\phi(\tilde{V}_{i}^{(r)},\mu_{V}^{(r)}\circ\gamma)
\]
\State Align outcomes: 
\[
\tilde{Y}_{i}^{(r+1)}(\nu)\gets\tilde{Y}_{i}^{(r)}(\nu)\circ\arg\min_{\gamma\in\Gamma}\phi(\tilde{Y}_{i}^{(r)},\mu_{Y}^{(r)}\circ\gamma)
\]
\EndFor

\State Map registered curves into RKHS: $\psi_{\mathcal{X}}(X_{i})\in \mathcal{H}_{\mathcal{X}}$,
$\psi_{\mathcal{Y}}(\tilde{Y}_{i}^{(r+1)})\in \mathcal{H}_{\mathcal{Y}}$, $\psi_{\mathcal{V}}(\tilde{V}_{i}^{(t+1)})\in \mathcal{H}_{\mathcal{V}}$.
\State Compute the kernel matrix $\mathcal{K}^{(r+1)}$ with its $(i,j)$-th
entry as: 
\[
\mathcal{K}_{ij}^{(t+1)}\gets k_{\mathcal{X}}(\psi_{\mathcal{X}}(X_{i}),\psi_{\mathcal{X}}(X_{j}))\cdot k_{\mathcal{V}}(\psi_{\mathcal{V}}(\tilde{V}_{i}^{(r+1)}),\psi_{\mathcal{V}}(\tilde{V}_{j}^{(r+1)}))
\]
\State Solve kernel ridge regression: 
\[
\alpha^{(r+1)}\gets(\mathcal{K}^{(r+1)}+\lambda I)^{-1}\tilde{Y}^{(r+1)}
\]
\State For a new input $(x,v)$, compute: 
\[
\hat{m}^{(r+1)}(x,v)\gets\sum_{i=1}^{n}k_{\mathcal{X}}(\psi_{\mathcal{X}}(x),\psi_{\mathcal{X}}(X_{i}))\cdot k_{\mathcal{V}}(\psi_{\mathcal{V}}(v),\psi_{\mathcal{V}}(\tilde{V}_{i}^{(r+1)}))\cdot\alpha_{i}^{(r+1)}
\]

\If{$\|\hat{m}^{(r+1)}-\hat{m}^{(r)}\|_{2}<\epsilon$, where the
norm is evaluated on the training set of size $n$} \State \textbf{Break.}
\EndIf 
\EndWhile

\State Compute dynamic average treatment effect (dATE): 
\[
\varphi^{dATE}\gets\frac{1}{m}\sum_{j=1}^{m}\big(\hat{m}^{(r+1)}(1,V_{j})-\hat{m}^{(r+1)}(0,V_{j})\big)
\]


\end{algorithmic}